\documentclass[%
	submission,
	]{IEEEtran}
\usepackage{epsfig,amssymb}
\usepackage[cmex10]{amsmath}
\usepackage{amssymb}
\usepackage{graphicx,times}
\usepackage{lipsum}
\usepackage{tabularx,booktabs}
\usepackage{cite}
\newtheorem{thm}{Lemma}

\usepackage[top=1in, bottom=1in, left=0.5in, right=0.5in]{geometry}\usepackage{array} 
\usepackage{booktabs} 
\usepackage{multirow} 
\usepackage{tabularx}
\usepackage{blindtext}
\usepackage{setspace}
\usepackage{cite}
\usepackage{dsfont}

\usepackage{hyperref}
\newcommand*\rot{\rotatebox{90}}
\begin{document}

\title{PROSE: Perceptual Risk Optimization for Speech Enhancement}
\author{ Jishnu Sadasivan, Chandra Sekhar Seelamantula,~\IEEEmembership{\it Senior Member,~IEEE}, and Nagarjuna Reddy Muraka
\thanks{J. Sadasivan is with the Department of Electrical Communication Engineering, Indian Institute of Science (IISc.), Bangalore, India; E-mail: sadasivan@iisc.ac.in. N. R. Muraka is currently with the Indian Railways; Email: nagarjunareddy100@gmail.com. C. S. Seelamantula is with the Department of Electrical Engineering, IISc., Bangalore-560012; Email: chandra.sekhar@ieee.org. Phone: +91 80 2293 2695. Fax: +91 80 2360 0444. 
The figures in this manuscript are in color in the electronic version.}
\markboth{IEEE Transactions on Audio, Speech, and Language Processing}}
\maketitle
 \begin{abstract}
The goal in speech enhancement is to obtain an estimate of clean speech starting from the noisy signal by minimizing a chosen distortion measure, which results in an estimate that depends on the unknown clean signal or its statistics. Since access to such prior knowledge is limited or not possible in practice, one has to estimate the clean signal statistics.  In this paper, we develop a new risk minimization framework for speech enhancement, in which, one optimizes an unbiased estimate of the distortion/risk instead of the actual risk. The estimated risk is expressed solely as a function of the noisy observations. We consider several perceptually relevant distortion measures and develop corresponding unbiased estimates under realistic assumptions on the noise distribution and a priori signal-to-noise ratio (SNR). Minimizing the risk estimates gives rise to the corresponding denoisers, which are nonlinear functions of the a posteriori SNR. Perceptual evaluation of speech quality (PESQ), average segmental SNR (SSNR) computations, and listening tests show that the proposed risk optimization approach employing Itakura-Saito and weighted hyperbolic cosine distortions gives better performance than the other distortion measures. For SNRs greater than $5$ dB, the proposed approach gives superior denoising performance over the benchmark techniques based on the Wiener filter, log-MMSE minimization, and Bayesian nonnegative matrix factorization.
\end{abstract}
\section{Introduction}
\IEEEPARstart{T}{he} goal in speech enhancement is to suppress noise and enhance signal intelligibility and quality. Over the past few decades, several techniques have been developed for noise suppression. The challenges are  nonstationarity of the speech signal, distribution of noise, type of noise distortion, noise being signal-dependent or independent, etc. The problem continues to be of significant interest to the speech community particularly considering the enormous increase in the number of smartphone users. An early review of various noise reduction techniques was given by Lim and Oppenheim~\cite{Lim1}. Loizou's book on speech enhancement~\cite{PLoizou} is a recent and comprehensive reference on the topic. We shall briefly review related literature before proceeding with the development of the new risk optimization framework for speech enhancement.
\subsection{Related Literature}
\indent Speech enhancement techniques can be classified as follows.
\subsubsection{Spectral subtraction algorithms}
\indent Boll~\cite{Boll} and Weiss et al.~\cite{Weiss} proposed to subtract an estimate of the noise power spectrum from the noisy signal spectrum, in order to estimate the clean signal spectrum. The assumption is that the noise is additive and stationary. The noisy signal phase is used in reconstructing the time-domain signal. Weiss et al.~\cite{Weiss} also proposed subtraction techniques in autocorrelation and cepstral domains. Lockwood and Boudy~\cite{Lockwood}, Kamath and Loizou~\cite{Kamath} proposed improved versions of  spectral subtraction algorithms.
\subsubsection{Wiener filtering techniques} 
\indent These are based on the minimum mean-squared error (MMSE) criterion~\cite{PLoizou}, in which one constructs the Wiener filter using an estimate of the clean and noisy speech power spectra. Lim and Oppenheim~\cite{Lim1} proposed a parametric Wiener filter, which allows for controlling the trade-off between the signal distortion and residual noise. Hu and Loizou incorporated psychoacoustic constraints into this framework \cite{Hu1, Hu2}.  In~\cite{Hu1}, they use a perceptual weighting filter to shape the residual noise to make it inaudible. In~\cite{Hu2}, they constrain the noise spectrum to lie below a preset threshold at each frequency. Chen et al.~\cite{Chen} quantified the amount of noise reduction and analyzed its relation to speech distortion. The Wiener filter  requires an estimate of the a priori signal-to-noise ratio (SNR). Scalart and Filho ~\cite{Scalart} used a recursive a priori SNR estimator whereas Lim and Oppenheim~\cite{Lim2} iteratively estimated the Wiener filter based on autoregressive modeling of the speech signal. Hansen and Clements~\cite{Hansen} imposed inter- and intra-frame constraints to ensure speech-like characteristics within each iteration. Sreenivas and Kirnapure~\cite{TVS} proposed a codebook constrained, iterative Wiener filter with superior convergence behavior. Srinivasan et al.~\cite{Srinivasan1, Srinivasan2} proposed maximum-likelihood and Bayesian methods for estimating the speech and noise power spectra. Rosenkranz and Puder~\cite{Tobias} proposed adaptation techniques to improve the performance of the codebook approaches reported in~\cite{Srinivasan1, Srinivasan2} against model mismatches and unknown noise types.
\subsubsection{Subspace techniques}
\indent Originally proposed by Ephraim and Van Trees, subspace techniques rely on eigenvalue decomposition of the data covariance matrix~\cite{Vantrees,Mittal,Huang,Rezayee} or singular-value decomposition of the data matrix~\cite{Dendrinos,Hansen1,Hansen2}. The noise eigenvalues/singular values are smaller than those of the noisy signal and denoising happens when the signal is reconstructed from the eigen/singular vectors corresponding to the signal subspace alone. Jabloun and Champagne~\cite{Jabloun} incorporated properties of human audition into the signal subspace approaches.
\subsubsection{Statistical model based methods}
\indent McAulay and Malpass~\cite{McAulay} proposed a maximum-likelihood (ML) estimator of the clean speech short-time Fourier transform (STFT) magnitude. The clean speech spectra are assumed to be deterministic and the noise is modeled as zero-mean complex Gaussian. The ML estimate of the magnitude spectrum is combined with the noisy phase spectrum in order to reconstruct the speech signal. Ephraim and Malah~\cite{Ephraim-Malah1} proposed a Bayesian MMSE estimator of the short-time spectral amplitude (STSA) by assuming speech and noise to be statistically independent, zero-mean, complex Gaussian random variables. McCallum and Guillemin~\cite{Mathew} proposed an MMSE-STSA estimator assuming a nonzero-mean speech signal. Ephraim~\cite{Ephraim} used hidden Markov model (HMMs) to model the dynamics of speech and noise processes. Erkelens et al.~\cite{Erkelens} proposed MMSE estimators of clean speech discrete Fourier transform (DFT) coefficients and DFT magnitudes  assuming generalized Gamma distributions on speech.  Kundu et al.~\cite{Kundu} developed an MMSE estimator by using a Gaussian mixture model (GMM) for the clean speech signal. Lotter and Vary~\cite{Lotter} proposed a maximum a posteriori (MAP) estimator assuming super-Gaussian statistics. Ephraim and Malah~\cite{Ephraim-Malah2} proposed an estimator that minimizes the MSE of log-magnitude spectra as it is perceptually more correlated. Loizou~\cite{Loizou2} computed Bayesian estimators for magnitude spectrum using perceptual distortion metrics such as the Itakura-Saito distortion, hyperbolic-cosine distortion, etc. Mohammadiha et al.~\cite{Mohammadiha} use a  Bayesian non-negative matrix factorization (NMF) approach to obtain an MMSE estimate of the the clean speech DFT magnitude.
\subsection{Our Contributions}
 \indent We introduce the notion of risk estimation for speech denoising. The clean speech signal is considered to be deterministic and the random noise to be additive (Section~\ref{sec:formulation}). Direct minimization of the risk results in estimates that are a function of the deterministic clean speech signal. Considering a transform-domain Gaussian observation model, one can develop an unbiased estimate of the MSE based on Stein's lemma~\cite{Stein}, which is referred to as Stein's unbiased risk estimator (SURE) in the literature~\cite{Blu1,Blu2}. The main advantage is that, unlike MSE, SURE does not require knowledge of the unknown deterministic clean signal. The state-of-the-art image denoising techniques are based on risk minimization~\cite{Blu1,Blu2} considering the MSE. In this paper, we solve the speech denoising problem within the framework of unbiased risk estimation, where we derive unbiased estimates of speech-specific perceptual distortion measures and minimize them to obtain the corresponding shrinkage functions. Distortion measures such as Itakura-Saito, hyperbolic-cosine (cosh), weighted cosh, etc. are considered as they are more perceptually relevant than MSE~\cite{Gray}.\\
\indent Further, in practice, real-world disturbances generate bounded noise amplitudes and quantization limits the dynamic range.
Therefore, we consider the more realistic case of a truncated Gaussian distribution for the samples. The details will be described in Section~\ref{sec:formulation}. In order to develop a risk estimator, we make use of Stein's lemma and its higher-order generalization, originally proposed for Gaussian noise (Section~\ref{sec:sure}). The higher-order generalization becomes important in the context of perceptual distortion measures. Correspondingly, we develop the notion of {\it perceptual risk} optimization for speech enhancement (PROSE) (Section~\ref{sec:prisk}). The key advantage of the PROSE framework is that it allows one to replace an ensemble-averaged distortion measure by a practically viable surrogate. We employ a transform-domain {\it point-wise shrinkage}, which is nonlinear in the observations. We also consider parametric versions, which give additional flexibility to trade-off between residual noise level and speech distortion. It turns out that perceptually optimized denoising functions result in more noise attenuation than MSE.  We also carry out objective assessment in terms of average segmental signal-to-noise ratio (SSNR), global SNR, perceptual evaluation of speech quality (PESQ)~\cite{pesq}, short-time objective intelligibility  (STOI)~\cite{stoi}, and subjective assessment by means of listening tests and scoring as per ITU-T recommendations~\cite{ITUscale} (Section~\ref{sec:exptres}).
\begin{figure*}[t]
\centering
\includegraphics[width=5.65in]{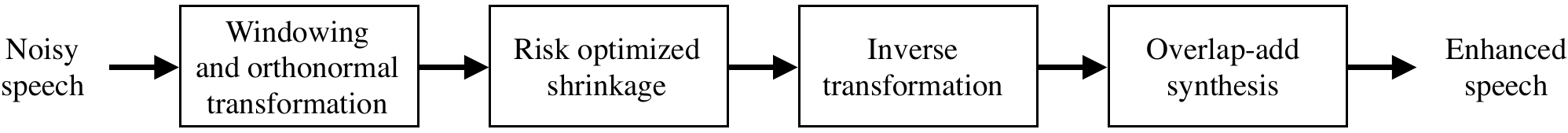}
\caption{A block diagram representation of the PROSE methodology.}
\label{fig:Block}
\end{figure*}

\section{Problem Formulation}
\label{sec:formulation}
\indent Consider a short-time frame of noisy speech in which samples of clean speech $s_n$ are distorted additively by noise $w_n$ resulting in the observations:
\begin{equation}
x_n = s_n + w_n\,, n = 1,2,\cdots,N,
\label{model}
\end{equation}
where $N$ is the frame length. The signal samples $\{s_n\}$ are assumed to be deterministic and noise samples $\{w_n\}$ to be zero-mean, bounded, i.i.d. random variables. Most real-world noise processes are bounded and the presence of a quantizer in a practical data acquisition scenario further justifies the assumption. Consequently, $\{x_n\}$ are bounded, and $\mathcal{E}\{x_n\} = s_n$, which implies that $\{x_n\}$ are independent, but not identically distributed.  The time-domain noise distribution is not restricted to be a Gaussian. Typical speech enhancement approaches work on the DFT magnitudes~\cite{PLoizou}, and the phase is left unaltered. We prefer the discrete cosine transform (DCT) as it is real-valued and known to give rise to a more parsimonious representation than the DFT \cite{Soon}. Further, Soon et al. established that, for shrinkage estimators, DCT-domain denoising is superior to DFT \cite{Soon}. Since our point-wise multiplicative shrinkage estimator belongs to this class, we prefer the DCT to DFT. The DCT representation of (\ref{model}) is 
\begin{equation}
X_k=S_k+W_k, \quad  k=1,2,\cdot\cdot\cdot,N,
\label{xk}
\end{equation}
where the transform-domain noise $\{W_k\}$ being a linear combination of i.i.d. random variables $\{w_n\}$ has a distribution that approaches a Gaussian by virtue of the {\it central limit theorem}. However, since $\{w_n\}$ are bounded, $\{W_k\}$ will also be bounded. These two properties taken together make a truncated Gaussian distribution model more appropriate and realistic for the transform coefficients than the standard Gaussian. The noisy samples are {\it concentrated} about the mean and the deviations from the mean are bounded. The suitability of the truncated Gaussian for modeling real-world processes has been advocated by Burkardt \cite{Burkardt}.\\
\indent The goal is to estimate $S_k$ given $X_k$ and noise statistics. Let $d(S_k,\widehat{S}_k)$ denote a distortion measure that quantifies the deviation of the estimate $\widehat{S}_k$ from $S_k$. 
The corresponding ensemble averaged distortion or {\it risk}, as referred to in the statistics literature, is defined as $\mathcal{R}=\mathcal{E}\left\{d(S_k, \widehat{S}_k)\right\},$ where ${\cal E}$ denotes the expectation operator. The estimate is expressed as ${\hat S_k} = f(X_k)$, where $f$ is the {\it denoising function}, which may not always be linear. We consider point-wise shrinkage, $f(X_k) = a_kX_k, a_k\in\left[0,1\right]$, and optimize an estimate of $\mathcal{R}$ with respect to $a_k$.\\
\indent A block diagram representation of the proposed method is shown in Figure~\ref{fig:Block}. To take into account the quasi-stationarity of speech, denoising is performed on a frame-by-frame basis and the enhanced speech is reconstructed using the standard overlap-add synthesis methodology.
\section{Risk Estimation Results}
\label{sec:sure}
\indent We recall a key result from~\cite{Stein}, which is central to the subsequent developments
\begin{thm} (Stein, 1981)
Let $W$ be a ${\cal N}(0,\sigma^2)$ real random variable and let $f: \mathbb{R}\rightarrow\mathbb{R}$ be an indefinite integral of the Lebesgue measurable function $f^{(1)}$, essentially the derivative of $f$. Suppose also that $\mathcal{E}\{|f^{(1)}(W)|\}<\infty$. Then, $\mathcal{E}\{Wf(W)\}=\sigma^2 \mathcal{E}\{f^{(1)}(W)\}.$
\end{thm}
\indent Stein's lemma facilitates estimation of the mean of $Wf(W)$ in terms of $f^{(1)}(W)$. Effectively, $\sigma^2\,f^{(1)}(W)$ could be used as an unbiased estimate of  $\mathcal{E}\{Wf(W)\}$. The implications of this apparently simple result can be appreciated when we are required to compute an unbiased estimator of the MSE. Next, we develop a higher-order generalization of Stein's lemma.
\begin{thm}
({\it Generalized Stein's lemma}) Let $W$ be a ${\cal N}(0,\sigma^2)$ real random variable and let $f: \mathbb{R}\rightarrow\mathbb{R}$ be an $n$-fold indefinite integral of the Lebesgue measurable function $f^{(n)}$, which is the $n^{th}$ derivative of $f$. Suppose also that $\mathcal{E}\left\{|W^{(n+1-k)}f^{(k)}(W)|\right\}<\infty, k = 1, 2, \cdots, n$. Then
$$\mathcal{E}\{W^{n+1}f(W)\}=\sigma^2 \mathcal{E}\{f^{(1)}(W)W^{n}\} + \sigma^2\,n\mathcal{E}\{f(W)W^{n-1}\}.$$
\label{gensteinlemma}
\end{thm}
\begin{proof}
Let us express $\mathcal{E}\{W^{n+1}f(W)\}$ as $\mathcal{E}\{Wg(W)\},$
where $g(W)=W^{n}f(W)$. Applying Stein's lemma to $\mathcal{E}\{Wg(W)\}$, we get that
\begin{align*}
\mathcal{E}\{W^{n+1}f(W)\}&=\mathcal{E}\{Wg(W)\} =\sigma^2\mathcal{E}\{g^{(1)}(W)\},  \\
&=\sigma^2 \mathcal{E}\{f^{(1)}(W)W^{n}\} + \sigma^2\,n\mathcal{E}\{f(W)W^{n-1}\}.
\end{align*}
\end{proof}
Stein's lemma could be applied recursively to each of the terms on the right-hand side, up to a stage where the terms comprise only the derivatives of all orders of $f$ up to $n$.\\
\indent Our next set of results is in the context of developing Stein-type lemmas for the practical case of a truncated Gaussian distribution. 
\begin{thm} 
Let $W$ be a  real random variable with probability density function (p.d.f.)
\begin{equation}
p\left(w; c, \sigma\right)=\frac{1}{\sqrt{2\pi}\sigma K} \exp\left(-\frac{w^2}{2 \sigma^2}\right) \mathds{1}_{\{w<|c\sigma|\}},
\label{truncated_Gaussian_pdf}
\end{equation}
where $K$ ensures that $p$ integrates to unity, $c \in \mathbb{R}^+$, and $\mathds{1}$ denotes the indicator function. Let $f: \mathbb{R}\rightarrow\mathbb{R}$ be an indefinite integral of the Lebesgue measurable function $f^{(1)}$. Suppose also that $\mathcal{E}\{|f^{(1)}\left(W\right)|\}<\infty$ and $f$ does not grow faster than an exponential, then $\mathcal{E}\{Wf\left(W\right)\} = \sigma^2 \mathcal{E}\{f^{(1)}\left(W\right)\} + \mathcal{O}\{\exp(-c^2)\}$.
\label{truncated_Gaussian_lemma1}
\end{thm} 
\begin{proof}
Using the property: $\displaystyle -\sigma^2 \frac{\mathrm{d} \exp\left(-\frac{w^2}{2 \sigma^2}\right)}{\mathrm{d} w} = w \exp\left(-\frac{w^2}{2 \sigma^2}\right)$, we write 
\begin{align}
\nonumber
\mathcal{E}\{Wf\left(W\right)\}&=\int_{-\infty}^{+\infty} w f(w)p\left(w;c,\sigma\right) \mathrm{d}w, \\
\nonumber
&=- \int_{-c\sigma}^{+c\sigma} \sigma^2 f(w) \frac{1}{\sqrt{2\pi}\sigma K} \frac{\mathrm{d} \exp\left(-\frac{w^2}{2 \sigma^2}\right)}{\mathrm{d} w} \mathrm{d}w,\\ 
\nonumber
&=- \sigma^2 f(w)p\left(w;c,\sigma\right)\Big|_{-c\sigma}^{+c\sigma}\\
\nonumber
& \hspace{2cm}+ \int_{-c\sigma}^{+c\sigma} \sigma^2\,f^{(1)}(w)p\left(w;c,\sigma\right)\mathrm{d}w\\
\nonumber
&= \mathcal{O}\{\exp(-c^2)\}+ \sigma^2 \mathcal{E} \left\{ f^{(1)}(W)\right\}.
\end{align}
\end{proof}
For even $f$, the approximation error is zero. Next, we state the counterpart of Lemma~\ref{gensteinlemma} for truncated Gaussian distribution, which has a similar proof mechanism.
\begin{thm}
Let $W$ be a  real random variable  with p.d.f $p\left(w; c, \sigma\right)$ and let $f: \mathbb{R}\rightarrow\mathbb{R}$ be an $n$-fold indefinite integral of the Lebesgue measurable function $f^{(n)}$, which is the $n^{th}$ derivative of $f$. Suppose also that $\mathcal{E}\left\{|W^{(n+1-k)}f^{(k)}(W)|\right\}<\infty, k = 1, 2, \cdots, n$ and  $f^{(k)}$ does not grow faster than an exponential, then
\begin{align*}
\mathcal{E}\{W^{n+1}f\left(W\right)\}=\sigma^2 \mathcal{E}\{f^{(1)}(W)W^{n}\} &+ \sigma^2\,n \mathcal{E}\{f\left(W\right)W^{n-1}\}\nonumber\\
&+ \mathcal{O}\{\exp(-c^2)\}.
\end{align*}
\label{truncated_Gaussian_recursive_lemma2}
\end{thm}
The approximation error is negligible for large values of $c$. These results will be handy in computing risk estimators for various perceptual distortion measures.

\section{Perceptual Risk Optimization for Speech Enhancement (PROSE)}
\label{sec:prisk}
\subsection{Mean-Square Error (MSE)}
\indent The squared error is the most commonly employed distortion measure largely because of ease of optimization. The distortion function for squared error in the transform domain is
\begin{equation}
d(S_k,\widehat{S}_k)=\left(\widehat{S}_k - S_k\right)^2,\,\,\,\mbox{where}\,\,\,\widehat{S}_k=f(X_k). \nonumber
\label{SE}
\end{equation}
The MSE is $\mathcal{R}={\cal E}\{d(S_k,\widehat{S}_k)\}$, which may be expanded as
\begin{equation}
\mathcal{R}=\mathcal{E}\left\{f^2(X_k)+ S_k^2 - 2f(X_k)X_k+2f(X_k)W_k\right\}. 
\label{avgrisk}
\end{equation}
Applying Lemma~\ref{truncated_Gaussian_lemma1} gives $\mathcal{E}\{f(X_k)W_k\} \approx \sigma^2 \mathcal{E}\{f^{(1)}(X_k)\}$, and from (\ref{avgrisk}), we get that
\begin{equation}
\mathcal{R} \approx \mathcal{E}\{f^2(X_k)-2f(X_k)X_k+2\sigma^2f^{(1)}(X_k)\}+S_k^2, \nonumber
\end{equation}
from which it can be concluded that
\begin{equation}
{\widehat {\cal R}} = f^2(X_k)-2f(X_k)X_k+2\sigma^2f^{(1)}(X_k) +S_k^2,
\label{msrisk}
\end{equation}
is a nearly unbiased estimator of the MSE. Although ${\widehat {\cal R}}$ contains the signal term $S_k^2$, it does not affect the minimization with respect to $f$. Consider the point-wise shrinkage estimator $f(X_k)=a_kX_k$, where $a_k\in\left[0,1\right]$. The optimum value of $a_k$ is obtained by  minimizing $\widehat {{\cal R}}$ subject to the constraint $a_k\in\left[0,1\right]$. The Karush-Kuhn-Tucker (KKT) conditions~\cite[pp.\ 211]{Fletcher} for solving this problem are given in Appendix~\ref{KKT_conditions} (cf. \eqref{lagrange_derivative} -- \eqref{second_order}). The optimum $a_k$ that satisfies the KKT conditions\footnotemark\footnotetext{The calculations related to the constrained optimization of all the perceptual risk estimates considered in this paper are provided in the supporting document.} is given by
$$a_k=\text{max}\left\{1-\displaystyle\frac{\sigma^2}{X_k^2},0\right\}.$$
The optimum shrinkage estimator becomes: 
$$\widehat{S}_k=\text{max}\left\{1-\displaystyle\frac{1}{\xi_k},0\right\}X_k =a_{\text{MSE}}(\xi_k)X_k,$$
where $\xi_k=\displaystyle \displaystyle\frac{X_k^2}{\sigma^2}$ denotes the a posteriori SNR determined based on the noisy signal, and $a_{\text{MSE}}(\xi_k)$ denotes the MSE-related shrinkage function. To impart additional flexibility, we consider parametric refinements, which allow us to trade-off between residual noise level and speech distortion. It is also useful when estimates of the noise variance may not be sufficiently accurate. The parametrically refined version is given by
\begin{equation}
\widehat{S}_k  = \text{max}\left\{1-\displaystyle\frac{\alpha}{\xi_k},0\right\}X_k,
\label{Parametric_MSEestimate}
\end{equation}
where $\alpha$ is the parameter, akin to the over-subtraction factor in spectral subtraction algorithms.
\subsection{Weighted Euclidean (WE) Distortion }
\indent The MSE is perceptually less relevant for speech signals since a large MSE does not always imply poor signal quality. Auditory masking effects render humans more robust to errors at spectral peaks than at the valleys and are taken advantage of in speech/audio compression. One way to introduce differential weighting to spectral peaks and valleys is to consider the weighted Euclidean (WE) distortion:
$$d(S_k,\widehat{S}_k)=\frac{\left(\widehat{S}_k - S_k\right)^2}{S_k}, \mbox{where}\, \widehat{S}_k=f(X_k).$$
\indent  The measure $d < 0$ when $S_k < 0$, but this is not a problem, because in that case, the cost function must be maximized and not minimized. We shall see that the proposed methodology implicitly takes care of this aspect. Developing $d$, we get
$$d(S_k,\widehat{S}_k)= \frac{\widehat{S}^2_k}{S_k} +S_k -2\widehat{S}_k.$$ Let us consider a high-SNR scenario, that is, $\left |\frac{W_k}{X_k}\right |<1$. This event occurs with probability 1 if $|S_k| > 2c\sigma$ (cf. Appendix \ref{probability_high_snr_event}). This condition also implies that $X_k$ and $S_k$ have the same sign. In this scenario, the first term is expanded as
$$\frac{\widehat{S}^2_k}{S_k} = \frac{\widehat{S}^2_k}{X_k}\left(1-\frac{W_k}{X_k}\right)^{-1} =\frac{\widehat{S}^2_k}{X_k}  \sum_{n=0}^{\infty} \left(\frac{W_k}{X_k} \right)^n,$$
and the distortion function is rewritten as
\begin{equation}
d(S_k,\widehat{S}_k)=\frac{\widehat{S}^2_k}{X_k}  \sum_{n=0}^{\infty} \left(\frac{W_k}{X_k} \right)^n+S_k -2\widehat{S}_k . \nonumber
\end{equation}
We truncate the infinite sum beyond the fourth-order, in order to maintain high accuracy in the calculations:
$$d(S_k,\widehat{S}_k) \approx \frac{\widehat{S}^2_k}{X_k} \sum_{n=0}^{4} \left(\frac{W_k}{X_k} \right)^n+S_k -2\widehat{S}_k.$$
Considering $f(X_k)=a_kX_k$, we seek to minimize the risk $\mathcal{R} $:
\begin{equation}
\mathcal{R} = \mathcal{E}\left\{a_k^2X_k  \sum_{n=0}^{4} \left(\frac{W_k}{X_k} \right)^n \right\}+S_k-2 \mathcal{E}\{a_k X_k\}.
\label{reciprocal} 
\end{equation}
To proceed further, we are required to compute expectations of reciprocals of truncated Gaussian random variables, which may not always be finite. A necessary and sufficient condition for the expectations to be finite is that $|S_k| > c\sigma$, which is satisfied in the high-SNR regime. This is an added benefit of working with more realistic distributions such as the truncated Gaussian.\\
\indent A simplification for the expectation of the sum appearing in the first term of \eqref{reciprocal} could be made by invoking \emph{Lemma} \ref{truncated_Gaussian_recursive_lemma2} (cf. Appendix~\ref{app:WE}). Consequently, we get
\begin{align*}
\mathcal{R}& =
a_k^2\mathcal{E} \left \{X_k+ \frac{\sigma^2}{X_k}-\frac{\sigma^4}{X_k^3}+48\frac{\sigma^6}{X_k^5}+360\frac{\sigma^8}{X_k^7} -2\frac{X_k}{a_k} \right \}+S_k.
\end{align*}
The corresponding unbiased risk estimator is obtained as
\begin{equation}
{\widehat {\cal R}}=a_k^2\Big(X_k+ \frac{\sigma^2}{X_k}-\frac{\sigma^4}{X_k^3}+48\frac{\sigma^6}{X_k^5}+360\frac{\sigma^8}{X_k^7} \Big)-2a_kX_k+S_k.
\label{WE_risk_estimate}
\end{equation}
The goal is to determine $a_k \in [0,1]$ such that ${\widehat {\cal R}}$ is minimized if $S_k>0$ and maximized if $S_k<0$. Solving the KKT conditions \eqref{lagrange_derivative}$-$\eqref{second_order} corresponding to these two scenarios results in the same optimum $a_k$.
The corresponding estimator is given by
$$\widehat{S}_k=\underbrace{\left(1+\displaystyle\frac{1}{\xi_k}-\displaystyle\frac{1}{\xi_k^2}+\displaystyle\frac{48}{\xi_k^3}+\displaystyle\frac{360}{\xi_k^4}\right)^{-1}}_{a_{\text{WE}}(\xi_k)}X_k,$$
where $a_{\text{WE}}(\xi_k)$ is the shrinkage factor. Akin to (\ref{Parametric_MSEestimate}), one can define parametric counterparts of $a_{\text{WE}}$ by replacing $\xi_k$ with $\xi_k/\alpha$.
Figure~\ref{gain_MSE_and_WE}(a) compares $a_{\text{MSE}}$ and $a_{\text{WE}}$ for different values of $\alpha$. We observe that the attenuation provided by $a_{\text{MSE}}$ is considerably smaller than $a_{\text{WE}}$ when the a posteriori SNR is greater than zero dB.  This indicates that, in the high-SNR regime, minimizing the weighted Euclidean distortion results in higher noise attenuation than the MSE. As $\alpha$ increases, the shrinkage curves shift to the right, implying more attenuation for a given a posteriori SNR.
\subsection{Logarithmic Mean-Square Error (log MSE)}
\indent Since loudness perception of the peripheral auditory system is logarithmic, one could compare speech spectra by computing the MSE on a logarithmic scale. This property has been used to advantage in vector quantization and speech coding applications~\cite{Gray}. The log MSE between $S_k$ and ${\widehat S_k}$ is given by
\begin{align}
d(S_k,\widehat{S}_k)&=\Big(\log \frac{\widehat{S}_k}{S_k}\Big)^2, \nonumber \\
&=(\log \widehat{S}_k)^2+(\log S_k)^2-2\log S_k \log \widehat{S}_k .
\label{log_distortion}
\end{align}
Recall that we consider only non-negative shrinkage functions for denoising, and that under the high SNR assumption both $\widehat{S}_k$ and $S_k$ have the same sign. Consequently, the argument of the logarithm is always positive. The last term in (\ref{log_distortion}) is rewritten as
\begin{align*}
\log \widehat{S}_k \log S_k &=\log \widehat{S}_k \log (X_k-W_k), (\text{from \eqref{xk}}),\nonumber\\
&=\log \widehat{S}_k \log X_k+ \log \widehat{S}_k \log \left(1-\displaystyle\frac{W_k}{X_k} \right), \\
&=\log \widehat{S}_k \log X_k- \log \widehat{S}_k \displaystyle \sum_{n=1}^{\infty} \displaystyle\frac{1}{n} \left(\displaystyle\frac{W_k}{X_k} \right)^n.
\end{align*}
Substituting in (\ref{log_distortion}), truncating the series beyond $n=4$, and considering the expectation results in
\begin{align*}
\mathcal{R}=\mathcal{E}\left\{(\log \widehat{S}_k)^2+(\log S_k)^2-2\log \widehat{S}_k \log X_k\right\} \\
+\mathcal{E}\left\{2\log \widehat{S}_k   \displaystyle{\sum_{n=1}^{4}} \displaystyle{\frac{1}{n}} \left(\displaystyle{\frac{W_k}{X_k}} \right)^{n} \right\}.
\end{align*}
For the shrinkage estimate $\widehat{S}_k=a_kX_k$, the risk becomes
\begin{align*}
\mathcal{R}=\mathcal{E}&\left\{(\log a_kX_k)^2+(\log S_k)^2 \right\}-2\mathcal{E}\left\{\log a_kX_k \log X_k \right\}\\
&+2\mathcal{E}\left\{\displaystyle \sum_{n=1}^{4} \displaystyle\frac{\log a_kX_k}{n} \left(\displaystyle\frac{W_k}{X_k}\right)^n\right\}.
\nonumber
\end{align*}
The last term may be simplified as shown in Appendix~\ref{app:log}, using \emph{Lemma} \ref{truncated_Gaussian_recursive_lemma2}. Consequently, the unbiased risk estimator is
\begin{align}
 \nonumber
{\widehat {\cal R}}=&(\log S_k)^2+(\log a_kX_k)^2-2\log a_kX_k \log X_k\\ \nonumber
+2&\left(\displaystyle\frac{\sigma^2}{X_k^2}-1.5\displaystyle\frac{\sigma^4}{X_k^4}+2.17\displaystyle\frac{\sigma^6}{X_k^6}-159.5\displaystyle\frac{\sigma^8}{X_k^8}\right)\\  
-2&\log a_kX_k \left(0.5\displaystyle\frac{\sigma^2}{X_k^2}-0.75\displaystyle\frac{\sigma^4}{X_k^4}-10\displaystyle\frac{\sigma^6}{X_k^6}-210\displaystyle\frac{\sigma^8}{X_k^8}\right).
\label{log_MSE_risk_estimate}
\end{align}
The optimum value of $a_k \in [0, 1]$ obtained by solving \eqref{lagrange_derivative} -- \eqref{second_order} in this case is:
\begin{align*}
a_k=\min \left\{\exp \left(0.5\displaystyle\frac{\sigma^2}{X_k^2}-0.75\displaystyle\frac{\sigma^4}{X_k^4}-10\displaystyle\frac{\sigma^6}{X_k^6}-210\displaystyle\frac{\sigma^8}{X_k^8}\right),1\right\}.
\end{align*}
The corresponding estimate of $S_k$ is given by
\begin{align*}
\widehat{S}_k&=\underbrace{\text{min}\left\{\exp \left(\displaystyle\frac{0.5}{\xi_k}-\displaystyle\frac{0.75}{\xi_k^2}-\displaystyle\frac{10}{\xi_k^3}-\displaystyle\frac{210}{\xi_k^4} \right),1\right\}}_{a_{\text{log MSE}}(\xi_k)}X_k.
\end{align*}
The parametrized shrinkage is given by $a_{\text{log MSE}}\left(\frac{\xi_k}{\alpha}\right)$. A comparison of $a_{\text{MSE}}$ and $a_{\text{log MSE}}$ is shown in  Figure~\ref{gain_MSE_and_WE}(b). At low SNRs, log MSE results in a higher attenuation than MSE.
\begin{figure}[t]
\centering
$
\begin{array}{cc}
\includegraphics[width=1.75in]{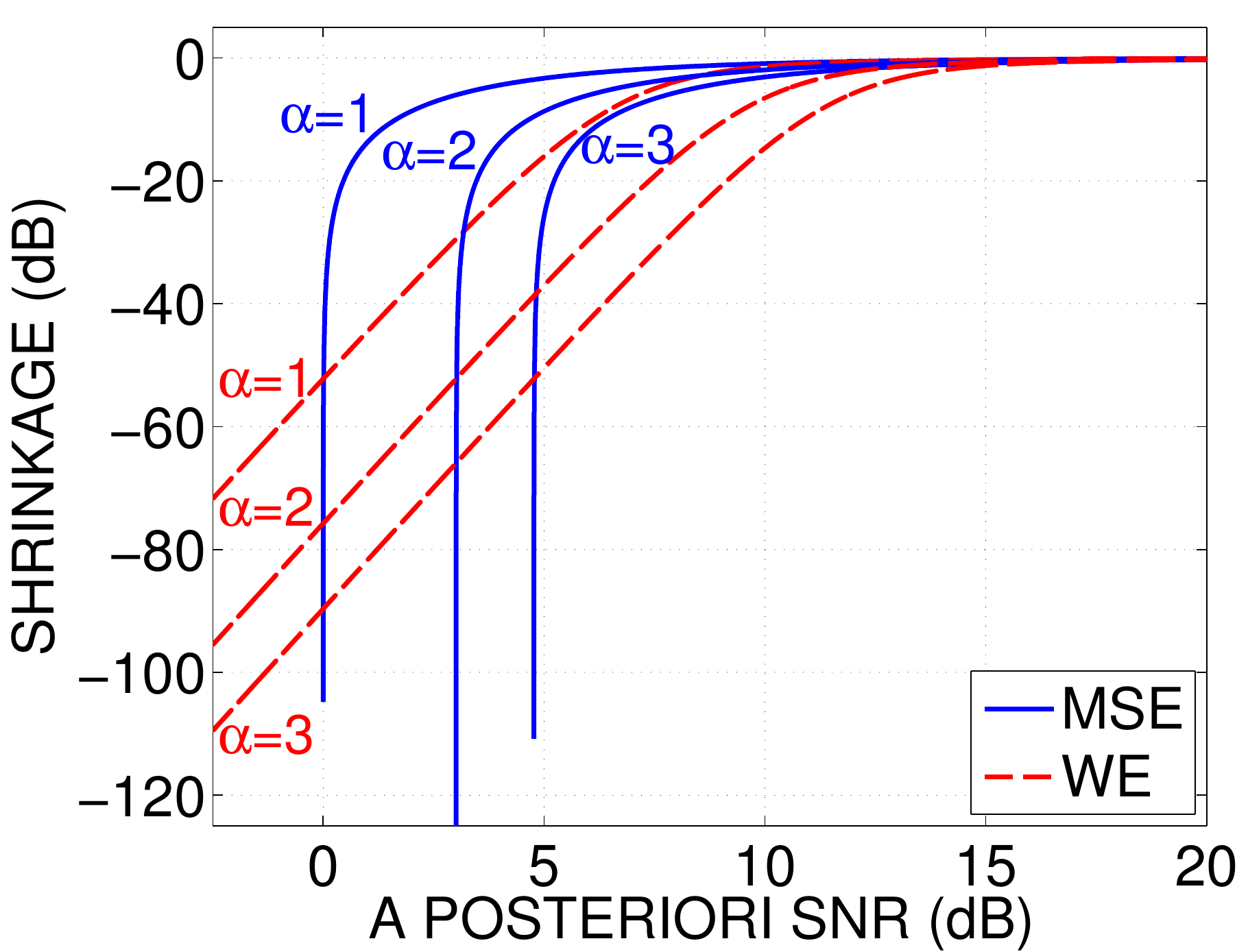}&
\includegraphics[width=1.75in]{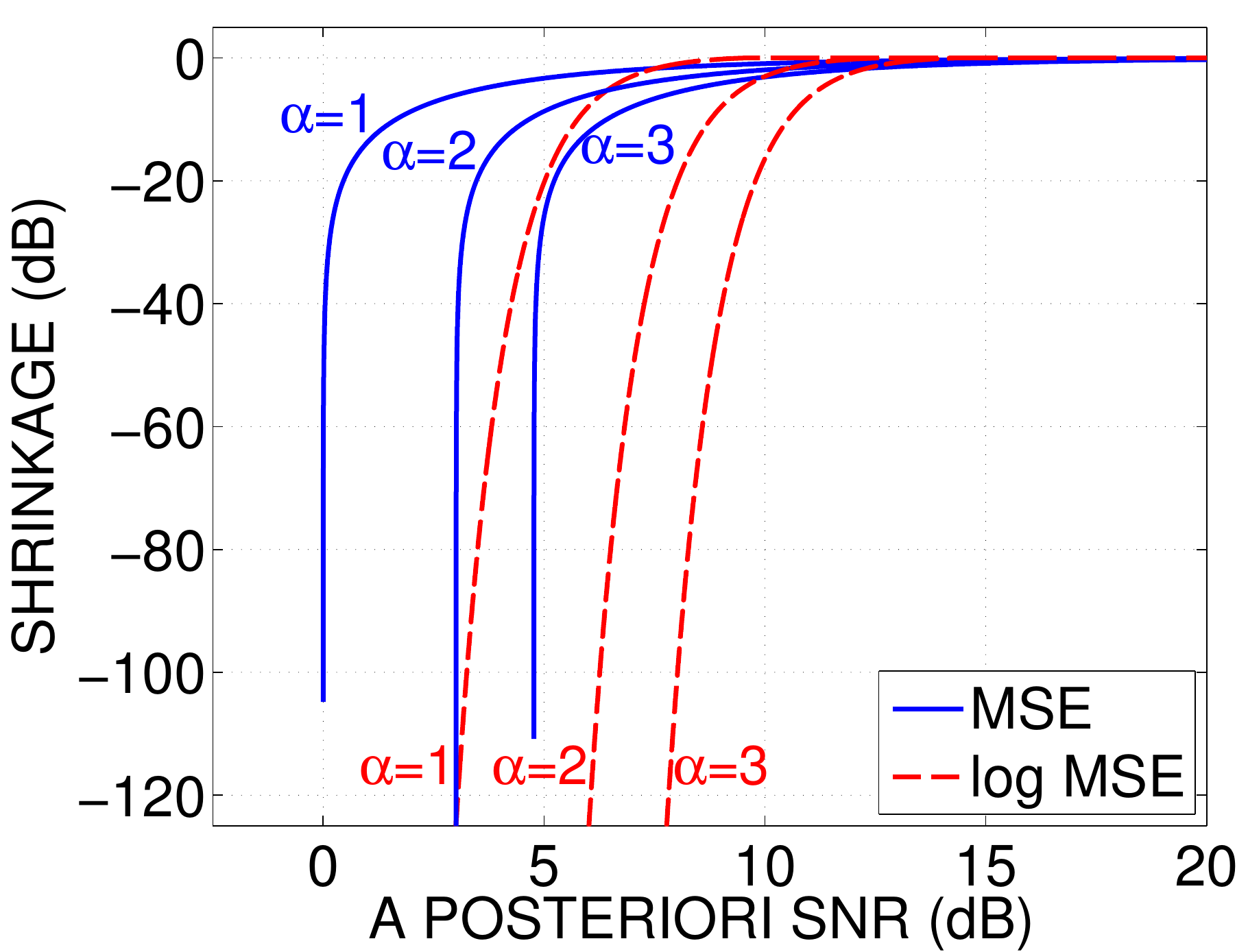}\\
\text{(a)} & \text{(b)} \\
\end{array}
$
\caption{A comparison of shrinkage profiles: (a) MSE versus WE; and (b) MSE versus Iog MSE.}
\label{gain_MSE_and_WE}
\end{figure}
\begin{figure*}[t]
\centering
$
\begin{array}{cccc}
\includegraphics[width=1.75in]{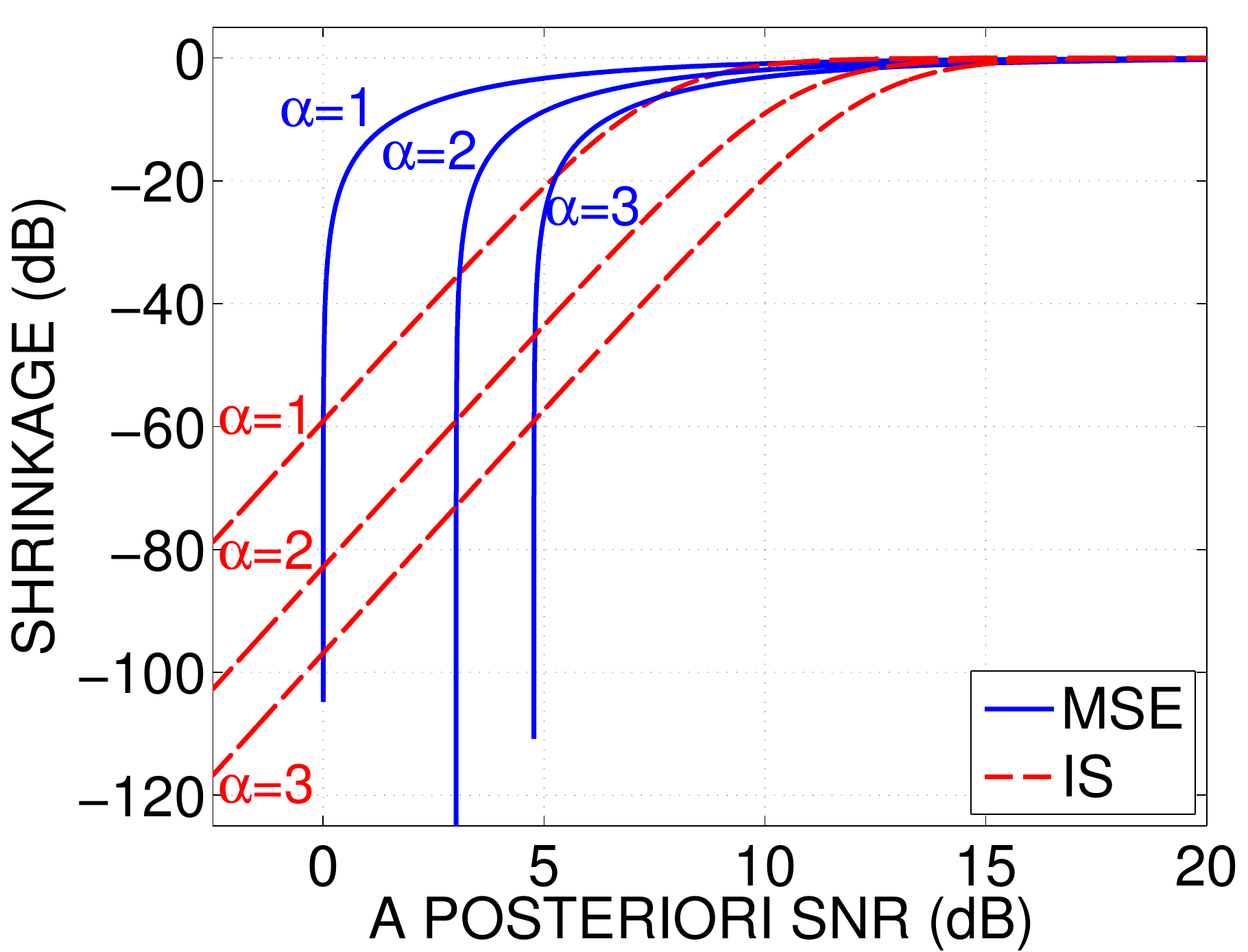}&
\includegraphics[width=1.75in]{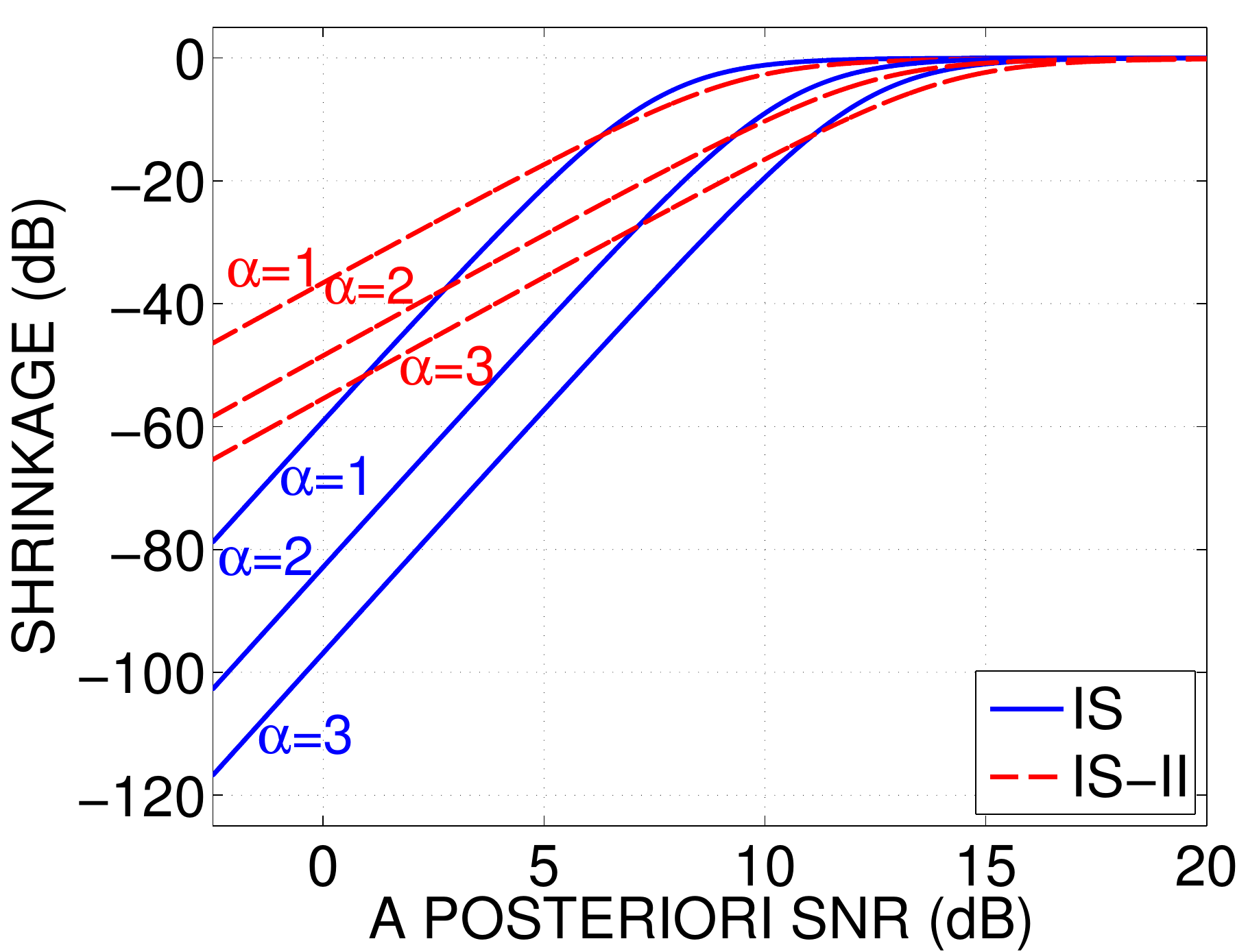}&
\includegraphics[width=1.75in]{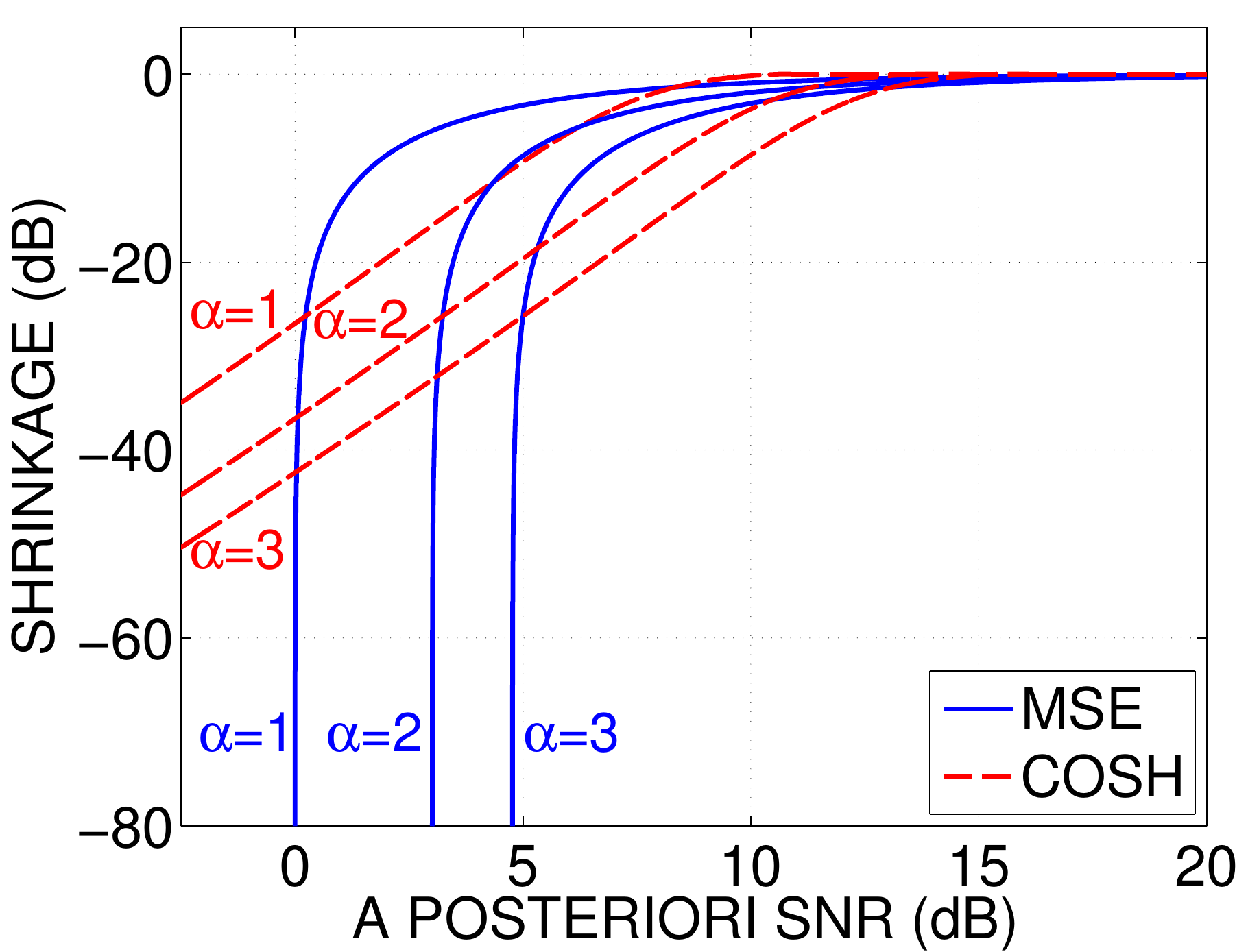}&
\includegraphics[width=1.75in]{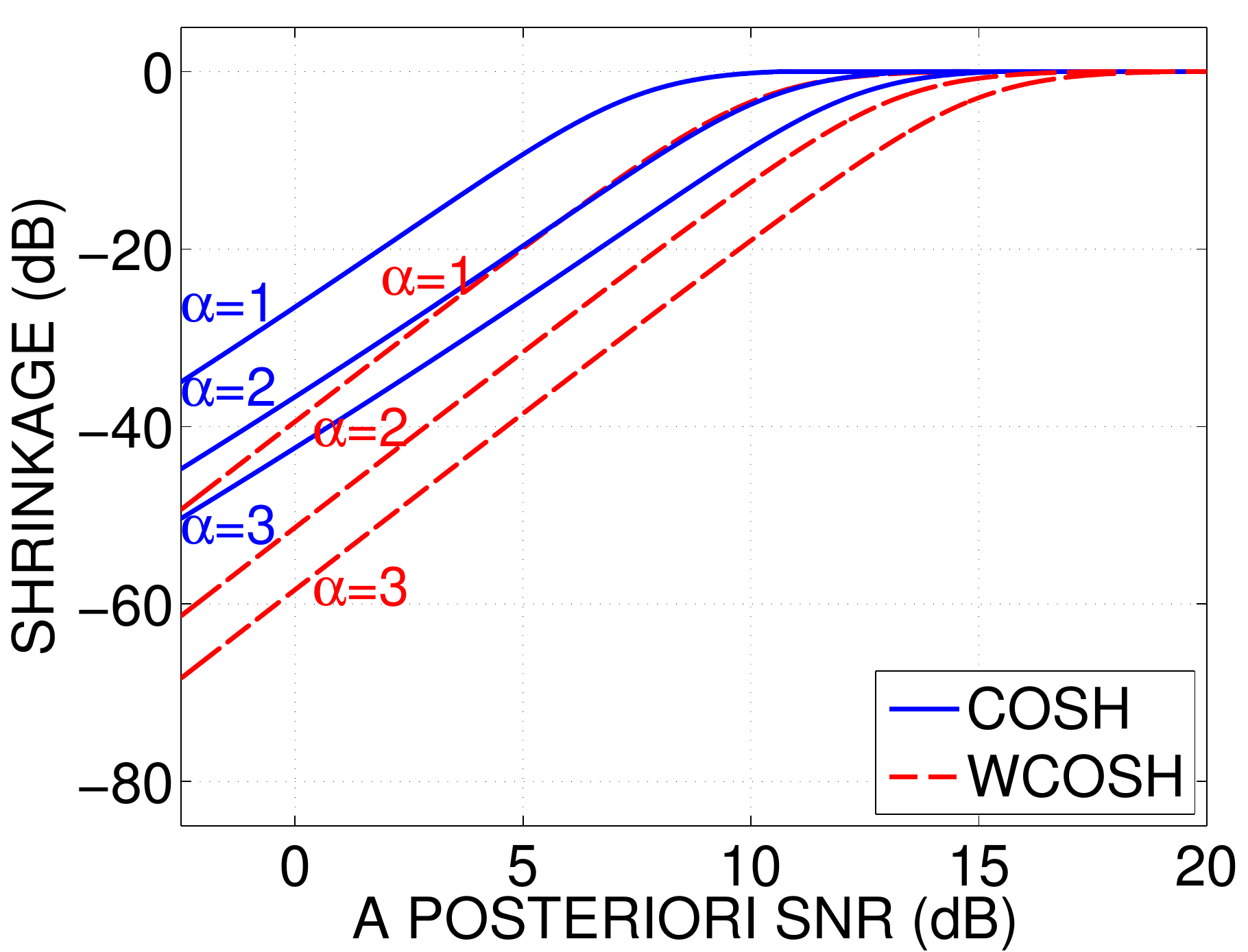}\\
\text{(a)} & \text{(b)} & \text{(c)} & \text{(d)} \\
\end{array}
$
\caption{A comparison of shrinkage profiles: (a) MSE versus IS; (b) IS versus IS-II; (c) MSE versus COSH; and (d) COSH versus WCOSH.}
\label{fig:ISandIS_2}
\end{figure*}
\subsection{Itakura-Saito (IS) Distortion}
\indent The IS distortion, although not symmetric, is a popular quality measure used in speech coding due to its perceptual relevance in matching two power spectra. Here, we compute the IS distortion between the DCT coefficients of the noise-free speech and its estimate considering both of them to have the same sign:
$$d(S_k,\widehat{S}_k)=\displaystyle\frac{\widehat{S}_k}{S_k}-\log \displaystyle\frac{\widehat{S}_k}{S_k}-1.$$
In the high-SNR scenario, the distortion measure is expanded as
\begin{align*}
d(S_k,\widehat{S}_k)&=\displaystyle\frac{\widehat{S}_k}{X_k}\left(1-\displaystyle\frac{W_k}{X_k}\right)^{-1}-\log \widehat{S}_k+ \log S_k-1,\\
&=\displaystyle\frac{\widehat{S}_k}{X_k}\displaystyle \sum_{n=0}^{\infty} \left(\displaystyle\frac{W_k}{X_k}\right)^n-\log \widehat{S}_k+ \log S_k-1.
\end{align*}
Expressing $\widehat{S}_k=a_kX_k$, and truncating the series beyond $n=4$, the risk turns out to be
\begin{equation}
\mathcal{R} = \displaystyle \sum_{n=0}^{4} {\cal E}\left\{\frac{a_kW_k^n}{X_k^n}\right\} - {\cal E}\left\{\log a_kX_k\right\} + \log S_k -1.
\label{RiskIS}
\end{equation}
The first term is evaluated using \emph{Lemma} \ref{truncated_Gaussian_recursive_lemma2} (cf. Appendix~\ref{app:IS}) resulting in the risk
$$\mathcal{R}=\mathcal{E} \left\{a_k\left(1+60\displaystyle\frac{\sigma^6}{X_k^6}+840\displaystyle\frac{\sigma^8}{X_k^8}\right)-\log a_kX_k \right\} + \log S_k -1.$$
The corresponding unbiased estimator of ${\cal R}$ is 
\begin{equation}
{\widehat {\cal R}}=a_k\left(1+60\displaystyle\frac{\sigma^6}{X_k^6}+840\displaystyle\frac{\sigma^8}{X_k^8}\right)-\log a_kX_k +\log\,S_k-1,
\label{IS_risk_estimate}
\end{equation}
which when optimized with respect to $a_k\in\left[0,1\right]$ (solution of \eqref{lagrange_derivative} to \eqref{second_order}) gives 
$$
a_k=\left(1+\displaystyle\frac{60}{\xi_k^3}+\displaystyle\frac{840}{\xi_k^4}\right)^{-1}\Rightarrow \widehat{S}_k=\underbrace{\left(1+\displaystyle\frac{60}{\xi_k^3}+\displaystyle\frac{840}{\xi_k^4}\right)^{-1}}_{a_{\text{IS}}(\xi_k)}X_k.
$$
Figure~\ref{fig:ISandIS_2}(a) shows a comparison of $a_{\text{IS}}$ and $a_{\text{MSE}}$. For a posteriori SNR greater than $0$ dB, the attenuation is higher in the case of Itakura-Saito distortion.
\subsection{Itakura-Saito Distortion Between DCT Power Spectra (IS-II)}
\indent We next consider the IS distortion between $S_k^2$ and $\widehat{S}_k^2$:
\begin{align*}
d(S_k,\widehat{S}_k)&=\displaystyle\frac{\widehat{S}_k^2}{S_k^2}-\log \displaystyle\frac{\widehat{S}_k^2}{S_k^2}-1,\\
&=\displaystyle\frac{\widehat{S}_k^2}{X_k^2} \left(1-\displaystyle\frac{W_k}{X_k}\right)^{-2}-\log \widehat{S}_k^2 +\log S_k^2 -1,\\
&=\displaystyle\frac{\widehat{S}_k^2}{X_k^2} \displaystyle \sum_{n=0}^{\infty} \left(n+1\right) \left(\displaystyle\frac{W_k}{X_k}\right)^n-\log \widehat{S}_k^2+\log S_k^2-1.
\end{align*} 
Considering ${\widehat S}_k = a_kX_k$, and truncating the series beyond $n=4$, results in the risk
\begin{align*}
\mathcal{R}=\mathcal{E}\left\{a_k^2 \displaystyle \sum_{n=0}^{4} (n+1) \left(\displaystyle\frac{W_k}{X_k}\right)^n \right\}-\mathcal{E}\left \{\log a_k^2X_k^2 \right\}+\log S_k^2-1.
\end{align*} 
Simplifying the first term using \emph{Lemma} \ref{truncated_Gaussian_recursive_lemma2} gives
\begin{align*}
\mathcal{R}&=a_k^2\mathcal{E}\left\{1+\displaystyle\frac{\sigma^2}{X_k^2}-3\displaystyle\frac{\sigma^4}{X_k^4}+360\displaystyle\frac{\sigma^6}{X_k^6}+4200\displaystyle\frac{\sigma^8}{X_k^8}\right\}\\
&\text{\hspace{1cm}}-\mathcal{E}\left \{\log a_k^2X_k^2 \right\}+\log S_k^2-1.
\end{align*}
An unbiased estimator of ${\cal R}$ is
\begin{align}
\nonumber
{\widehat {\cal R}}&=a_k^2\left(1+\displaystyle\frac{\sigma^2}{X_k^2}-3\displaystyle\frac{\sigma^4}{X_k^4}+360\displaystyle\frac{\sigma^6}{X_k^6}+4200\displaystyle\frac{\sigma^8}{X_k^8}\right) \\
&\text{\hspace{1cm}}-\log a_k^2X_k^2+\log S_k^2-1.
\label{IS_II_risk_estimate}
\end{align}
Optimizing ${\widehat {\cal R}}$ with respect to $a_k\in \left[0,1\right]$, following \eqref{lagrange_derivative} -- \eqref{second_order}, we get
\begin{align*}
a_k&=\min \left\{ 1,\left(1+\displaystyle\frac{\sigma^2}{X_k^2}-3\displaystyle\frac{\sigma^4}{X_k^4}+360\displaystyle\frac{\sigma^6}{X_k^6}+4200\displaystyle\frac{\sigma^8}{X_k^8} \right)^{-\frac{1}{2}}\right\},
\end{align*}
\begin{align*}
\Rightarrow \widehat{S}_k&=\underbrace{\min \left\{ 1,\left(1+\displaystyle\frac{1}{\xi_k}-\displaystyle\frac{3}{\xi_k^2}+\displaystyle\frac{360}{\xi_k^3}+\displaystyle\frac{4200}{\xi_k^4} \right) ^{-\frac{1}{2}}\right\} }_{a_{\text{IS-II}}(\xi_k)}X_k,
\end{align*}
which we shall refer to as the IS-II estimator. A comparison of the IS and IS-II shrinkage functions is shown in Figure~\ref{fig:ISandIS_2}(b). At low SNRs, IS attenuates more than IS-II.
\subsection{Hyperbolic Cosine Distortion Measure (COSH)}
\indent A symmetrized version of the IS distortion results in the  {\it cosh} measure:
\begin{equation}
d(S_k,\widehat{S}_k)=\cosh \left(\log \displaystyle\frac{S_k}{\widehat{S}_k}\right)-1=\displaystyle\frac{1}{2}\left[\displaystyle\frac{S_k}{\widehat{S}_k}+\displaystyle\frac{\widehat{S}_k}{S_k} \right]-1.
\nonumber
\end{equation}
The corresponding risk is $\mathcal{R}=\displaystyle\displaystyle\frac{1}{2}\mathcal{E}\left\{\displaystyle\frac{S_k}{\widehat{S}_k}\right\}+\displaystyle\frac{1}{2}\mathcal{E}\left\{\displaystyle\frac{\widehat{S}_k}{S_k}\right\}-1.$
Substituting $\widehat{S}_k=a_kX_k$, and invoking \emph{Lemma} \ref{truncated_Gaussian_recursive_lemma2}, the expectations turn out to be
\begin{align*}
&\mathcal{E}\left\{\displaystyle\frac{S_k}{a_kX_k}\right\}=\mathcal{E}\left\{\displaystyle\frac{1}{a_k}+\displaystyle\frac{\sigma^2}{a_kX_k^2}\right\},\,\,\,\mbox{and}\\
&\mathcal{E}\left\{\displaystyle\frac{a_kX_k}{S_k}\right\}=\mathcal{E}\left\{a_k\left(1+60\displaystyle\frac{\sigma^6}{X_k^6}+840\displaystyle\frac{\sigma^8}{X_k^8}\right)\right\}.
\end{align*}
Correspondingly, the unbiased risk estimator is
\begin{equation}
{\widehat {\cal R}}=\displaystyle\frac{1}{2}\left(\displaystyle\frac{1}{a_k}+\displaystyle\frac{\sigma^2}{a_kX_k^2}+a_k\left(1+60\displaystyle\frac{\sigma^6}{X_k^6}+840\displaystyle\frac{\sigma^8}{X_k^8}\right)\right)-1.
\end{equation}
Optimizing ${\widehat {\cal R}}$ with respect to $a_k\in\left[0,1\right]$ following \eqref{lagrange_derivative}--\eqref{second_order} gives
\begin{align*}
\widehat{S}_k&=\underbrace{\min  \left\{1,\sqrt{\displaystyle\frac{1+\displaystyle\frac{1}{\xi_k}}{1+60\displaystyle\frac{1}{\xi_k^3}+840\displaystyle\frac{1}{\xi_k^4}}} \right\}  }_{a_{\text{COSH}}(\xi_k)}X_k.
\end{align*}
Figure~\ref{fig:ISandIS_2}(c) shows a comparison of $a_{\text{MSE}}$ and $a_{\text{COSH}}$ versus $\xi_k$ -- it is clear that $a_{\text{COSH}}$ results in a higher attenuation than $a_{\text{MSE}}$.
\subsection{Weighted COSH Distortion  (WCOSH)}
\indent Similar to the weighted Euclidean measure, we consider the weighted cosh distortion:
\begin{equation}
d(S_k, \widehat{S}_k)=\left(\frac{1}{2}\left[\displaystyle\frac{S_k}{\widehat{S}_k}+\displaystyle\frac{\widehat{S}_k}{S_k}\right]-1\right)\displaystyle\frac{1}{S_k}.
\end{equation}
The risk estimator can be computed as in the case of the cosh measure (cf. Appendix~\ref{app:wcosh}). The optimal estimate is given by 
\begin{align}
\widehat{S}_k&=\underbrace{\min \left\{1,\left(1-\displaystyle\frac{1}{\xi_k}+\displaystyle\frac{3}{\xi_k^2}+\displaystyle\frac{420}{\xi_k^3}+\displaystyle\frac{8400}{\xi_k^4}\right)^{-\frac{1}{2}}\right\}}_{a_{\text{WCOSH}}(\xi_k)}X_k.\nonumber
\end{align}
A comparison of $a_{\text{COSH}}$ and $a_{\text{WCOSH}}$ shown in Figure~\ref{fig:ISandIS_2}(d) indicates that the latter results in higher noise attenuation at high noise levels.
\section{Implementation Details}
\indent For experimental validation, we use the clean speech and nonstationary noise data from the NOIZEUS database~\cite{Loizou_SE_comparison} ($8$\,kHz sampling frequency and $16$-bit quantization). Noisy speech is generated by adding noise at a desired global SNR. The noisy speech signal is processed in the DCT domain on a frame-by-frame basis. The frame size is $40$ ms, the window function is Hamming, and the overlap between consecutive frames is $75\%$. We experiment with both stationary and nonstationary noise types. The developments in the PROSE framework assumed knowledge of the noise variance in each bin of a given frame. In practice, the noise variance has to be estimated. For this purpose, we use the likelihood-ratio-test-based voice activity detector (VAD) of Sohn et al. \cite{Sohn}, which was also employed in the extensive comparisons reported by Hu and Loizou \cite{Loizou_SE_comparison}. The inverse a posteriori SNR is then estimated using the recursive formula
\begin{equation}
\frac{1}{\widehat{\xi}_k\left(i\right)}=\beta \frac{\widehat {\sigma}^{2}_k\left(i\right)}{X^{2}_{k}\left(i\right)} + \left(1-\beta\right) \text{max} \left( 1- \frac{\widehat{S}^{2}_k\left(i-1\right)}{X^{2}_{k}\left(i-1\right)},0\right),
\end{equation}
where $k$ denotes the index of the DCT coefficient, $i$ denotes the frame index, and $\beta$ is a smoothing parameter (set to $0.98$ in our experiments), $\widehat{\sigma}^{2}_k(i)$ is the estimate of the noise variance at the $k^{th}$ DCT coefficient in the $i^{th}$ frame, ${\widehat{S}_k\left(i-1\right)}$ is the denoised $k^{th}$ DCT coefficient in the $(i-1)^{th}$ frame. The first few frames are assumed to contain only noise. For the first frame, $\beta$ is set to unity, and the $k^{th}$ noise variance is estimated by averaging the noise variances of the $k^{th}$ coefficient in the first ten frames. The noise variance is updated in the noise-only frames as classified by the VAD following the recursion:
$$
\widehat{\sigma}^{2}_{k}\left(i \right) =
\begin{cases}
&\eta \widehat{\sigma}^{2}_{k}\left(i -1\right)+\left(1-\eta\right) X^{2}_{k}\left(i\right), \text{under}\quad \mathcal{H}_0,\\
&\widehat{\sigma}^{2}_{k}(i-1), \text{under}\quad \mathcal{H}_1,\\
\end{cases}
$$
where $\mathcal{H}_0$ and $\mathcal{H}_1$ are the null hypothesis (noise-only) and the alternative hypothesis (signal + noise), respectively. The value of $\eta = 0.98$ following \cite{Loizou_SE_comparison}. The noisy speech signals are denoised using shrinkage functions corresponding to various distortion measures considered in this paper. We set $\alpha = 1.75$ uniformly across all measures in the PROSE framework as it was found to give better results than $\alpha=1$. The enhanced speech frames are combined using overlap-and-add synthesis to result in the denoised speech signal.
\section{Experimental Results}
\label{sec:exptres}
\indent The denoising performance of PROSE estimators is evaluated using both objective measures and subjective listening tests. For benchmarking, we use three algorithms:  (i) Wiener filter  method, where a priori SNR is estimated using the decision-directed approach (WFIL)~\cite{Scalart}; (ii) a Bayesian estimator for short-time log-spectral amplitude (LSA)~\cite{Ephraim-Malah2};  and (iii) Bayesian non-negative matrix factorization algorithm (BNMF)~\cite{Mohammadiha}\footnote{Matlab implementation available online: \url{https://www.uni-oldenburg.de/en/mediphysics-acoustics/sigproc/staff/nasser-mohammadiha/matlab-codes/}}, which gives an MMSE estimate of the clean speech DFT magnitude. In the BNMF approach, the speech bases are learned offline and the noise bases are learned online. The WFIL and LSA algorithms were shown to perform better than the other techniques in an extensive evaluation carried out by Hu and Loizou \cite{Loizou_SE_comparison}\footnote{Matlab implementations of algorithms used in \cite{Loizou_SE_comparison} are available with \cite{PLoizou}. For performance comparisons, we have used Matlab codes given with \cite{PLoizou}. }. In terms of intelligibility also, these techniques were shown to be superior (Ch. 11, pp. 567 of \cite{PLoizou}). BNMF has been shown to be the best among the NMF approaches for speech enhancement.

\subsection{Objective Evaluation}
\indent The denoising performance is quantified in terms of: (i) global SNR; (ii) average segmental SNR (SSNR), a local metric, which is the average of  the SNRs computed over short segments; (iii) perceptual evaluation of speech quality (PESQ)~\cite{pesq}, which is an objective score for assessing end-to-end speech quality in narrowband telephone networks and speech codecs, described in ITU-T Recommendation P.862~\cite{pesq}; and (iv)  short-time objective intelligibility measure (STOI) \cite{stoi}\footnote{Matlab implementation available at \url{http://siplab.tudelft.nl/}}, which ranges from $0$ to $1$ and reflects the correlation between short-time temporal envelope of clean speech and denoised speech. It has been shown to correlate highly with the intelligibility scores obtained through listening tests~\cite{stoi}. Measures (i), (ii), and (iii) assess the speech quality, whereas (iv) measures the  intelligibility. For SNR, SSNR, and PESQ, we report the {\it gains} achieved by denoising. The SNR gain is the difference between the output and input SNR values. The PESQ gain and the SSNR gain are also computed similarly.\\
\indent We consider all $30$ speech files from the NOIZEUS database~\cite{Loizou_SE_comparison}, and three noise types -- white Gaussian noise, train noise, and street noise, the last two being real-world nonstationary noise types. The results presented here are obtained after averaging over the entire database for $50$ noise realizations.
\subsubsection{White Gaussian noise}
\indent Figure~\ref{figure:objective_scores_white_noise}(a) shows the SNR gain of different algorithms in white Gaussian noise. For input SNRs in the range of $-5$ dB to $20$ dB,  PROSE with cosh and log-MSE distortions results in a higher SNR gain compared with the other algorithms, followed by IS and WE. As the input SNR increases, the margin of improvement offered by PROSE techniques over BNMF, LSA, and WFIL increases. The segmental SNR (SSNR) gain is shown in Figure~\ref{figure:objective_scores_white_noise}(b). The trends of SSNR gain and SNR gain are similar. The gains are negative for the competing techniques at high SNRs. A negative gain indicates that the output SSNR or PESQ score is worse than that of the input. The PESQ scores are shown in  Figure~\ref{figure:objective_scores_white_noise}(c).  For input SNR in the range $-5$ dB to $20$ dB, the  PESQ gain is maximum for WE, log MSE, and IS.  Below $5$ dB input SNR,  BNMF also shows a high PESQ gain. We observe that, within the PROSE framework, PESQ gains are higher for perceptually motivated distortions than the MSE. The denoising performance of PROSE based on WE, log MSE, COSH, and IS distortions is better in terms of SNR, SSNR, and PESQ compared with the benchmark techniques.\begin{figure}[t]
\centering
$\begin{array}{c}
\hspace{-0.3cm}\includegraphics[width=9.5 cm]{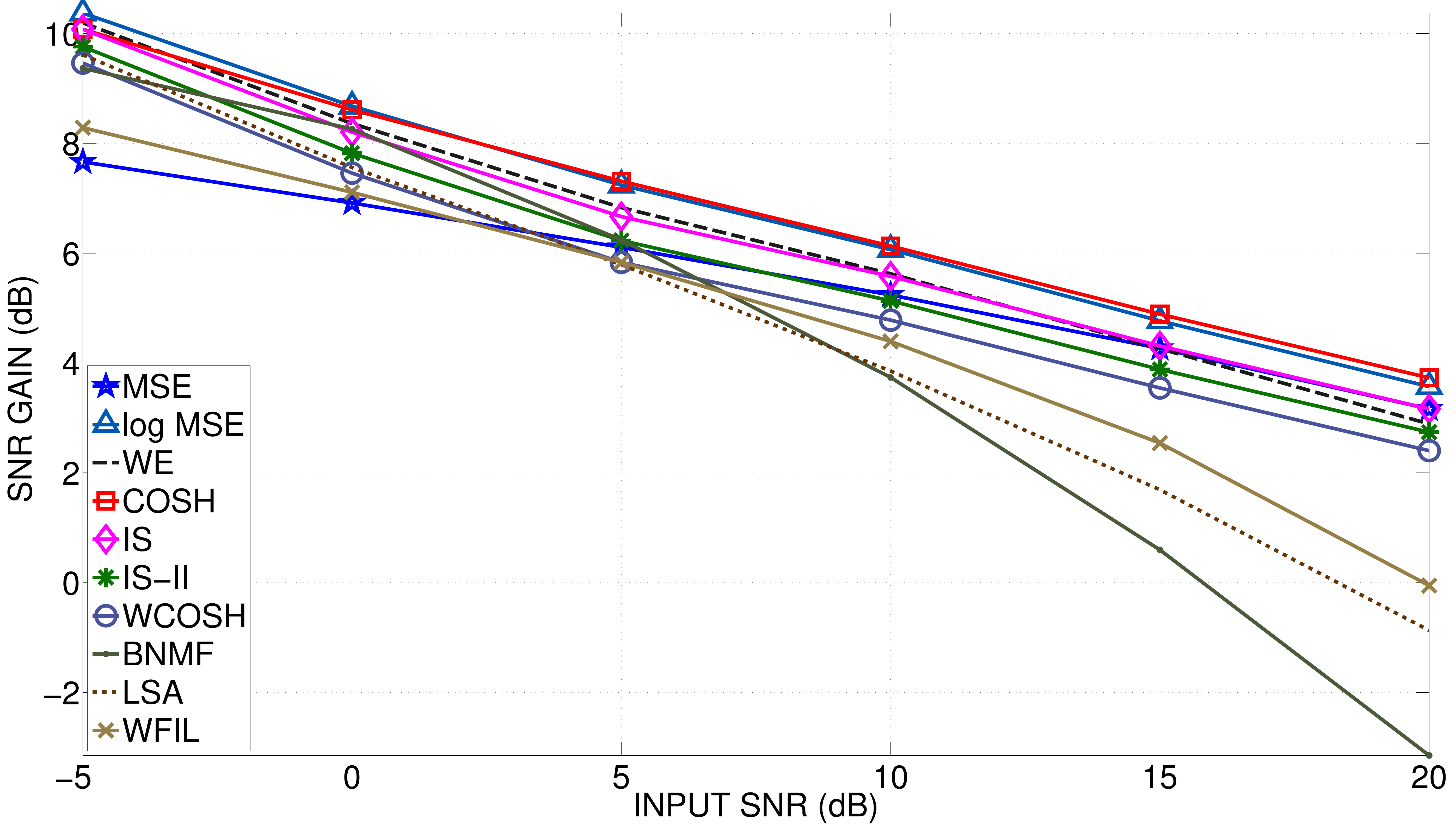}\\
\text{(a)}\\
\hspace{-0.3cm}\includegraphics[width=9.5 cm]{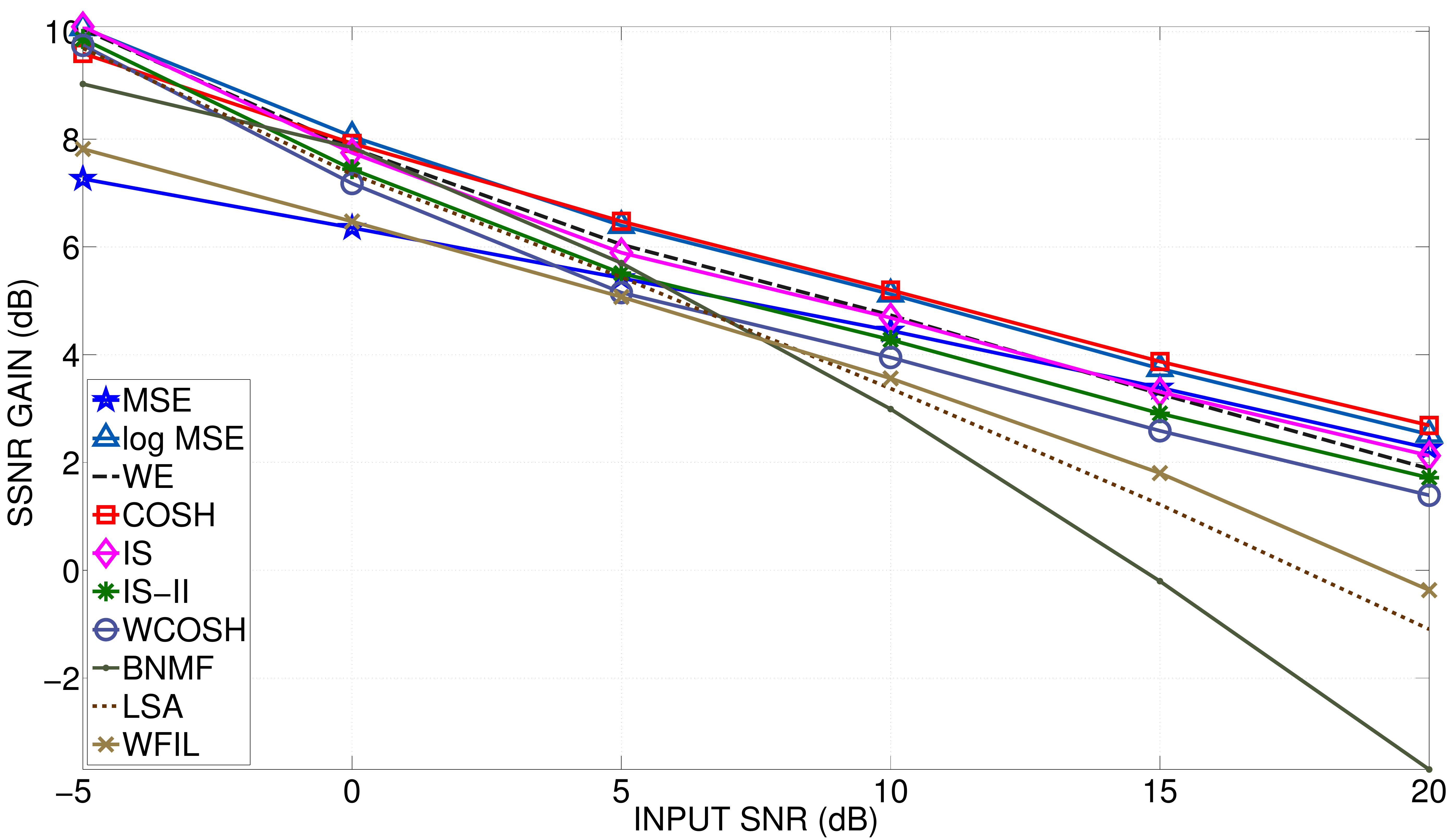}\\
\text{(b)}\\
\hspace{-0.3cm}\includegraphics[width=9.5 cm]{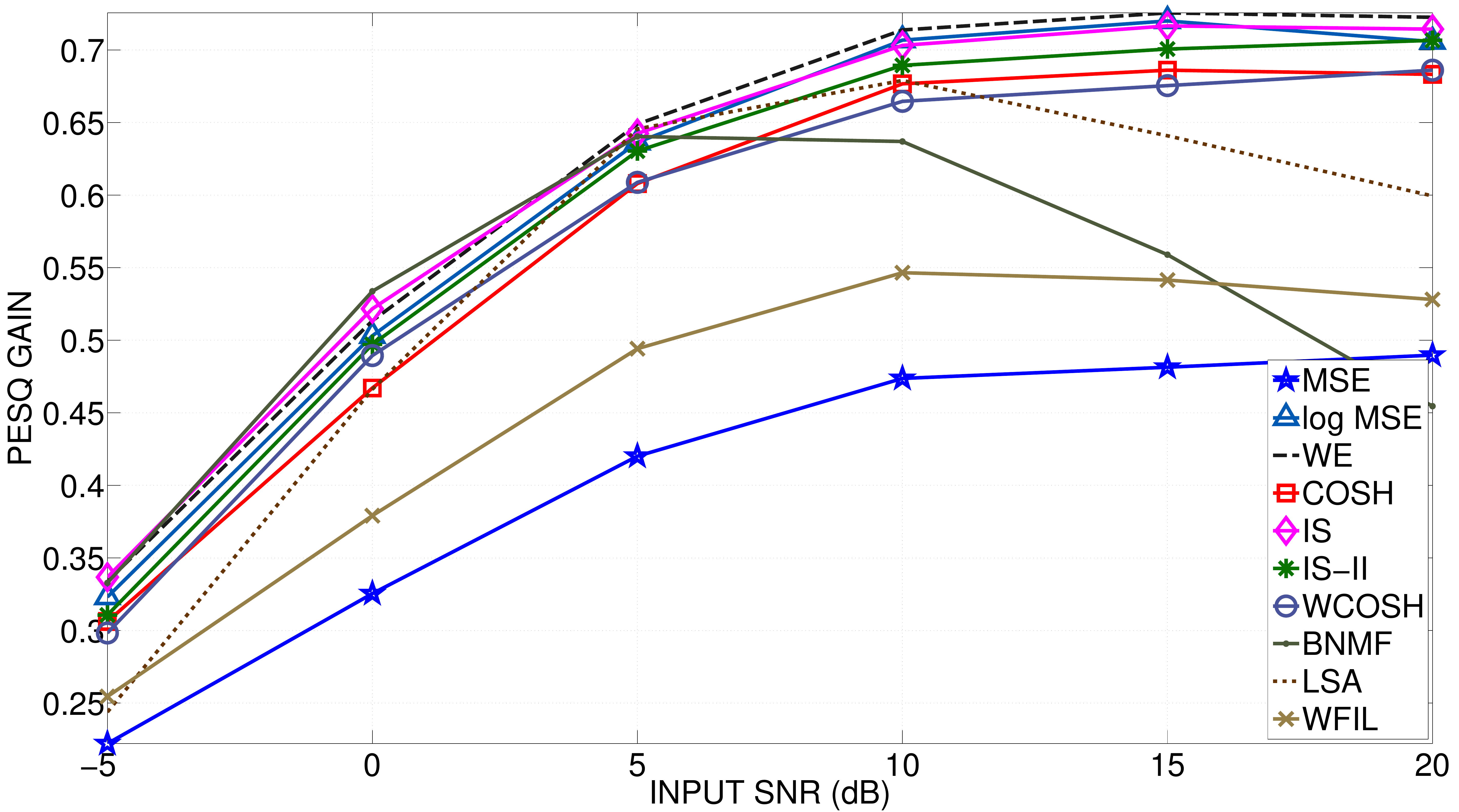}\\
\text{(c)}
\end{array}$
\caption{\small [Color online] Denoising performance in white Gaussian noise: (a) SNR gain; (b) SSNR gain; and (c) PESQ gain.}
\label{figure:objective_scores_white_noise}
\end{figure}
\begin{figure}[t]
\centering
$\begin{array}{c}
\hspace{-0.3cm}\includegraphics[width=9.5 cm]{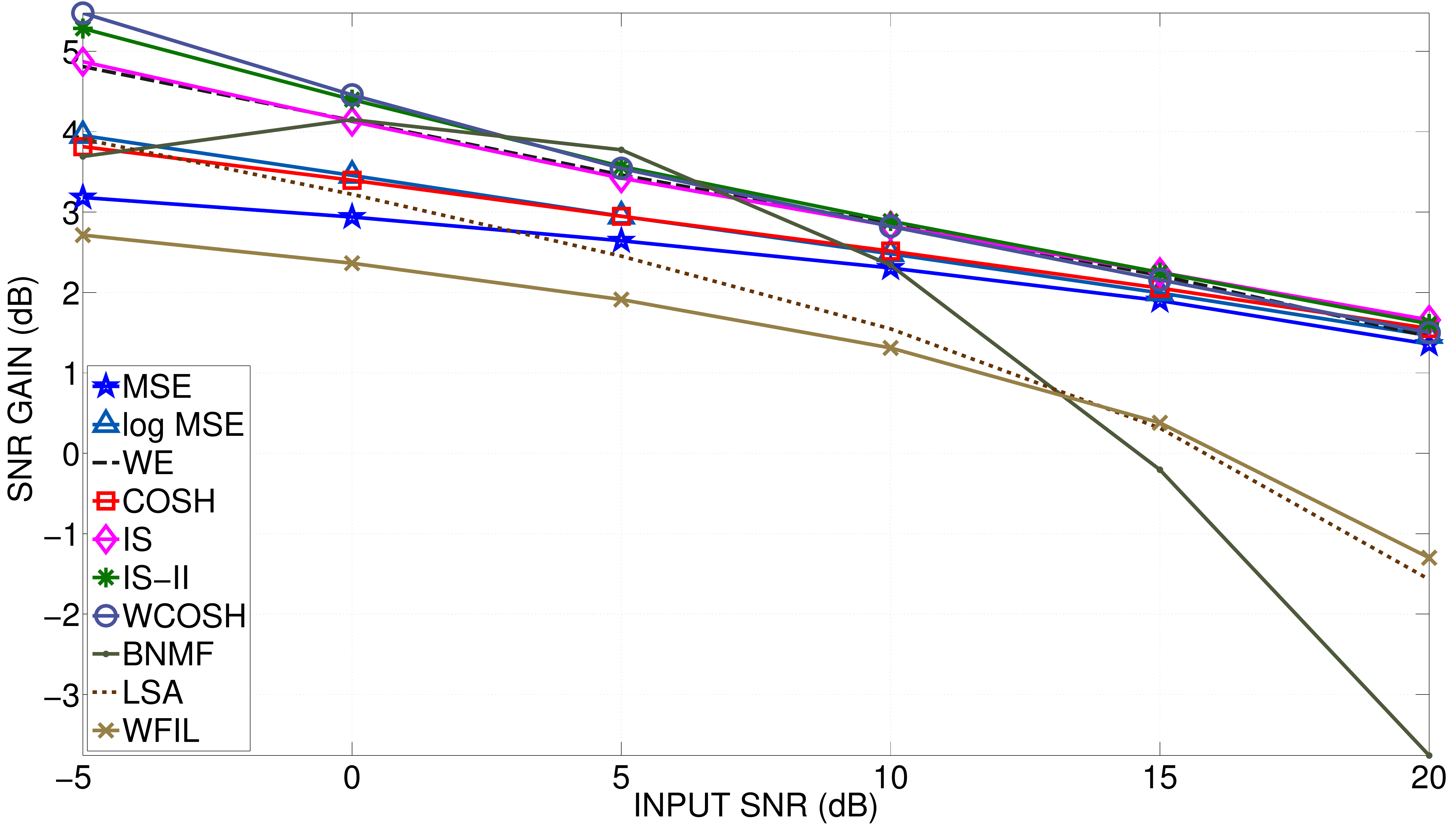}\\
\text{(a)}\\
\hspace{-0.3cm}\includegraphics[width=9.5 cm]{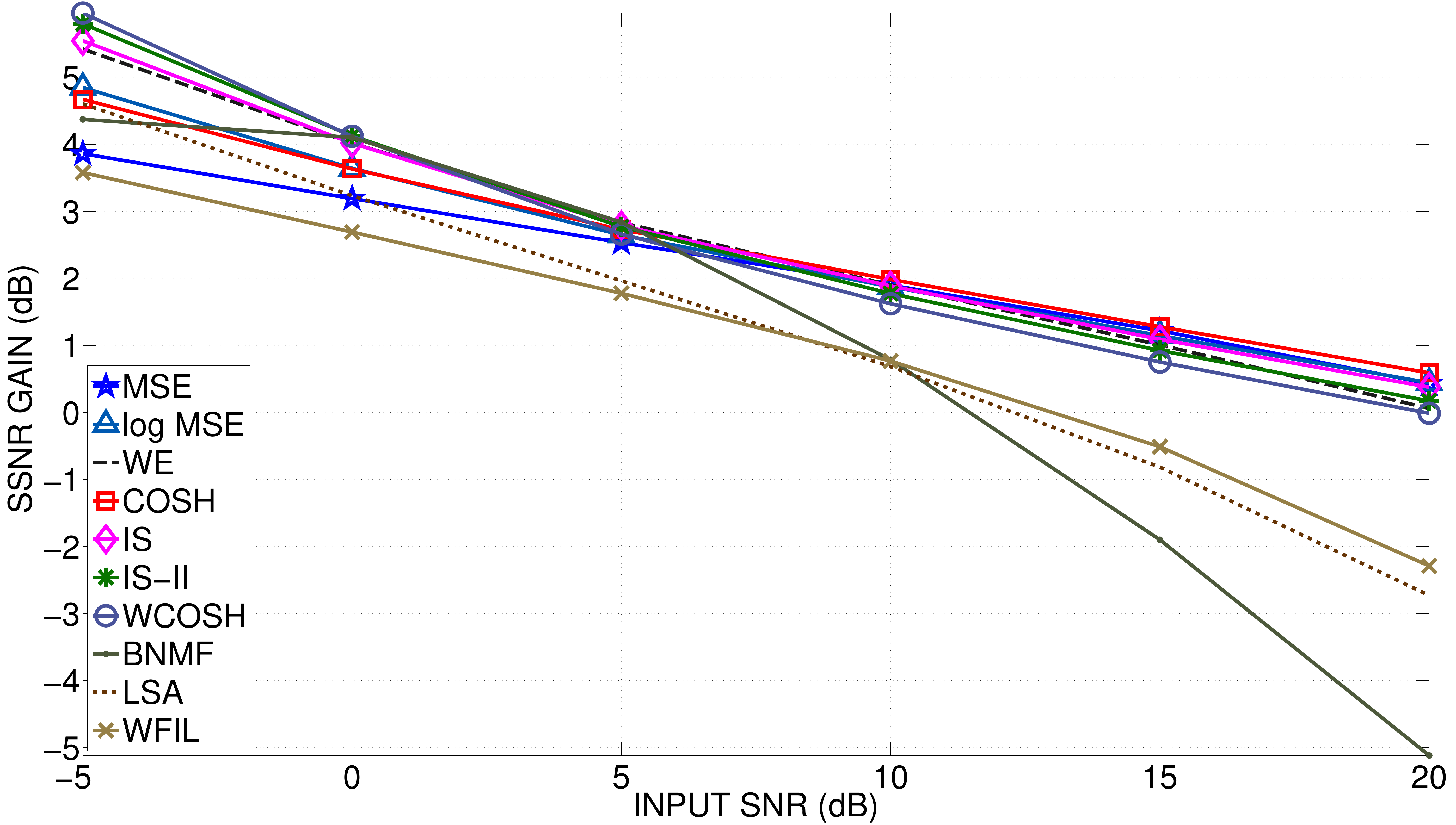}\\
\text{(b)}\\
\hspace{-0.3cm}\includegraphics[width=9.5 cm]{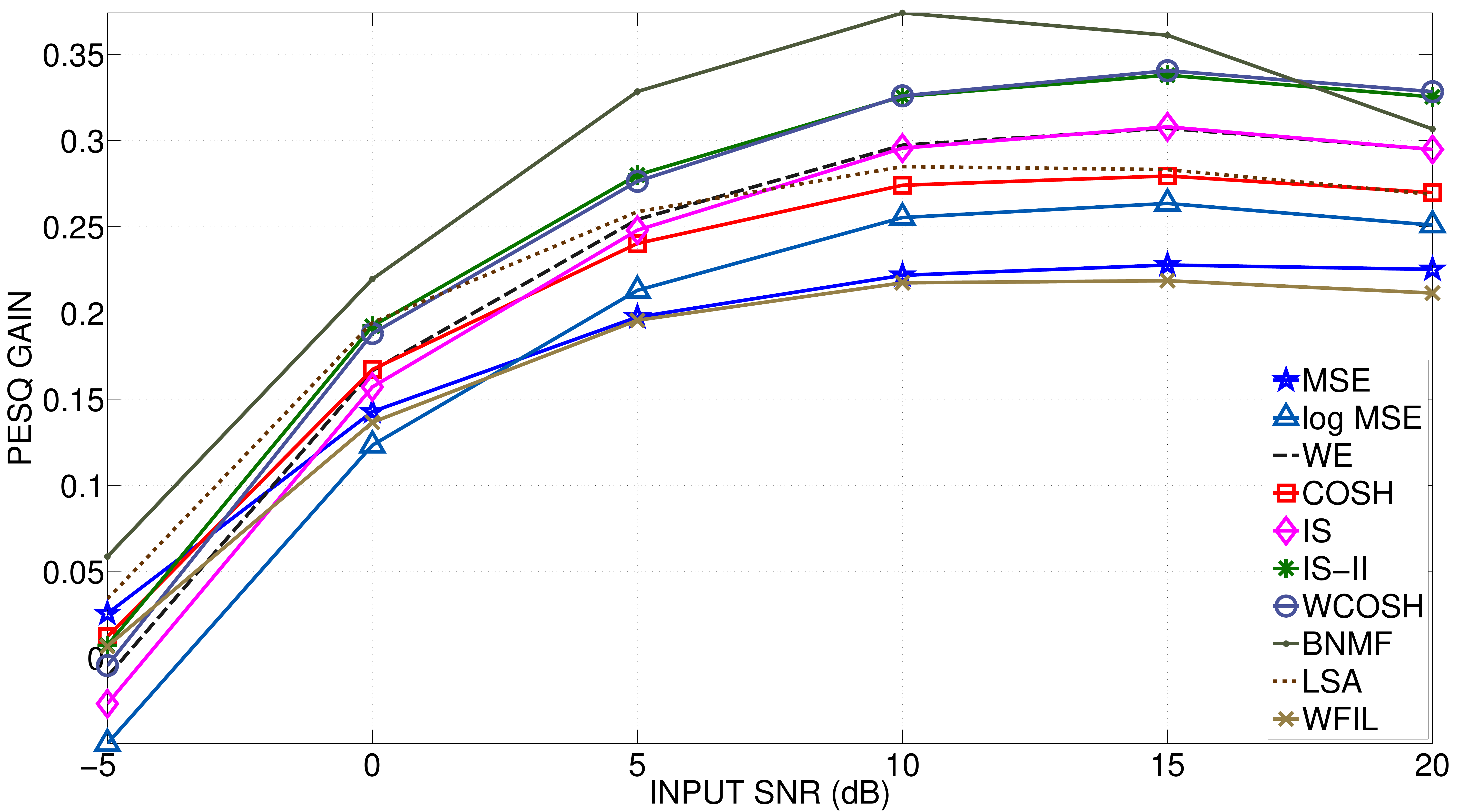}\\
\text{(c)}
\end{array}$
\caption{\small [Color online] A comparison of denoising performance in street noise: (a) SNR gain; (b) SSNR gain; and (c) PESQ gain.}
\label{figure:objective_scores_street_noise}
\end{figure}
\begin{figure}[t]
\centering
$\begin{array}{c}
\hspace{-0.3cm}\includegraphics[width=9.5 cm]{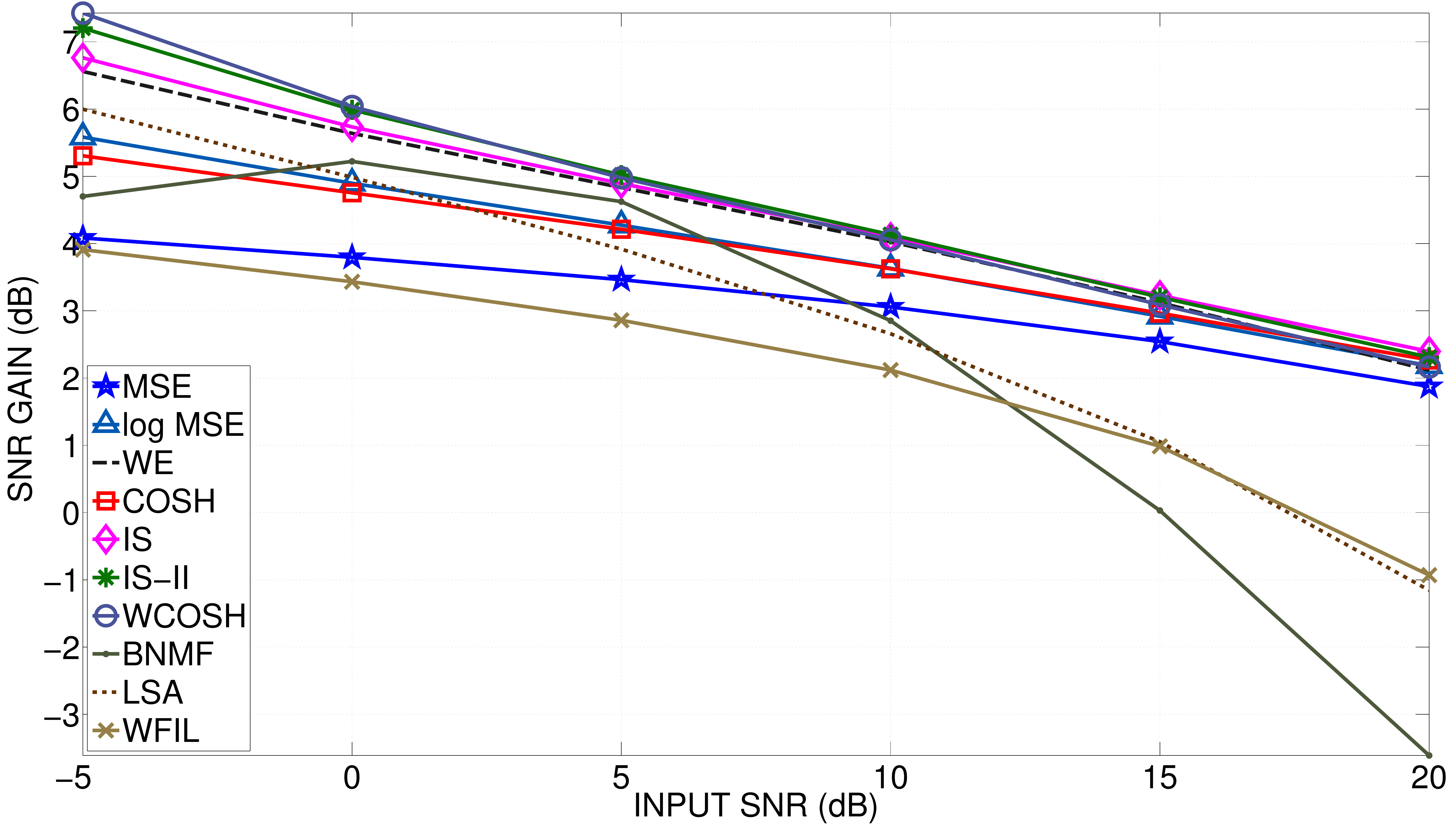}\\
\text{(a)}\\
\hspace{-0.3cm}\includegraphics[width=9.5 cm]{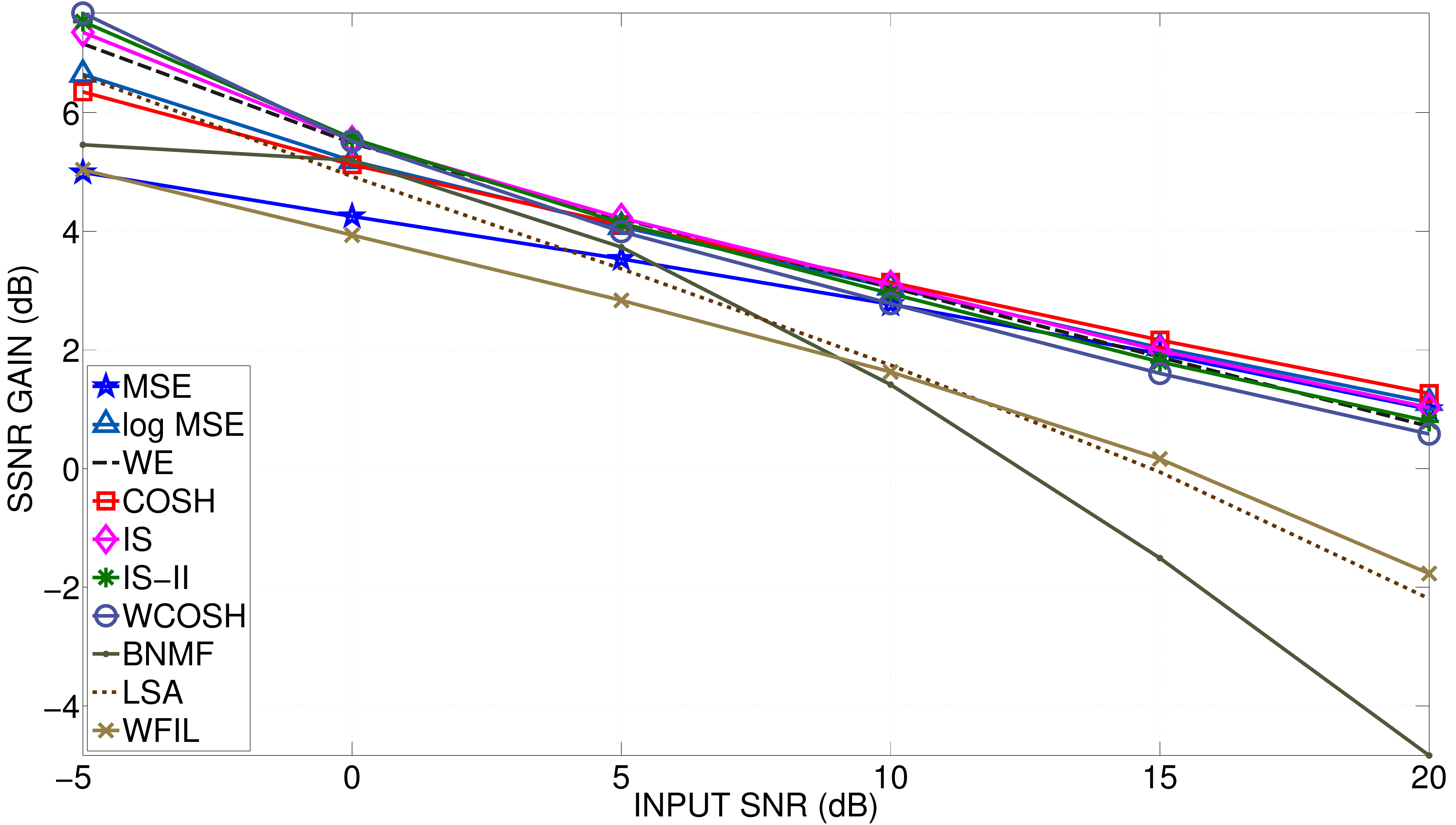}\\
\text{(b)}\\
\hspace{-0.3cm}\includegraphics[width=9.5 cm]{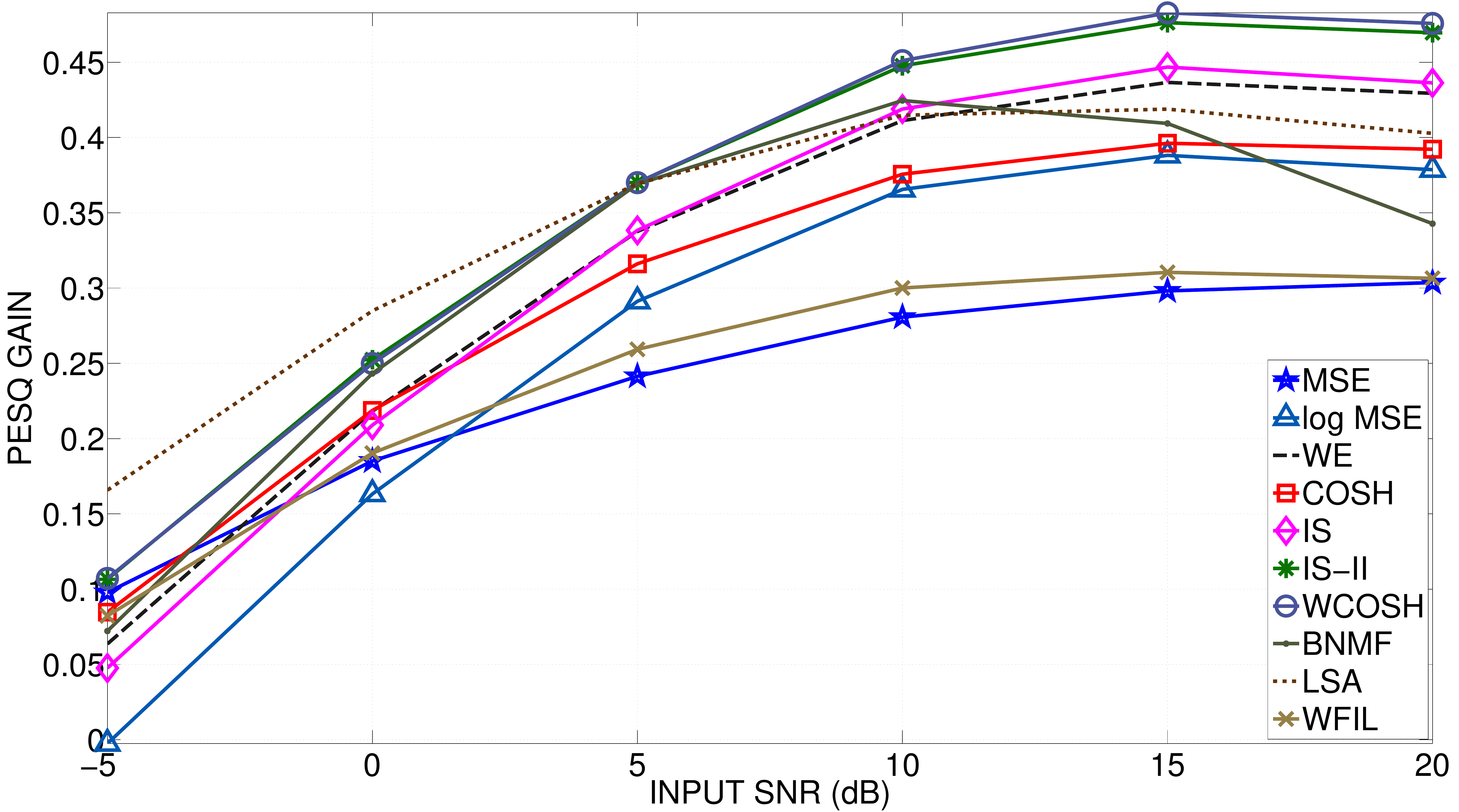}\\
\text{(c)}
\end{array}$
\caption{\small [Color online]  A comparison of denoising performance in train noise: (a) SNR gain; (b) SSNR gain; and (c) PESQ gain.}
\label{figure:objective_scores_train_noise}
\end{figure}
\indent The STOI scores are shown in Table~\ref{STOI_SCORES_WHITE_NOISE}. We observe that for SNRs below $5$ dB, BNMF has higher STOI scores than the other algorithms. For SNRs greater than $ 5$ dB, the PROSE framework results in a higher STOI.
\begin{table}[t]
\centering
\caption{\small Comparison of denoising performance in terms of STOI scores for different input SNRs (white Gaussian noise).}
\begin{tabular}{ p{1.3cm} | p{0.75cm}| p{0.75cm}|p{0.75cm} |p{0.75cm}| p{0.75cm}|p{0.75cm}}
\hline
\hline
\small {SNR (dB)}  & $-5$  &  $0$   &  $5$      &  $10$    &   $15$    &  $20$    \\
\hline
\hline
    \small{Input}&$ 0.560 $&$ 0.663 $&$   0.757 $&$   0.835 $&$   0.894 $&$   0.936$  \\ \hline
     \small{MSE}  &$ 0.571 $&$ 0.686 $&$   0.785 $&$   0.861 $&$   0.916 $&$   0.952$  \\ \hline
\small{log MSE} &$ 0.534 $&$ 0.659 $&$   0.772 $&$   0.863 $&$   0.919 $&$   \textbf{0.956}$  \\ \hline
    \small{WE}     &$ 0.546 $&$ 0.664 $&$   0.773 $&$   0.863 $&$   \textbf{0.920} $&$\textbf{0.956}$  \\ \hline
    \small{COSH}&$ 0.559 $&$ 0.677 $&$   0.782 $&$    \textbf{0.865} $&$    \textbf{0.920} $&$    \textbf{0.956}$  \\ \hline
    \small{ IS}      &$ 0.543 $&$ 0.658 $&$   0.767 $&$   0.861 $&$   0.919 $&$   0.955$  \\ \hline
    \small{IS-II}    &$ 0.553 $&$ 0.664 $&$   0.768 $&$   0.860 $&$   0.918 $&$   0.955$  \\ \hline
   \small{WCOSH}&$ 0.553 $&$ 0.660 $&$   0.763 $&$   0.858 $&$   0.916 $&$   0.953$  \\ \hline
    \small{BNMF}&$ \textbf{0.586} $&$ \textbf{0.707} $&$  \textbf{0.802} $&$   0.864 $&$   0.901 $&$   0.923$  \\ \hline
    \small{LSA}&$ 0.565 $&$ 0.664 $&$   0.764 $&$   0.843 $&$   0.898 $&$   0.939$  \\ \hline
    \small{WFIL}&$ 0.570 $&$ 0.682 $&$   0.781 $&$   0.858 $&$   0.911 $&$   0.948$  \\ \hline
 \hline
\end{tabular}
\label{STOI_SCORES_WHITE_NOISE}
\end{table}
\subsubsection{Street noise}  The denoising results are presented in Figure~\ref{figure:objective_scores_street_noise}. The PROSE framework based on WCOSH, WE, IS and IS-II distortions results in a higher SNR gain and SSNR gain than the other algorithms. (cf. Figures~\ref{figure:objective_scores_street_noise}(a) and (b)). Similar to the white noise scenario, the margin of improvement in case of PROSE increases with increase in input SNR than the benchmark algorithms. The PESQ gain shows a slightly different trend (Figure~\ref{figure:objective_scores_street_noise}(c)). BNMF gives a marginally higher PESQ gain (about 0.05 higher) than PROSE. This is attributed to the training phase in BNMF, which is particularly advantageous in low SNR and nonstationary noise conditions. Within the PROSE family of denoisers, distortion measures other than the MSE result in a higher PESQ score. Table~\ref{STOI_SCORES_STREET_NOISE} shows the STOI scores. We observe that, for SNRs greater than $5$ dB, PROSE algorithms yields a higher STOI score than BNMF, WFIL, and LSA.

\begin{table}[t]
\centering
\caption{\small Comparison of denoising performance in terms of STOI scores for different input SNRs (street noise).}
\begin{tabular}{ p{1.3cm} | p{0.7cm}| p{0.7cm}|p{0.7cm} |p{0.7cm}| p{0.7cm}|p{0.7cm}}
\hline
\hline
\small {SNR (dB)}  & $-5$  &  $0$   &  $5$      &  $10$    &   $15$    &  $20$    \\
\hline
\hline
        \small{Input}&$ 0.540    $&$ 0.659    $&$ 0.774    $&$ 0.867    $&$ 0.930    $&$ 0.967$ \\ \hline
           \small{MSE}&$ 0.536    $&$ 0.669    $&$  \textbf{0.790}    $&$  \textbf{0.881}    $&$ 0.939    $&$  \textbf{0.972}$ \\ \hline
     \small{log MSE}&$ 0.512    $&$ 0.650    $&$ 0.778    $&$ 0.878    $&$ 0.939   $&$ 0.971$ \\ \hline
             \small{WE}&$ 0.515    $&$ 0.653    $&$ 0.780    $&$ 0.879    $&$  \textbf{0.940}    $&$  \textbf{0.972}$ \\ \hline
        \small{COSH}&$ 0.524    $&$ 0.663    $&$ 0.785    $&$  \textbf{0.881}    $&$  \textbf{0.940}    $&$  \textbf{0.972}$ \\ \hline
               \small{IS}&$ 0.510    $&$ 0.647    $&$ 0.775    $&$ 0.876    $&$ 0.939    $&$ 0.971$ \\ \hline
            \small{IS-II}&$ 0.515    $&$ 0.650    $&$ 0.776    $&$ 0.876    $&$ 0.939    $&$ 0.971$ \\ \hline
     \small{WCOSH}&$ 0.511    $&$ 0.646    $&$ 0.772    $&$ 0.873    $&$ 0.937    $&$ 0.971$ \\ \hline
        \small{BNMF}&$  \textbf{0.543}    $&$  \textbf{0.672}    $&$ 0.786    $&$ 0.865   $&$ 0.910    $&$ 0.932$ \\ \hline
         \small{LSA}&$ 0.512    $&$ 0.637    $&$ 0.757    $&$ 0.854    $&$ 0.921    $&$ 0.960$ \\ \hline
        \small{WFIL}&$   0.528    $&$ 0.658    $&$ 0.777    $&$ 0.871    $&$ 0.932    $&$ 0.967$ \\ \hline
 \hline
\end{tabular}
\label{STOI_SCORES_STREET_NOISE}
\end{table}
\begin{table}[t]
\centering
\caption{\small Comparison of denoising performance in terms of STOI scores for various input SNRs (train noise).}
\begin{tabular}{ p{1.3cm} | p{0.7cm}| p{0.7cm}|p{0.7cm} |p{0.7cm}| p{0.7cm}|p{0.7cm}}
\hline
\hline
\small {SNR (dB)}  & $-5$  &  $0$   &  $5$      &  $10$    &   $15$    &  $20$    \\
\hline
\hline
    \small{Input}&$0.554 $&$     0.681 $&$     0.796 $&$     0.883 $&$    0.942 $&$     0.975$ \\ \hline
       \small{MSE}&$\textbf{0.552} $&$      \textbf{0.695} $&$      \textbf{0.813} $&$     0.896 $&$     \textbf{0.949} $&$      \textbf{0.978}$ \\ \hline
 \small{log MSE}&$0.516 $&$     0.675 $&$     0.808 $&$     0.896 $&$    0.948 $&$     0.977$ \\ \hline
         \small{WE}&$0.522 $&$     0.678 $&$     0.809 $&$     0.897 $&$     \textbf{0.949} $&$     0.977$ \\ \hline
    \small{COSH}&$0.536 $&$     0.687 $&$      \textbf{0.813} $&$      \textbf{0.898} $&$     \textbf{0.949} $&$      \textbf{0.978}$ \\ \hline
           \small{IS}&$0.513 $&$     0.670 $&$     0.805 $&$     0.896 $&$    0.948 $&$     0.977$ \\ \hline
        \small{IS-II}&$0.521 $&$     0.674 $&$     0.806 $&$     0.896 $&$    0.948 $&$     0.977$ \\ \hline
 \small{WCOSH}&$0.517 $&$     0.667 $&$     0.801 $&$     0.893 $&$    0.947 $&$     0.976$ \\ \hline
    \small{BNMF}&$0.549 $&$     0.688 $&$     0.804 $&$     0.878 $&$    0.917 $&$     0.936$ \\ \hline
    \small{LSA}&$0.522 $&$     0.661 $&$     0.784 $&$     0.874 $&$    0.932 $&$     0.966$ \\ \hline
   \small{WFIL}&$0.538 $&$     0.680 $&$     0.799 $&$     0.886 $&$    0.942 $&$     0.974$ \\ \hline
 \hline
\end{tabular}
\label{STOI_SCORES_TRAIN_NOISE}
\end{table}

\subsubsection{Train noise} Figure~\ref{figure:objective_scores_train_noise} compares the denoising performance in train noise. Similar to the street noise scenario, PROSE denoisers based on WCOSH, IS-II, and WE yield a higher SNR and SSNR gain than BNMF, WFIL and LSA (cf. Figures~\ref{figure:objective_scores_train_noise}(a) and (b)). Although the SNR and SSNR gain trends are similar to the street noise scenario, the margin of improvement is higher in train noise. This may be because street noise is comparatively more nonstationary than train noise, which may have resulted in less accurate estimates of the noise standard deviation. From Figure~\ref{figure:objective_scores_train_noise}(c), we observe that for input SNRs greater than $5$ dB, PROSE denoisers with WCOSH and IS-II measures are better than all the other methods. For input SNRs lower than $5$ dB, LSA gives a PESQ gain about 0.05 higher than the rest. PESQ gains are also higher in case of train noise than street noise. Table~\ref{STOI_SCORES_TRAIN_NOISE} shows the corresponding STOI scores and the trends are similar to the street noise scenario.\\
\indent To summarize, the PROSE denoisers based on WCOSH, IS-II, IS, and WE show consistently superior denoising performance in terms of SNR and SSNR compared with LSA, WFIL, and BNMF. The margin of improvement over LSA, WFIL, and BNMF also increases with SNR. Within the PROSE family of denoisers, the PESQ gains are higher for perceptual measures than MSE. For highly nonstationary noise types such as the street noise, the BNMF technique is marginally better at low SNRs, which is attributed to the training process. For SNRs below $0$ dB, the PESQ gain offered by PROSE is not significant and in the case of street noise, even negative. This is due to inaccuracy in estimating noise variance in highly nonstationary conditions such as street noise. For input SNRs greater than $5$ dB, the PROSE methodology and WFIL give consistently better STOI scores, unlike LSA and BNMF approaches. The entire repository of denoised speech files under various noise conditions is available online at: \href{http://spectrumee.wix.com/prose}{http://spectrumee.wix.com/prose}.
\subsection{Subjective Evaluation}
\indent We consider four speech files (two male and two female speakers) at SNRs 0, 10, and 20 dB. Fifteen listeners in the age group of $20-35$ years endowed with normal hearing were selected for the listening test. The subjects were given a Sennheiser HD650 headphone for listening, and were asked to rate the enhanced speech signal based on the ITU-T recommended P.835 scale~\cite{ITUscale}:\\
(i) Speech signal distortion (SIG); $1$: very distorted, $2$: fairly distorted, $3$: somewhat distorted, $4$: little distorted, $5$: not distorted; \\
(ii) Background intrusiveness (BAK); $1$:  very intrusive, $2$: somewhat intrusive, $3$:  noticeable but not intrusive, $4$:  somewhat noticeable, $5$:  not noticeable; and\\
(iii) Overall quality (OVRL); $1$: bad, $2$: poor, $3$: fair, $4$: good, $5$: excellent.\\
\indent The listening tests were conducted in three sessions, one for each noise type (white/street/train) at all SNRs. Following the ITU-T recommendation, each session is further divided into two subsessions. In the first one, two files (one male and one female speaker) are used with the rating order as SIG$-$BAK$-$OVRL and in the second one, the other two files (again one male and one female speaker) are presented and the order of rating is BAK$-$SIG$-$OVRL. The two-session test is done to suppress listener's bias. In each subsession, the listeners rate the denoising performance of all the algorithms. Thus, in one subsession, a listener had to grade a total of 2 (speakers) $\times 3$ (SNRs) $\times 11$ (algorithms) = 66 files. The subjects were given sufficient time to relax within and across subsessions, and across SNRs. The denoised speech files were presented in a random order, and the scores were derandomized accordingly before calculating the average scores. The listeners were given a token reward at the end of the listening test. Figure~\ref{fig:listening_test_white} shows the mean scores of SIG, BAK, and OVRL  for white noise, street noise, and train noise.\\
\indent We observe that, among the PROSE algorithms, WCOSH exhibits a consistently high denoising performance in terms of SIG,  BAK, and OVRL scores compared with the other algorithms. Among the techniques considered for benchmarking performance, LSA shows a consistently higher performance. The performances of LSA and PROSE with WCOSH are comparable. Also, PROSE consistently improves the BAK and OVRL scores compared with the noisy signal. Listening results reveal that, for all the algorithms considered, the extent of improvement in SIG scores is not significant compared with the improvement in BAK and OVRL scores. Within the PROSE family, WCOSH, IS-II, and IS show a high BAK and OVRL scores, whereas MSE results in a lower score, compared with the other algorithms. 
\begin{figure*}[t]
\centering
$
\begin{array}{cccc}
&\text{SNR = 0 dB} & \text{SNR = 10 dB} & \text{SNR = 20 dB}\\
\rot{\quad \quad \quad White noise}&\includegraphics[width=2.2in]{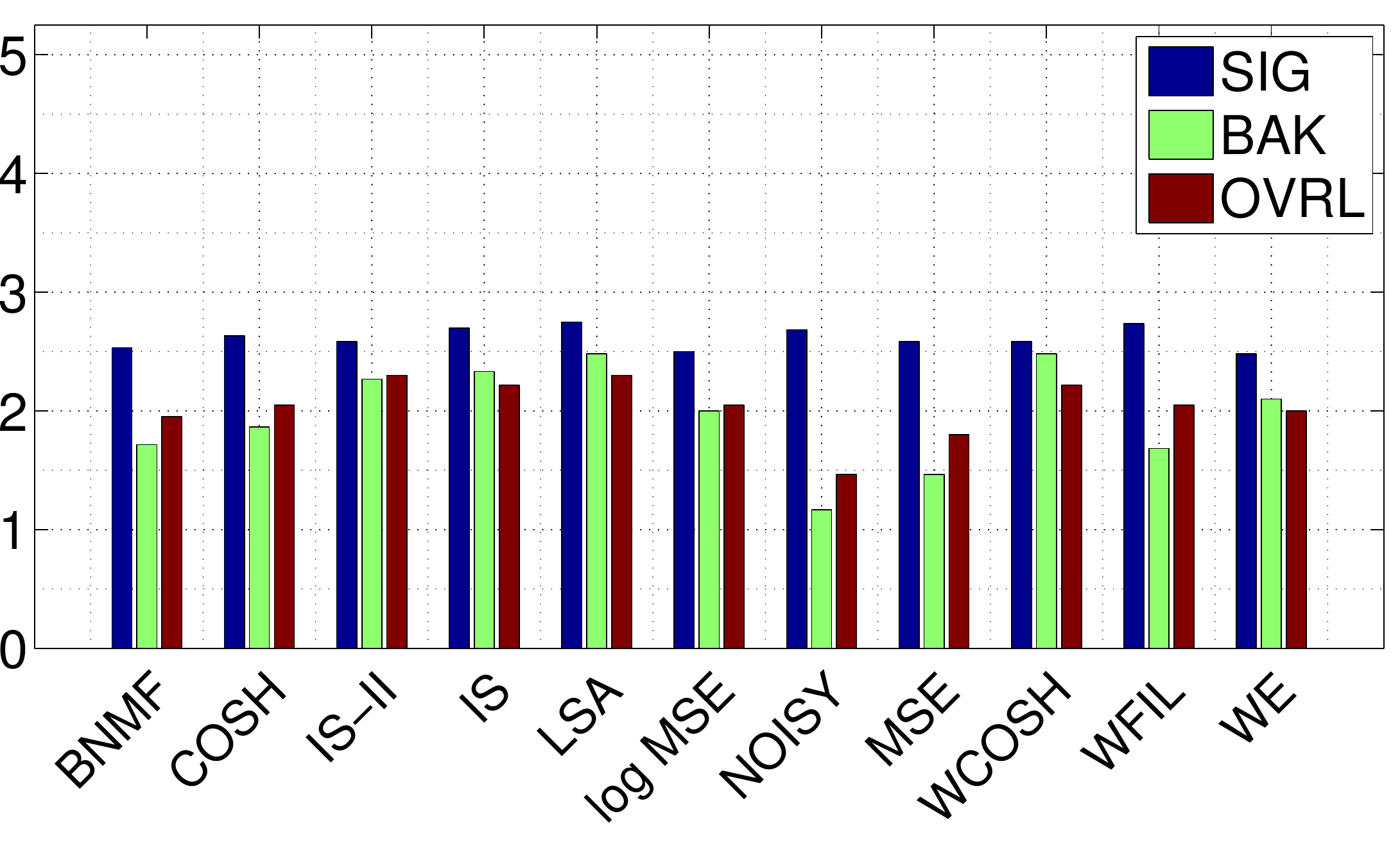} &
\includegraphics[width=2.2in]{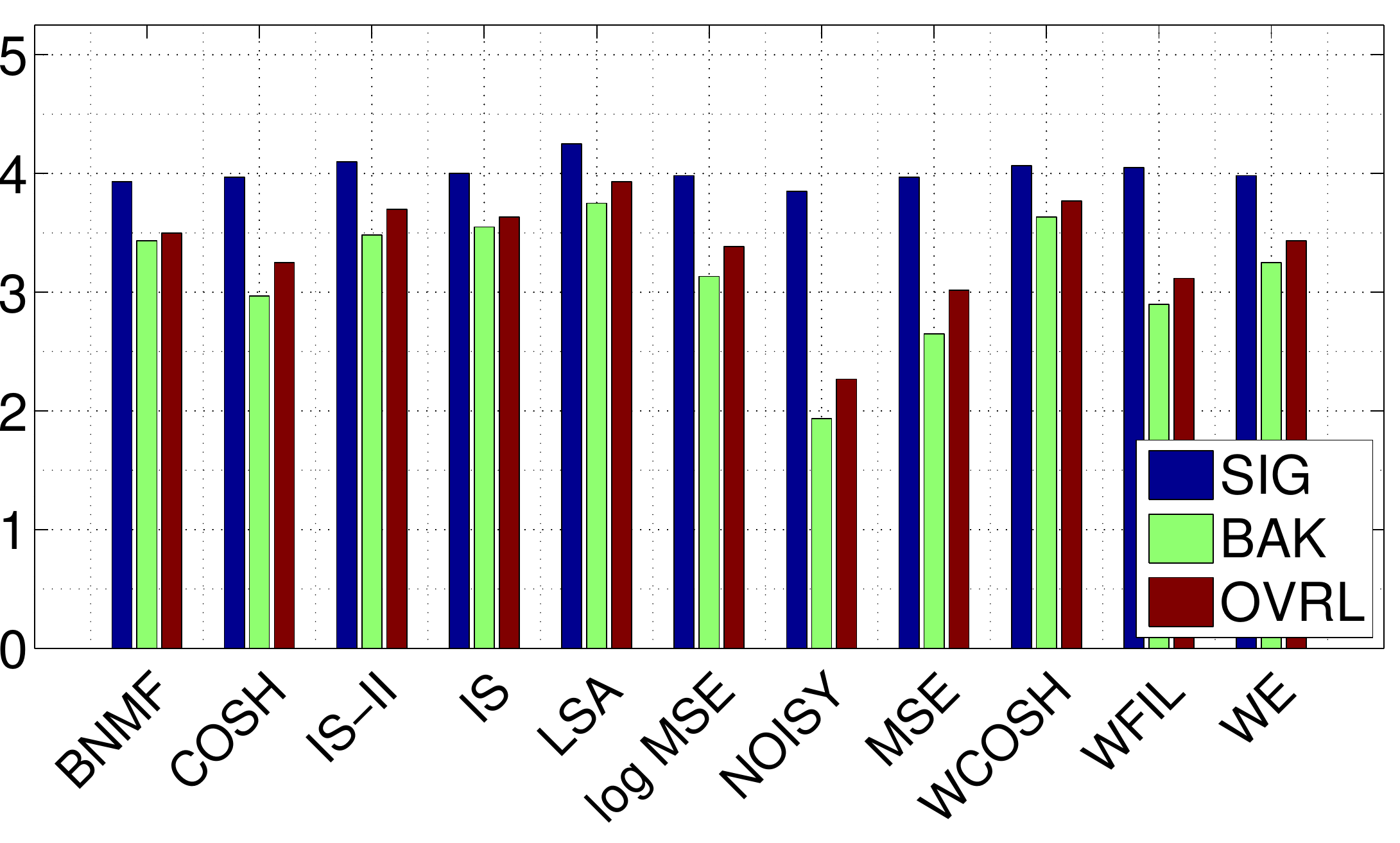} &
\includegraphics[width=2.2in]{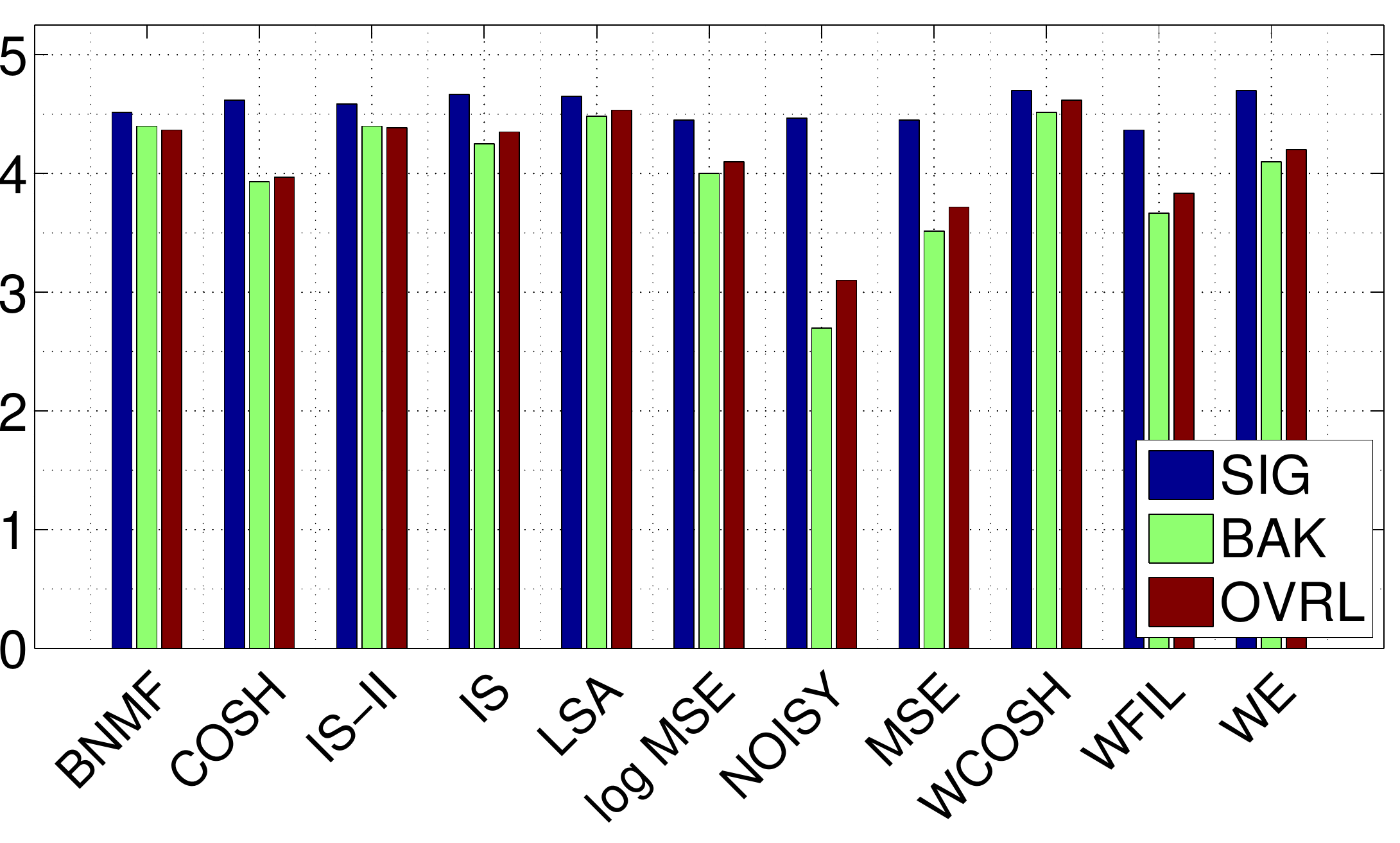}\\
\rot{\quad \quad \quad Street noise}&\includegraphics[width=2.2in]{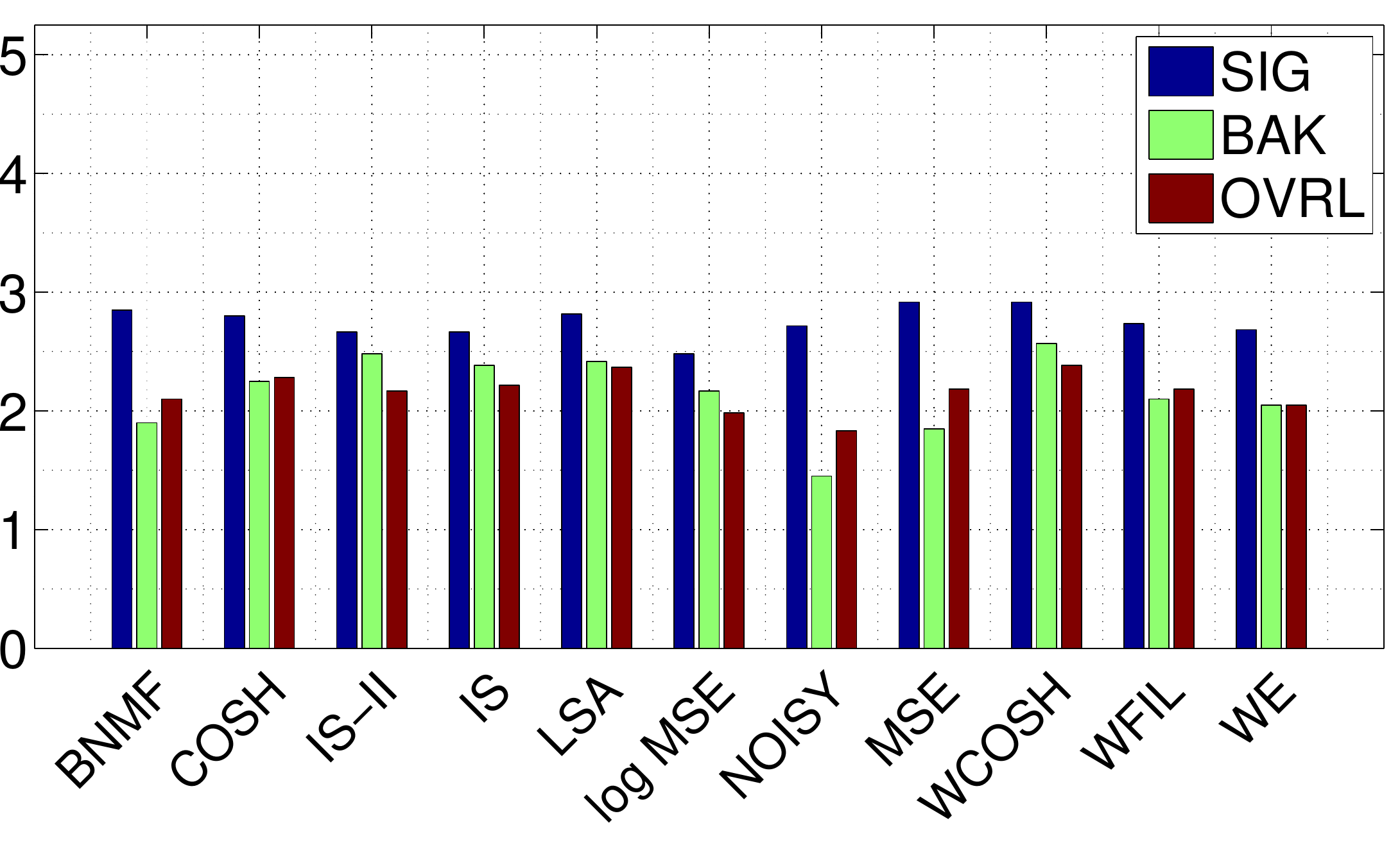} &
\includegraphics[width=2.2in]{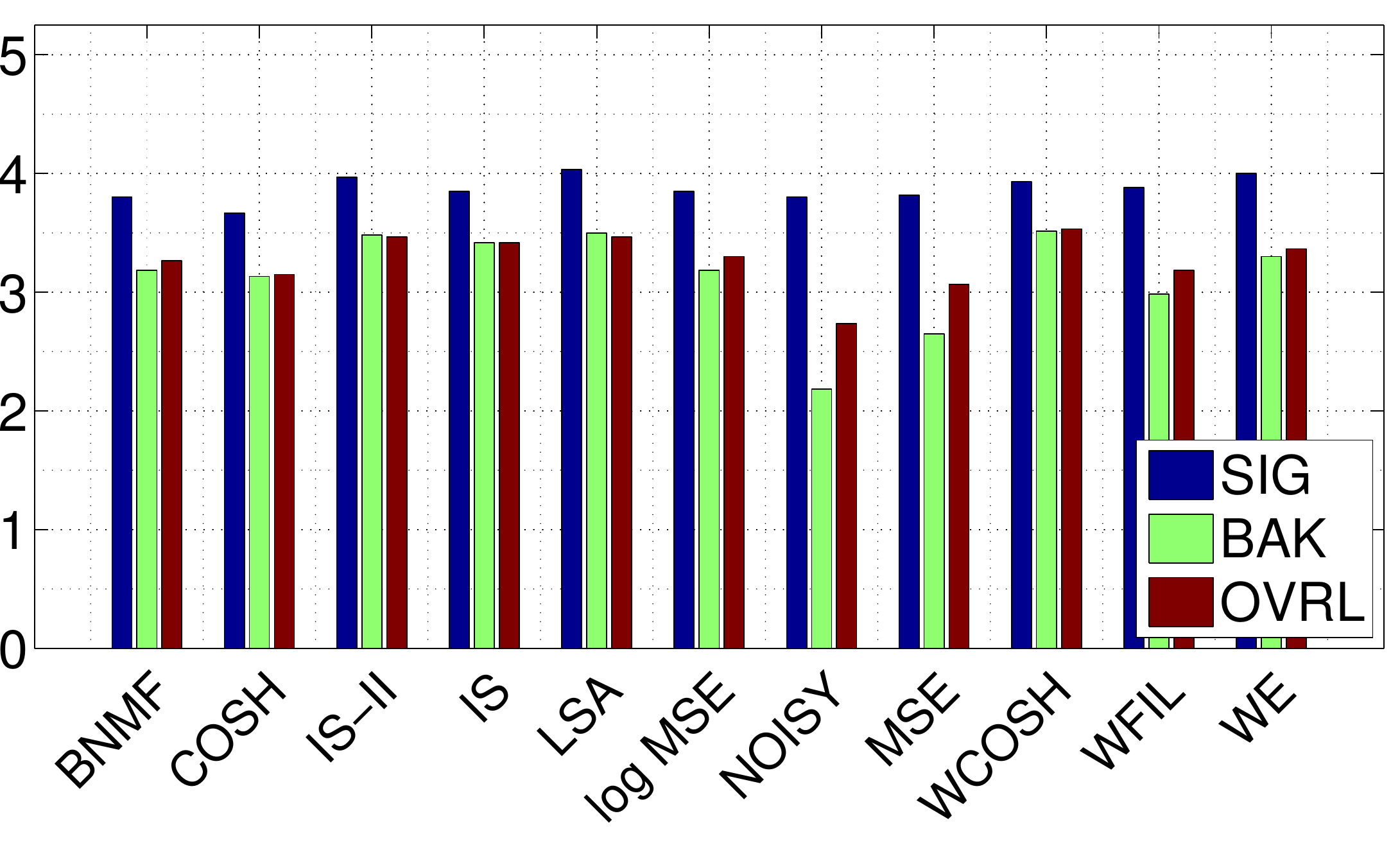} &
\includegraphics[width=2.2in]{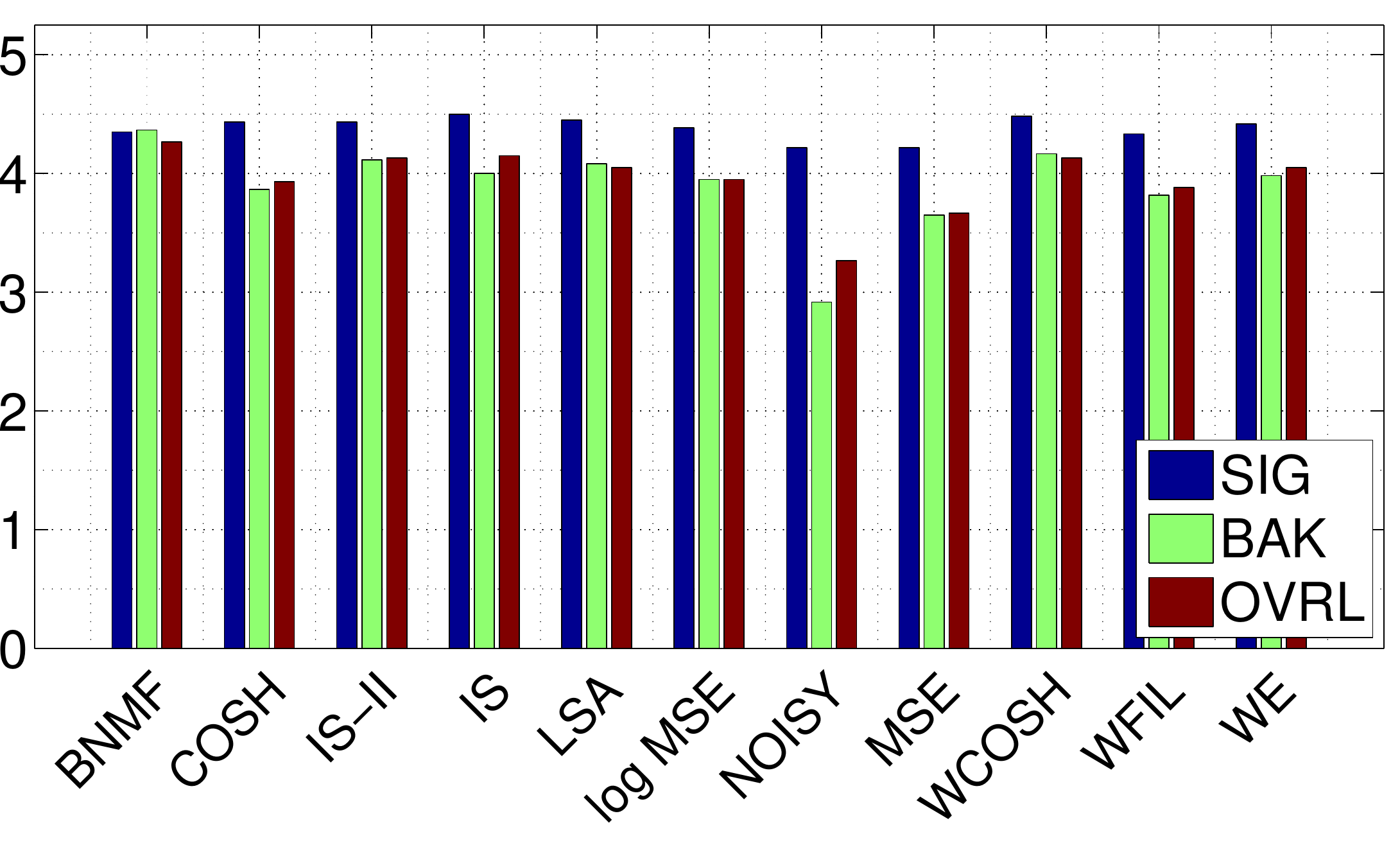}\\
\rot{\quad \quad \quad Train noise}&\includegraphics[width=2.2in]{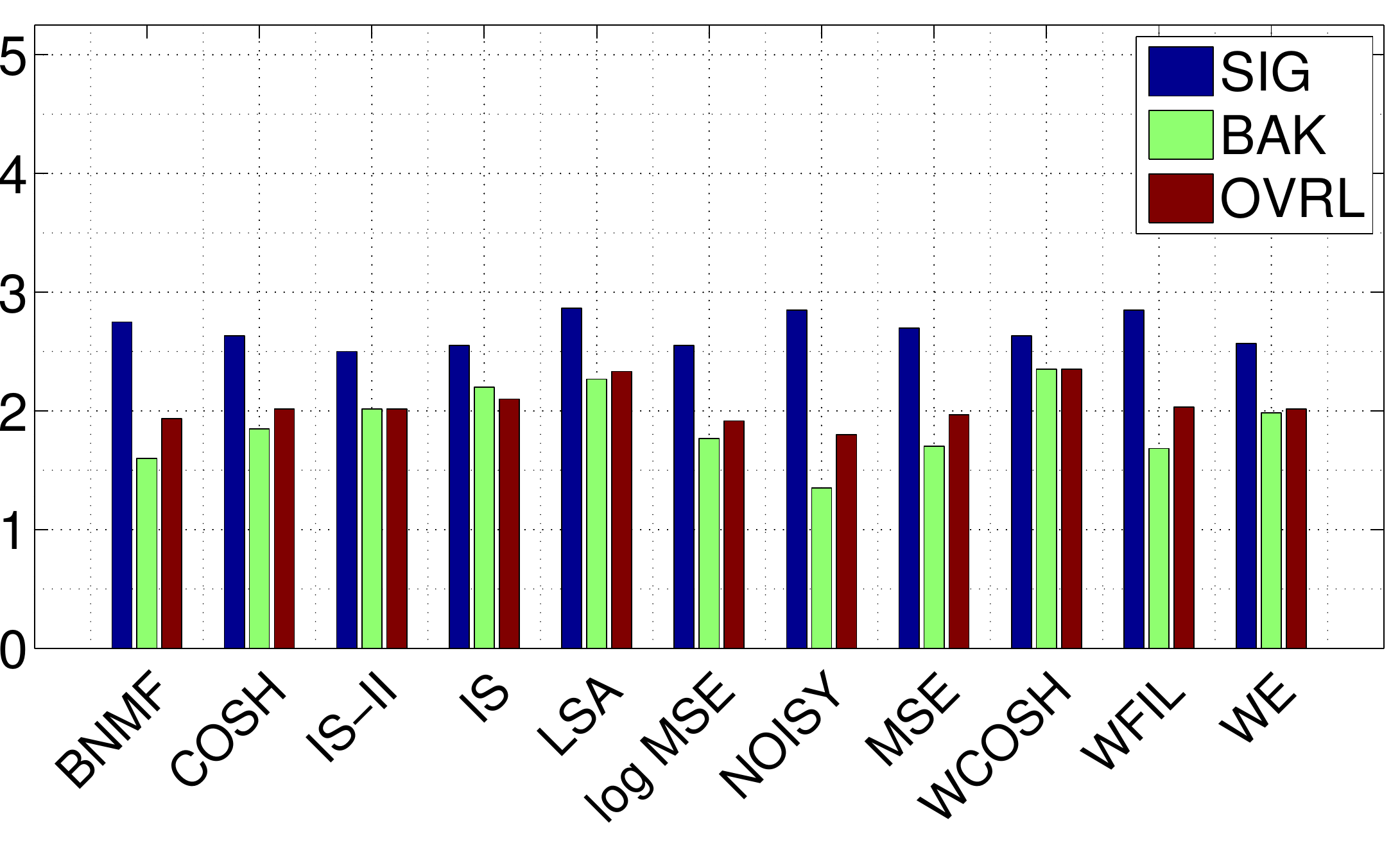} &
\includegraphics[width=2.2in]{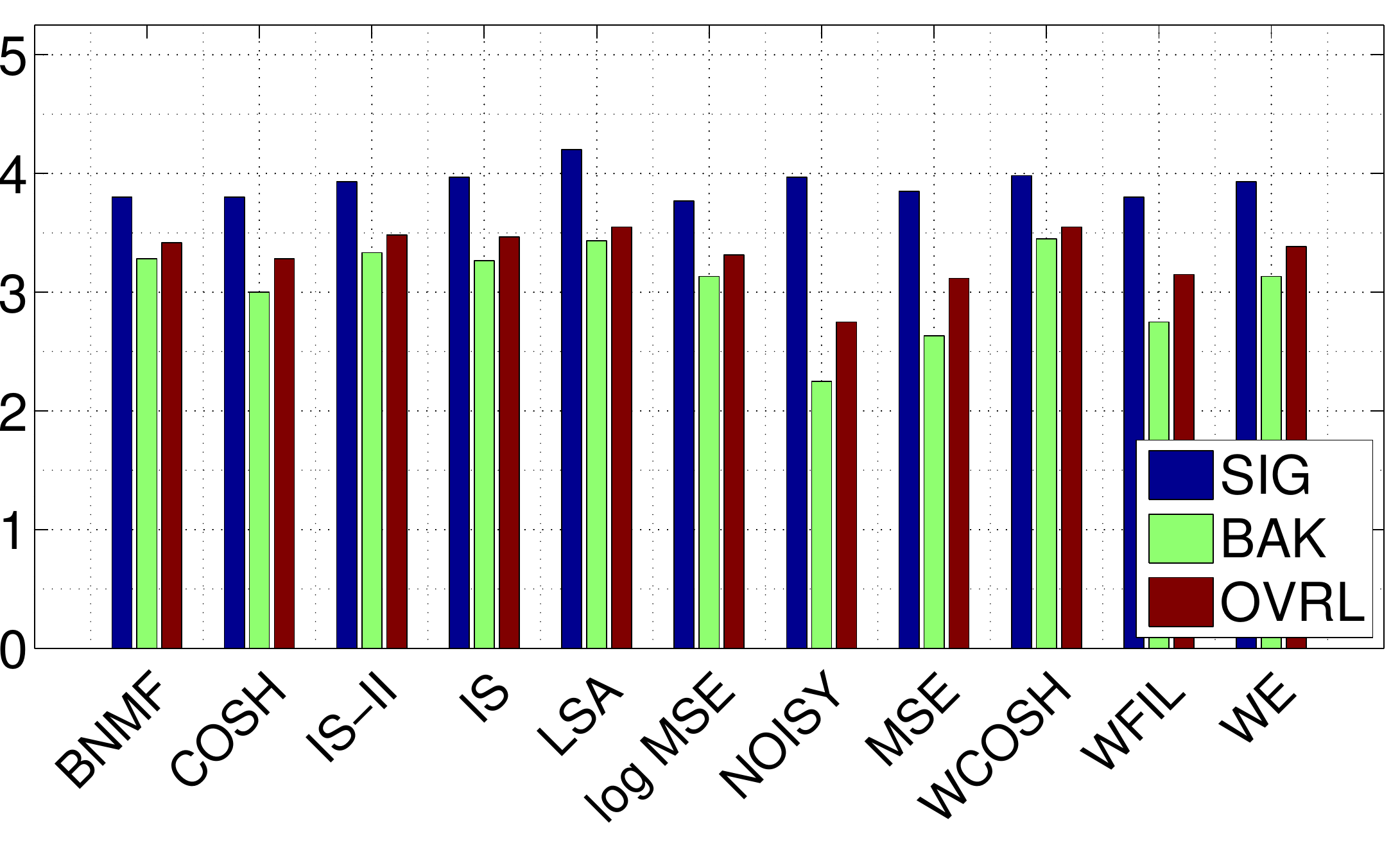} &
\includegraphics[width=2.2in]{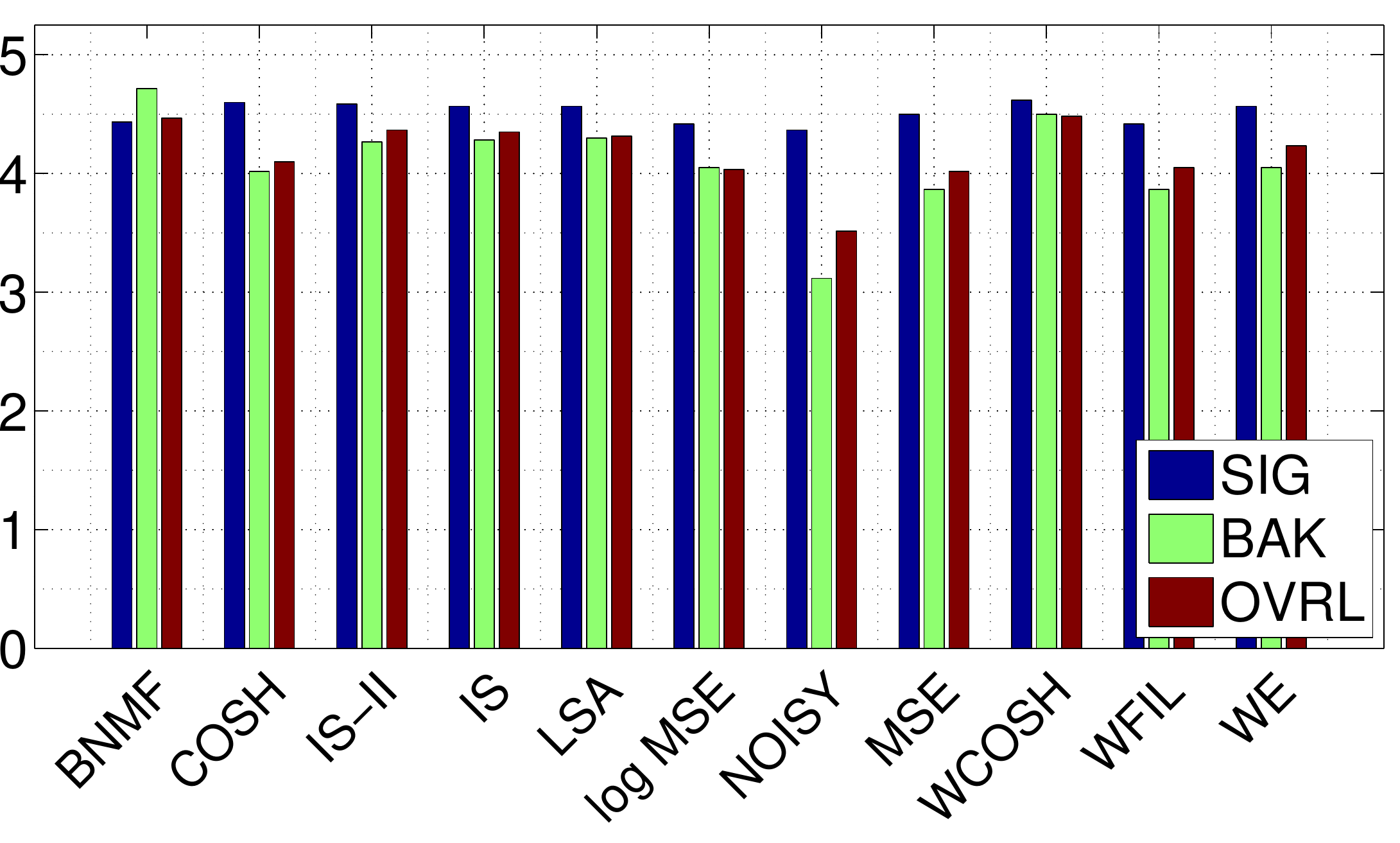}\\
\end{array}
$
\caption{[Color online] A comparison of the mean values of SIG, BAK, and OVRL ratings of the enhanced signal.}
\label{fig:listening_test_white}
\end{figure*}

\subsection{Spectrograms}
\indent Figure~\ref{fig:spectrograms_prose} shows the spectrograms of clean, noisy (SNR = 10 dB), and enhanced signals for the case of train noise. We observe that WFIL and LSA result in significant residual noise, whereas BNMF has relatively lower residual noise. BNMF recovers clean speech spectra with less distortion in some regions, compared with other algorithms (cf. green box), but it introduces distortions in the other parts (highlighted by the blue box for instance). PROSE with WCOSH results in superior noise suppression with minimal speech distortion than WFIL, LSA, and BNMF. Both PROSE and BNMF introduce a small amount of musical noise, in particular, in the silence regions. Regions that are submerged in noise are difficult to recover by any algorithm (red box, for instance), which explains why the improvements in SIG scores are lower than improvements in BAK and OVRL for all the techniques.

\begin{figure}[t]
\centering
$
\begin{array}{c}
\text{\small{Clean signal}} \\
\includegraphics[width=3.5in]{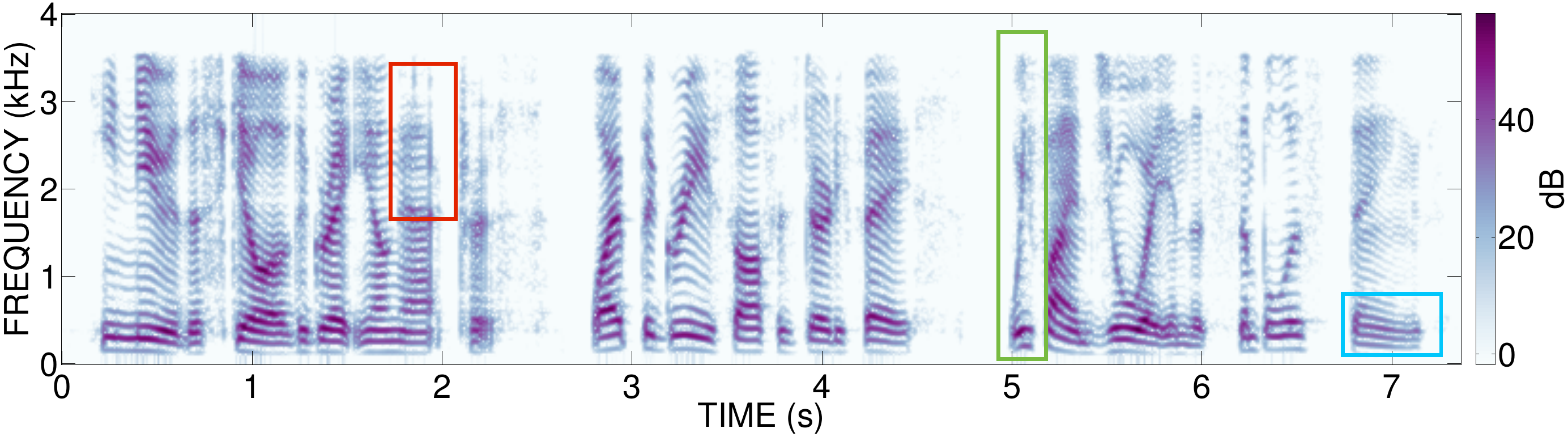}\\
\text{\small{Noisy speech}} \\
\includegraphics[width=3.5in]{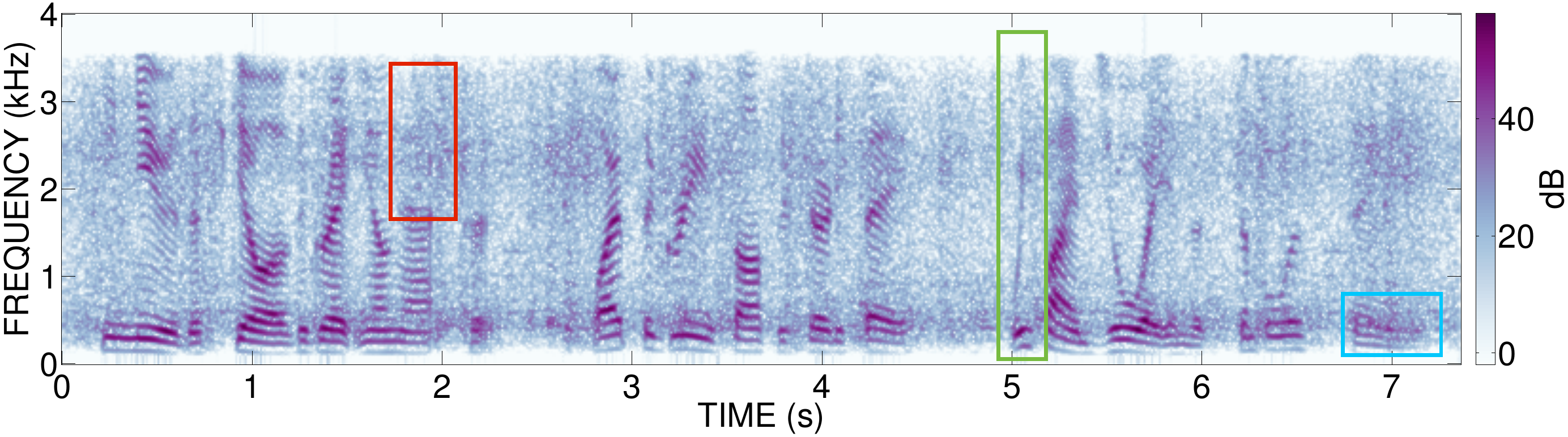} \\
\text{\small{PROSE (WCOSH)}}\\
\includegraphics[width=3.5in]{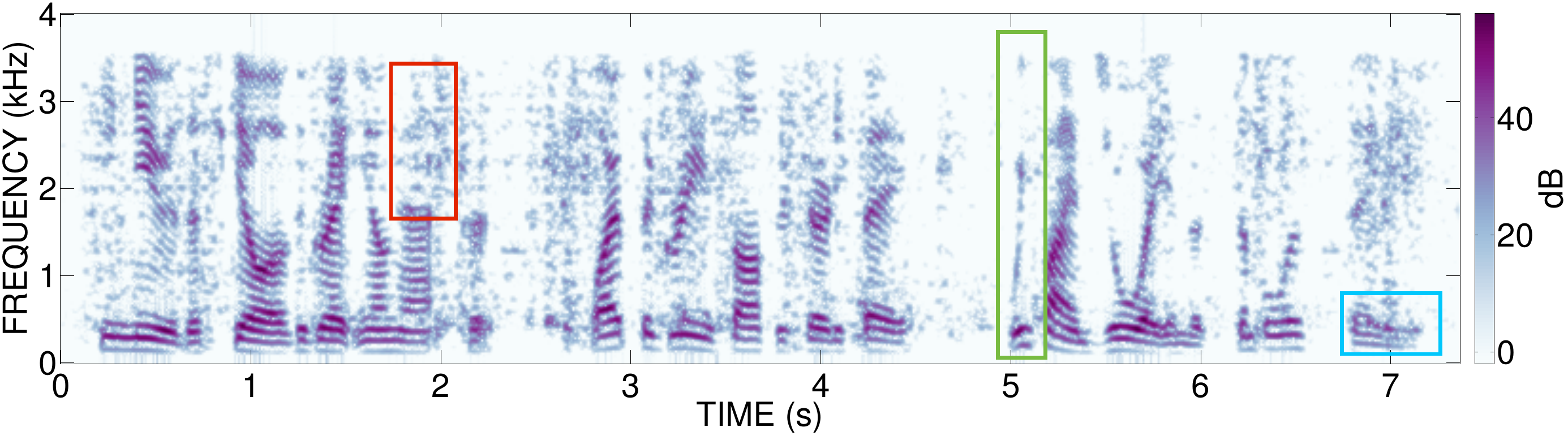}\\
\text{\small{LSA}} \\
\includegraphics[width=3.5in]{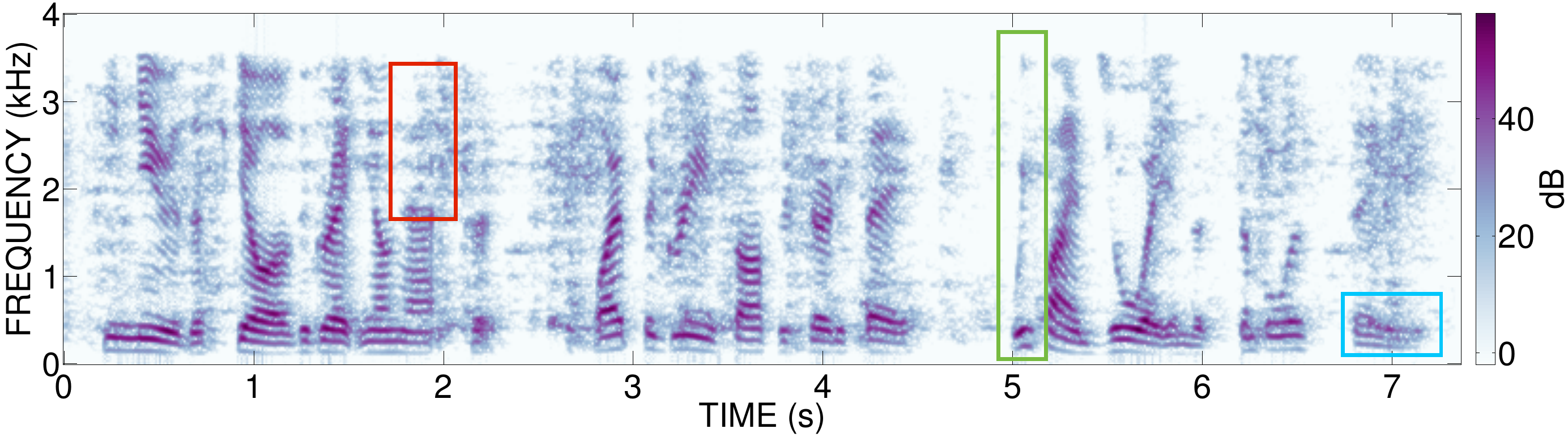}\\
\text{\small{WFIL}}\\
\includegraphics[width=3.5in]{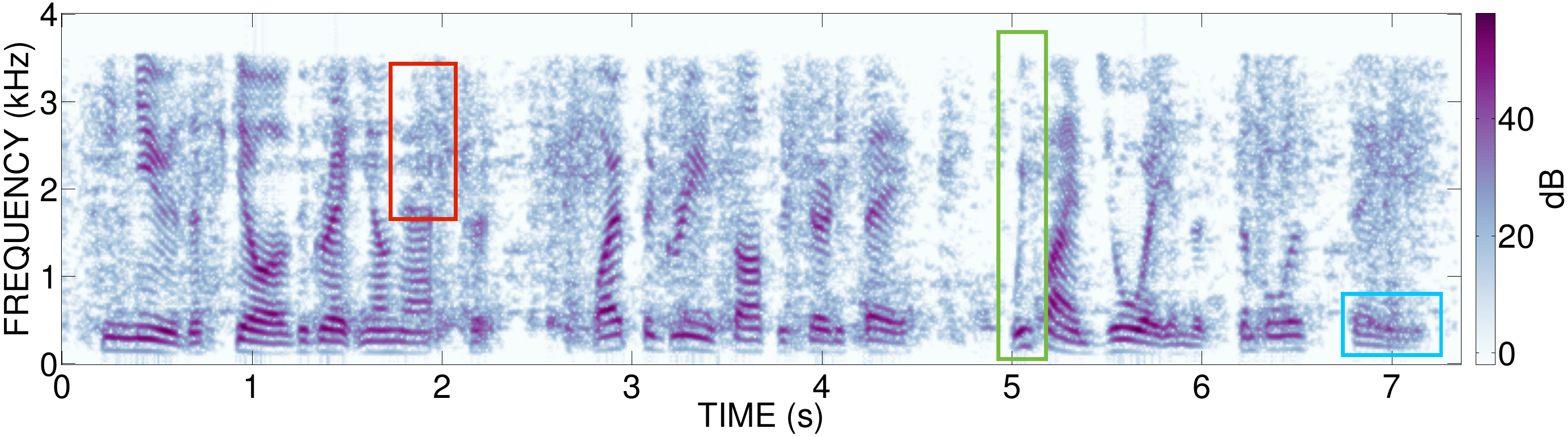}\\
\text{\small{BNMF}} \\
\includegraphics[width=3.5in]{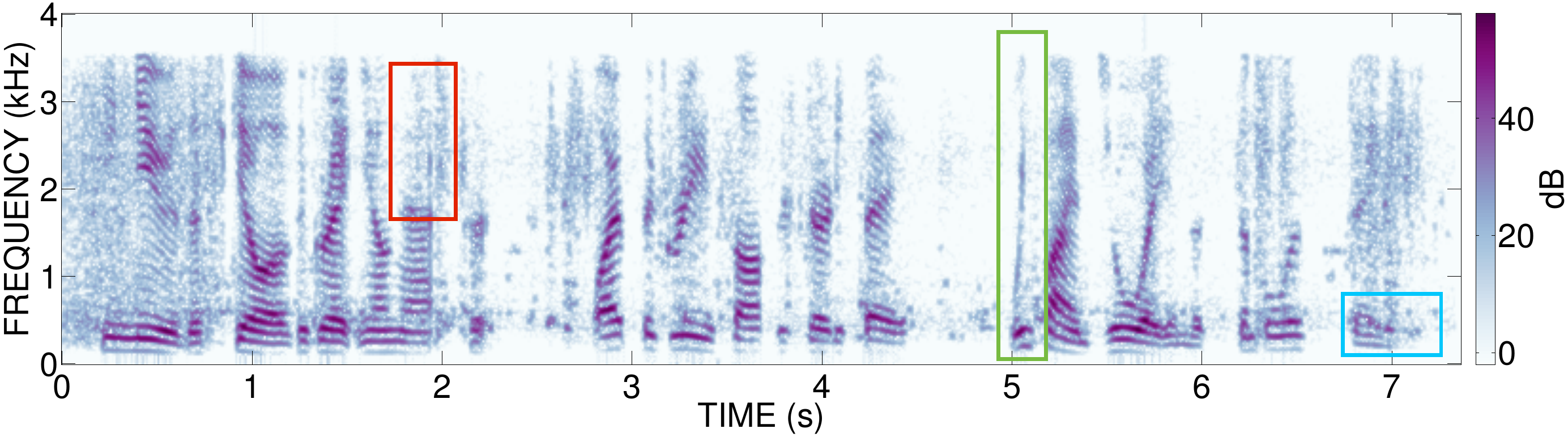}\\
\end{array}
$
\caption{[Color online] Spectrograms of clean speech, noisy speech, and speech denoised using various approaches. The utterance is: ``He knew the skill of the great young actress. Wipe the grease off his dirty face. We find joy in the simplest things.''}
\label{fig:spectrograms_prose}
\end{figure}

\section{Conclusions}
\label{sec:concl}
\indent We introduced the notion of unbiased risk estimation within a perceptual framework (abbreviated PROSE) for performing single-channel speech enhancement. The analytical developments are based on Stein's lemma and its generalized version introduced in this paper, which proved to be efficient for obtaining unbiased estimates of the distortion measures. We have also established the optimality of the shrinkage parameters considering the Karush-Kuhn-Tucker conditions. Validation on several speech signals in both synthesized noise and real-world nonstationary noise scenarios, and comparisons with benchmark techniques showed that, for input SNR greater than $5$ dB, the proposed PROSE method results in better denoising performance. Within the PROSE family, estimators based on Itakura-Saito distortion and weighted cosh distortion resulted in superior denoising performance. It must be emphasized that the PROSE methodology is relatively simpler from an implementation perspective (recall Fig.~\ref{fig:Block}) and does not require any training, making it ideal for deployment in practical applications involving hearing aids, mobile devices, etc. Further, we employed a voice-activity detector and a recursive algorithm for estimating the time-varying a posteriori SNR. Improving on the accuracy of noise variance estimation for handling rapidly changing noise would further enhance the performance of the PROSE method.
\section*{Acknowledgments}
We would like to thank Subhadip Mukherjee for technical discussions and the subjects for participating in the listening test. We would also like to thank Mohammadiha et al.~\cite{Mohammadiha}  and Taal et al.~\cite{stoi} for providing the Matlab implementations of their algorithms online, which facilitated the comparisons.
\begin{appendices}
\section {Karush-Kuhn-Tucker Conditions}
\label{KKT_conditions}
The goal is to solve the optimization problem $$\underset{a_k}{\operatorname{min}}\, \widehat{\mathcal{R}} \, \text{subject to}  \,  a_k \in \left[0,1\right].$$ The corresponding Lagrangian is $\mathbf{L}\left(a_k,\lambda_1,\lambda_2\right)=\widehat{\mathcal{R}}+\lambda_1 \left(a_k-1\right)-\lambda_2 a_k$, where $\lambda_1, \lambda_2 \in\mathbb{R}^+\cup\{0\}$. The KKT conditions are as follows:
\begin{subequations}
\begin{align}
 & {\tt C1} : \, \frac{\mathrm{d} \mathbf{L}\left(a_k,\lambda_1,\lambda_2\right)}{\mathrm{d} a_k}=0, \label{lagrange_derivative} \\
 & {\tt C2} : \, \lambda_1 (a_k-1)=0, -\lambda_2 a_k=0, \label{complimentary_slackness} \\
 & {\tt C3} : \,  a_k \in \left[0,1\right], \label{primal_feasibility}\\
 & {\tt C4} : \,  \lambda_1\geq0,  \lambda_2\geq0, \, \text{and}  \,  \label{dual_feasibility}\\
 & {\tt C5} : \, \frac{\mathrm{d}^2 \mathbf{L}\left(a_k,\lambda_1,\lambda_2\right)}{\mathrm{d} a_k^2}>0  \label{second_order}.
\end{align}
\end{subequations}
Solving \eqref{lagrange_derivative}$-$\eqref{second_order} gives the optimum shrinkage parameter $a_k$. The derivations for all the distortion measures are provided in the supporting document.
\section{The High-SNR Scenario}
\label{probability_high_snr_event}
The a priori SNR is defined as $\text{SNR} := \displaystyle\frac{{S_k}^2}{\sigma^2}$. By high SNR, we mean that $\text{SNR} > 4c^2$. Consider the probability:
\begin{align}
\nonumber
\text{Prob}\left\{\left|W_k\right|<\left|X_k\right| \right\}&=\text{Prob}\left\{{X_k}^2-{W_k}^2 >0\right\},\nonumber\\
&=\text{Prob}\left\{\right({X_k}+{W_k}\left) \right({X_k}-{W_k}\left) >0\right\},\nonumber\\
&=\text{Prob}\left\{S_k\right({X_k}+{W_k}\left) >0\right\},\nonumber\\
&=\text{Prob}\left\{{W_k} <-\frac{S_k}{2}\right\}, \text{if}\quad S_k < 0\nonumber\\
&=\text{Prob}\left\{{W_k} >-\frac{S_k}{2}\right\}, \text{if}\quad S_k > 0.
\label{prob_high_snr_event}
\end{align}
Therefore, $\text{Prob}\left\{\left|W_k\right|<\left|X_k\right|\right\} = \text{Prob}\left\{W_k<\frac{\left|S_k\right|}{2}\right\} = 1$ if $|S_k| > 2c\sigma \Rightarrow \text{a priori SNR} > 4c^2.$
\section{Weighted Euclidean Distance}
\label{app:WE}
\begin{multline}
 \text{Consider} \quad \mathcal{E}\left\{a_k^2X_k \displaystyle \sum_{n=0}^{4} \left(\displaystyle\frac{W_k}{X_k} \right)^n\right\}\\
 =a_k^2\mathcal{E}\left\{X_k+W_k+\displaystyle\frac{W_k^2}{X_k}+\displaystyle\frac{W_k^3}{X_k^2}+\displaystyle\frac{W_k^4}{X_k^3}\right\}. 
 \label{WEequation}
\end{multline}
Using \emph{Lemma} \ref{truncated_Gaussian_recursive_lemma2}, we have that
\begin{align*}
\mathcal{E}\left\{\displaystyle\frac{W_k^2}{X_k}\right\}
&=\sigma^4 \mathcal{E} \left \{ \displaystyle\frac{2}{X_k^3} \right\}+\sigma^2 \mathcal{E}\left\{\displaystyle\frac{1}{X_k}\right\},\\
\mathcal{E}\left\{\displaystyle\frac{W_k^3}{X_k^2}\right\}&=\sigma^6  \mathcal{E} \left\{ \displaystyle\frac{-24}{X_k^5} \right\} +\sigma^4 \mathcal{E} \left\{ \displaystyle\frac{-6}{X_k^3} \right\}, \text{and}\\
\mathcal{E}\left\{\displaystyle\frac{W_k^4}{X_k^3}\right\}&=\sigma^8 \mathcal{E} \left\{ \displaystyle\frac{360}{X_k^7} \right\}+\sigma^6 \mathcal{E} \left\{\displaystyle\frac{72}{X_k^5} \right\}+3\sigma^4 \mathcal{E}\left\{ \displaystyle\frac{1}{X_k^3} \right\}.
\end{align*}
Substituting the preceding expressions in (\ref{WEequation}), we get
\begin{multline}
 \mathcal{E}\left\{a_k^2X_k \displaystyle \sum_{n=0}^{4} \left(\displaystyle\frac{W_k}{X_k} \right)^n\right\}\\=a_k^2\mathcal{E} \left\{X_k+ \displaystyle\frac{\sigma^2}{X_k}-\displaystyle\frac{\sigma^4}{X_k^3}+48\displaystyle\frac{\sigma^6}{X_k^5}+360\displaystyle\frac{\sigma^8}{X_k^7} \right\}.\nonumber
\end{multline}

\section{Log Mean-Square Error}
\label{app:log}
\begin{equation}
\text{Consider}\,\mathcal{E}\left\{\displaystyle \sum_{n=1}^{4} \displaystyle\frac{\log a_kX_k}{n} \left(\displaystyle\frac{W_k}{X_k}\right)^n\right\}
=\mathcal{E}\left\{ \displaystyle \sum_{n=1}^{4} J_n(X_k)W_k^n \right\}, \nonumber
\end{equation}
where $J_n(X_k)=\log a_kX_k/(nX_k^n)$. Applying \emph{Lemma} \ref{truncated_Gaussian_recursive_lemma2} gives
\begin{align*}
&\mathcal{E}\{J_1(X_k)W_k\}=\sigma^2\mathcal{E}\{J^{(1)}_1(X_k)\}, \\
&\mathcal{E}\{J_2(X_k)W_k^2\}=\sigma^4\mathcal{E}\{J^{(2)}_2(X_k)\}+\sigma^2\mathcal{E}\{J_2(X_k)\},\\
&\mathcal{E}\{J_3(X_k)W_k^3\}=\sigma^6\mathcal{E}\{J^{(3)}_3(X_k)\}+3\sigma^4\mathcal{E}\{J^{(1)}_3(X_k)\}, \mbox{and}\\
&\mathcal{E}\{J_4(X_k)W_k^4\}=\mathcal{E}\{\sigma^8 J^{(4)}_4(X_k)+6\sigma^6 J^{(2)}_4(X_k) +3\sigma^4 J_4(X_k)\},\\
&\mbox{where}\, J^{(1)}_1=\displaystyle\frac{1}{X_k^2}-\displaystyle\frac{\log a_kX_k}{X_k^2}, J^{(2)}_2=\displaystyle\frac{-5}{2X_k^4}+\displaystyle\frac{3\log a_kX_k}{X_k^4},\\
&J^{(1)}_3=\displaystyle\frac{1}{3X_k^4}-\displaystyle\frac{\log a_kX_k}{X_k^4}, J^{(3)}_3=\displaystyle\frac{47}{3X_k^6}-\displaystyle\frac{20\log a_kX_k}{X_k^6},\\
&J^{(2)}_4=\displaystyle\frac{-9}{4X_k^6}+\displaystyle\frac{5\log a_kX_k}{X_k^6}, J^{(4)}_4=-\displaystyle\frac{638}{4X_k^8}+\displaystyle\frac{210\log a_kX_k}{X_k^8}.
\end{align*}

\section{Itakura-Saito Distortion}
\label{app:IS}
\indent With reference to (\ref{RiskIS}), consider the truncated approximation:
\begin{equation}
\displaystyle \sum_{n=0}^{\infty} {\cal E}\{a_kW_k^n/X_k^n\}\approx\displaystyle \sum_{n=0}^{4} {\cal E}\{a_kW_k^n/X_k^n\}.
\label{summation}
\end{equation}
Applying \emph{Lemma} \ref{truncated_Gaussian_recursive_lemma2}, we have that
\begin{align*}
&\mathcal{E}\left\{\displaystyle\frac{W_k}{X_k}\right\}=\sigma^2\mathcal{E}\left\{\displaystyle\frac{-1}{X_k^2}\right\},
\mathcal{E}\left\{\displaystyle\frac{W_k^2}{X_k^2}\right\}=\mathcal{E}\left\{\sigma^4\displaystyle\frac{6}{X_k^4}+\sigma^2\displaystyle\frac{1}{X_k^2}\right\},\\
&\mathcal{E}\left\{\displaystyle\frac{W_k^3}{X_k^3}\right\}=\sigma^6\mathcal{E}\left\{\displaystyle\frac{-60}{X_k^6}\right\}+3\sigma^4\mathcal{E}\left\{\displaystyle\frac{-3}{X_k^4}\right\},\text{and}\\
&\mathcal{E}\left\{\displaystyle\frac{W_k^4}{X_k^4}\right\}=\sigma^8\mathcal{E}\left\{\displaystyle\frac{840}{X_k^8}\right\}+6\sigma^6\mathcal{E}\left\{\displaystyle\frac{20}{X_k^6}\right\}+3\sigma^4\mathcal{E}\left\{\displaystyle\frac{1}{X_k^4}\right\}.
\end{align*}
The final expression for (\ref{summation}) is given by
\begin{equation}
\displaystyle \sum_{n=0}^{4} {\cal E}\{a_kW_k^n/X_k^n\} =a_k \mathcal{E} \left\{840 \displaystyle\frac{\sigma^8}{X_k^8}+60 \displaystyle\frac{\sigma^6}{X_k^6} +1 \right\}.
\end{equation}
\section{Weighted COSH Distance}
\label{app:wcosh}
\indent We provide certain simplifications for the expectation terms in the risk estimator for weighted cosh measure:
\begin{equation}
\label{expectation}
\mathcal{R}=\mathcal{E}\left\{d(S_k,\widehat{S}_k) \right\} = \displaystyle\frac{1}{2} \mathcal{E}\left \{ \displaystyle\frac{1}{\widehat{S}_k} + \displaystyle\frac{\widehat{S}_k}{S_k^2} \right\} - \displaystyle\frac{1}{S_k}.
\end{equation}
The second term in (\ref{expectation}) is approximated as
\begin{equation}
\displaystyle\frac{\widehat{S}_k}{X_k^2} \left ( 1-\displaystyle\frac{W_k}{X_k} \right )^{-2} \approx \displaystyle\frac{\widehat{S}_k}{X_k^2} \left( 1+2\displaystyle\frac{W_k}{X_k}+3\displaystyle\frac{W_k^2}{X_k^2}+4\displaystyle\frac{W_k^3}{X_k^3}+5\displaystyle\frac{W_k^4}{X_k^4} \right). \nonumber
 \end{equation}
Substituting ${\widehat S}_k=a_kX_k$ and taking expectation, we get that
\begin{equation}
\mathcal{E} \left\{ \displaystyle\frac{\widehat{S}_k}{S_k^2} \right\} = \mathcal{E}\left\{ \displaystyle\frac{a_kX_k}{X_k^2} \left( 1+2\displaystyle\frac{W_k}{X_k}+3\displaystyle\frac{W_k^2}{X_k^2}+4\displaystyle\frac{W_k^3}{X_k^3}+5\displaystyle\frac{W_k^4}{X_k^4} \right) \right\}. \nonumber
\end{equation}
Simplified expressions for the individual terms in the above equation are given below:
\begin{align*}
&\mathcal{E} \left\{ \displaystyle\frac{W_k}{X_k^2} \right\} =\sigma^2 \mathcal{E} \left\{ \displaystyle\frac{-2}{X_k^3} \right\}, \mathcal{E} \left\{ \displaystyle\frac{W_k^2}{X_k^3} \right\}= \mathcal{E} \left\{\sigma^4 \displaystyle\frac{12}{X_k^5} +\sigma^2 \displaystyle\frac{1}{X_k^3} \right\},\\
&\mathcal{E} \left\{ \displaystyle\frac{W_k^3}{X_k^4} \right\}=\sigma^6 \mathcal{E} \left\{ \displaystyle\frac{-120}{X_k^7} \right\}+3\sigma^4 \mathcal{E} \left\{ \displaystyle\frac{-4}{X_k^5} \right\},\\
&\mathcal{E} \left\{ \displaystyle\frac{W_k^4}{X_k^5} \right\}=\sigma^8 \mathcal{E} \left\{ \displaystyle\frac{1680}{X_k^9} \right\}+6\sigma^6 \mathcal{E} \left\{ \displaystyle\frac{30}{X_k^7} \right\}+3\sigma^4 \mathcal{E} \left\{ \displaystyle\frac{1}{X_k^5} \right\}.
\end{align*}
Substituting these expressions in (\ref{expectation}), we get
\begin{align}
\nonumber
\mathcal{R}= \mathcal{E}  \left\{ \frac{a_k}{2X_k} \left( 1-\frac{\sigma^2}{X_k^2}+3 \frac{\sigma^4}{X_k^4}+420\frac{\sigma^6}{X_k^6}+8400\frac{\sigma^8}{X_k^8} \right)\right.\nonumber\\
+\left.\frac{1}{2a_kX_k} -\frac{1}{S_k}\right\}.
\end{align}
The quantity inside the braces is an unbiased estimate of $\mathcal{R}$.
\end{appendices}

\newpage

\onecolumn
\section*{\large{Karush-Kuhn-Tucker Conditions and Optimization of Perceptual Distortion Measures} \\ Supporting Document}
In this supporting document, we provide the detailed derivations for obtaining the optimum shrinkage factors for various distortion measures considered in the main document. The section and equation references correspond to those in the main document.

\subsection{The Optimization Problem and KKT Conditions}
\indent The goal is to solve
\begin{equation}
\underset{a_k}{\operatorname{min}}\, \widehat{\mathcal{R}} \, \text{subject to}  \,  a_k \in \left[0,1\right],
\label{optimization_problem}
\end{equation} 
where $\widehat{\mathcal{R}}$ is a chosen risk/distortion measure. The unconstrained form using Lagrangian is given by
$\mathbf{L}\left(a_k,\lambda_1,\lambda_2\right)=\widehat{\mathcal{R}}+\lambda_1 \left(a_k-1\right)-\lambda_2 a_k,$ 
 where $\lambda_1\in\mathbf{R}$ and $\lambda_2\in\mathbf{R}$ .
In order to obtain the optimum  $a_k$, we have to solve the Karush-Kuhn-Tucker (KKT) conditions given by
\begin{subequations}
\begin{align}
 & {\tt C1} : \, \frac{\mathrm{d} \mathbf{L}\left(a_k,\lambda_1,\lambda_2\right)}{\mathrm{d} a_k}=0, \label{lagrange_derivative} \\
 & {\tt C2} :  \, \lambda_1 (a_k-1)=0, -\lambda_2 a_k=0, \label{complimentary_slackness} \\
 & {\tt C3} :  \,  a_k \in \left[0,1\right], \label{primal_feasibility}\\
 & {\tt C4} :  \,  \lambda_1\geq0,  \lambda_2\geq0, \, \text{and}  \,  \label{dual_feasibility}\\
 & {\tt C5} : \, \frac{\mathrm{d}^2 \mathbf{L}\left(a_k,\lambda_1,\lambda_2\right)}{\mathrm{d} a_k^2}>0  \label{second_order}.
\end{align}
\end{subequations}

\subsection{Mean-squared error (MSE)}
In this case, recall from (cf. Eq. (5) in Section IV-A) that 
$${\widehat {\cal R}} = a_k^2 X_k^2-2a_kX_k^2+2\sigma^2a_k +S_k^2.$$ 
The corresponding Lagrangian  is
\begin{equation}
\mathbf{L}\left(a_k,\lambda_1,\lambda_2\right) = a_k^2 X_k^2-2a_kX_k^2+2\sigma^2a_k +S_k^2+\lambda_1 \left(a_k-1\right)-\lambda_2 a_k.\nonumber
\end{equation}
Our goal is to determine $a_k$ that satisfies the KKT conditions \eqref {lagrange_derivative} to \eqref{second_order}. The first- and second-order  derivatives of the Lagrangian are  
$$\displaystyle \frac{\mathrm{d} \mathbf{L}}{\mathrm{d} a_k}= 2a_k X_k^2-2X_k^2+2\sigma^2 +\lambda_1-\lambda_2,$$ and  
$$\displaystyle \frac{\mathrm{d}^2 \mathbf{L}}{\mathrm{d} a_k^2}=2X_k^2,$$ respectively. In order to solve \eqref {lagrange_derivative} -- \eqref{second_order} for $a_k$, we consider different cases of $\lambda_1$ and $\lambda_2$. 
\begin{itemize}
\item \textbf{Case 1}: $\lambda_1=0$, and $\lambda_2>0$\\
\eqref{complimentary_slackness} and   \eqref{primal_feasibility}  $\Rightarrow \, a_k=0$, and from  \eqref{lagrange_derivative}  we obtain
\begin{align}
\nonumber
\frac{\mathrm{d} \mathbf{L}}{\mathrm{d} a_k}\Bigg|_{a_k=0}&= -2X_k^2+2\sigma^2 -\lambda_2=0,\nonumber\\ 
&\Rightarrow \, \sigma^2-X_k^2 = \frac{\lambda_2}{2}  \overset{\eqref{dual_feasibility}}{>} 0.\nonumber\\
&\Rightarrow \, \sigma^2>X_k^2 \Rightarrow \, a_k=0.
\end{align}
\item \textbf{Case 2}: $\lambda_1>0$, and $\lambda_2=0$\\
\eqref{complimentary_slackness} and  \eqref{primal_feasibility} $\Rightarrow \, a_k=1$, and from \eqref{lagrange_derivative}  we obtain
\begin{align}
\nonumber
\label{mse_Lagrang_der_case2}
\frac{\mathrm{d} \mathbf{L}}{\mathrm{d} a_k}\Bigg|_{a_k=1}&= 2\sigma^2 +\lambda_1=0,\\ 
&\Rightarrow \, -\sigma^2 = \frac{\lambda_1}{2}. 
\end{align}
\vspace{-1ex}
Since $\sigma^2\geq0$ and $\lambda_1>0$ by assumption, \eqref{mse_Lagrang_der_case2} contradicts \eqref{dual_feasibility}, implying that $a_k=1$ is not a solution.  
\item \textbf{Case 3}: $\lambda_1>0$, and $\lambda_2>0$\\
\eqref{complimentary_slackness} and \eqref{primal_feasibility} imply that no solution exists.
 \item \textbf{Case 4}: $\lambda_1=0$, and $\lambda_2=0$\\
In this case, for $a_k \in \left[0,1\right]$, \eqref{complimentary_slackness} is satisfied. Using  \eqref{lagrange_derivative}, one can obtain $a_k$ as follows:
\begin{align}
\nonumber
\frac{\mathrm{d} \mathbf{L}}{\mathrm{d} a_k}&=2a_k X_k^2-2X_k^2+2\sigma^2=0,\\ 
&\Rightarrow \,a_k = 1-\frac{\sigma^2}{X_k^2}.
\end{align}
\end{itemize}
Consolidating  \textbf{Case 1} to  \textbf{Case 4}, we obtain that 
$$a_k=\displaystyle \max\left\{0,  1-\frac{\sigma^2}{X_k^2}\right\}.$$
Also observe that $ \displaystyle \frac{\mathrm{d}^2 \mathbf{L}}{\mathrm{d} a_k^2}=2X_k^2>0$, which implies that $a_k$ is the unique minimizer of $\widehat{\mathcal{R}}$ over $\left[0,1\right]$.
\vspace{-5ex}
\subsection{Weighted Euclidean (WE) distortion}
\vspace{-2ex}
In this case, the goal is to obtain $a_k \in \left[0,1\right]$ that minimizes  ${\widehat {\cal R}}$ defined in  (cf. Eq. (8) in Section IV-B) if $S_k>0$ and maximizes it if $S_k<0$. The high-SNR assumption led to $S_k$ and $X_k$ having the same sign (cf. Section IV-B). First, we consider the case $S_k, X_k > 0$. To obtain the optimum $a_k$, we solve \eqref {lagrange_derivative}$-$\eqref{second_order}, where
\begin{align}
&{\widehat {\cal R}}=a_k^2\left(X_k+ \frac{\sigma^2}{X_k}-\frac{\sigma^4}{X_k^3}+48\frac{\sigma^6}{X_k^5}+360\frac{\sigma^8}{X_k^7} \right)-2a_kX_k+S_k, \nonumber\\
&\mathbf{L}\left(a_k,\lambda_1,\lambda_2\right) = a_k^2\left(X_k+ \frac{\sigma^2}{X_k}-\frac{\sigma^4}{X_k^3}+48\frac{\sigma^6}{X_k^5}+360\frac{\sigma^8}{X_k^7} \right)-2a_kX_k+S_k+\lambda_1 \left(a_k-1\right)-\lambda_2 a_k,\nonumber\\
&\frac{\mathrm{d} \mathbf{L}}{\mathrm{d} a_k}  = 2a_k\left(X_k+ \frac{\sigma^2}{X_k}-\frac{\sigma^4}{X_k^3}+48\frac{\sigma^6}{X_k^5}+360\frac{\sigma^8}{X_k^7} \right)-2X_k+\lambda_1-\lambda_2, \nonumber \text{and}\\
& \frac{\mathrm{d}^2 \mathbf{L}}{\mathrm{d} a_k^2} = 2X_k\left(1+\displaystyle\frac{\sigma^2}{X_k^2}-\displaystyle\frac{\sigma^4}{X_k^4}+48\displaystyle\frac{\sigma^6}{X_k^6}+360\displaystyle\frac{\sigma^8}{X_k^8} \right)\nonumber.
\end{align}
Again, we consider four cases of  $\lambda_1$ and $\lambda_2$.
\begin{itemize}
 \item \textbf{Case 1}: $\lambda_1=0$, and $\lambda_2>0$\\
\eqref{complimentary_slackness} and   \eqref{primal_feasibility}  $\Rightarrow \, a_k=0$, and from  \eqref{lagrange_derivative}  we obtain
\begin{align}
\nonumber
\frac{\mathrm{d} \mathbf{L}}{\mathrm{d} a_k}\Bigg|_{a_k=0}&= -2X_k-\lambda_2=0,\\ 
&\Rightarrow \, -X_k = \frac{\lambda_2}{2}. \label{we_Lagrang_der_case1}
\end{align}
Since $X_k > 0$,  \eqref{we_Lagrang_der_case1} contradicts \eqref{dual_feasibility}, implies $a_k=0$ is not a solution. 
\item \textbf{Case 2}: $\lambda_1>0$, and $\lambda_2=0$\\
\eqref{complimentary_slackness} and  \eqref{primal_feasibility} $\Rightarrow \, a_k=1$, and from \eqref{lagrange_derivative}  we obtain
\begin{align}
\nonumber
\frac{\mathrm{d} \mathbf{L}}{\mathrm{d} a_k}\Bigg|_{a_k=1}&= 2X_k\left(1+\displaystyle\frac{\sigma^2}{X_k^2}-\displaystyle\frac{\sigma^4}{X_k^4}+48\displaystyle\frac{\sigma^6}{X_k^6}+360\displaystyle\frac{\sigma^8}{X_k^8} \right)-2X_k+\lambda_1=0,\\ 
&\Rightarrow \, 1-\left(1+\displaystyle\frac{\sigma^2}{X_k^2}-\displaystyle\frac{\sigma^4}{X_k^4}+48\displaystyle\frac{\sigma^6}{X_k^6}+360\displaystyle\frac{\sigma^8}{X_k^8} \right) = \frac{\lambda_1}{2}.
\label{we_Lagrang_der_case2}
\end{align}
Since  $\left(1+\displaystyle\frac{\sigma^2}{X_k^2}-\displaystyle\frac{\sigma^4}{X_k^4}+48\displaystyle\frac{\sigma^6}{X_k^6}+360\displaystyle\frac{\sigma^8}{X_k^8} \right)> 1$, \eqref{we_Lagrang_der_case2} contradicts \eqref{dual_feasibility}, which implies that $a_k=1$ is not a solution.  
\item \textbf{Case 3}: $\lambda_1>0$, and $\lambda_2>0$\\
\eqref{complimentary_slackness} and \eqref{primal_feasibility} imply that no solution exists.
 \item \textbf{Case 4}: $\lambda_1=0$, and $\lambda_2=0$\\
In this case, for $a_k \in \left[0,1\right]$, \eqref{complimentary_slackness} is satisfied. Using  \eqref{lagrange_derivative}, one can obtain $a_k$ as follows:
\begin{align}
\nonumber
\frac{\mathrm{d} \mathbf{L}}{\mathrm{d} a_k}&=2a_k\left(X_k+ \frac{\sigma^2}{X_k}-\frac{\sigma^4}{X_k^3}+48\frac{\sigma^6}{X_k^5}+360\frac{\sigma^8}{X_k^7} \right)-2X_k=0,\\ 
&\Rightarrow \,a_k = \left(1+\displaystyle\frac{\sigma^2}{X_k^2}-\displaystyle\frac{\sigma^4}{X_k^4}+48\displaystyle\frac{\sigma^6}{X_k^6}+360\displaystyle\frac{\sigma^8}{X_k^8} \right)^{-1},
\end{align}
 \end{itemize}
which belongs to $[0, 1]$ because the term $-\displaystyle\frac{\sigma^4}{X_k^4}$ is dominated by $\displaystyle\frac{\sigma^2}{X_k^2}$ if $\displaystyle\frac{\sigma^2}{X_k^2} < 1$, and by $48\displaystyle\frac{\sigma^6}{X_k^6}+360\displaystyle\frac{\sigma^8}{X_k^8}$ if $\displaystyle\frac{\sigma^2}{X_k^2} > 1$\\
\indent Considering \textbf{Case 1} to \textbf{Case 4}, we obtain $a_k= \left(1+\displaystyle\frac{\sigma^2}{X_k^2}-\displaystyle\frac{\sigma^4}{X_k^4}+48\displaystyle\frac{\sigma^6}{X_k^6}+360\displaystyle\frac{\sigma^8}{X_k^8} \right)^{-1}$.
One can observe that the sign of $\displaystyle \frac{\mathrm{d}^2 \mathbf{L}}{\mathrm{d} a_k^2}$ is same as that of  $X_k$. Since $X_k$ is positive,  $\displaystyle\frac{\mathrm{d}^2 \mathbf{L}}{\mathrm{d} a_k^2}$ is also positive, hence, $a_k = \left(1+\displaystyle\frac{\sigma^2}{X_k^2}-\displaystyle\frac{\sigma^4}{X_k^4}+48\displaystyle\frac{\sigma^6}{X_k^6}+360\displaystyle\frac{\sigma^8}{X_k^8} \right)^{-1},$ is the unique minimizer of  ${\widehat {\cal R}}$ over $\left[0,1\right]$.\\
\indent Next, we repeat the analysis for the case where $S_k<0$. In this case, the goal is to find an optimum $a_k \in \left[0,1\right]$ that maximizes ${\widehat {\cal R}}$. To obtain the optimum $a_k$, instead of maximizing ${\widehat {\cal R}}$ over $\left[0,1\right]$, one can minimize $-{\widehat {\cal R}}$ over $\left[0,1\right]$. To obtain the solution, we solve \eqref {lagrange_derivative} to \eqref{second_order}, where Lagrangian takes the form
$$\mathbf{L}\left(a_k,\lambda_1,\lambda_2\right) = -a_k^2\left(X_k+ \frac{\sigma^2}{X_k}-\frac{\sigma^4}{X_k^3}+48\frac{\sigma^6}{X_k^5}+360\frac{\sigma^8}{X_k^7} \right)+2a_kX_k-S_k+\lambda_1 \left(a_k-1\right)-\lambda_2 a_k.$$
The first- and second-order derivatives of the Lagrangian are
\begin{align}
&\frac{\mathrm{d} \mathbf{L}}{\mathrm{d} a_k} =-2a_k\left(X_k+ \frac{\sigma^2}{X_k}-\frac{\sigma^4}{X_k^3}+48\frac{\sigma^6}{X_k^5}+360\frac{\sigma^8}{X_k^7} \right)+2X_k+\lambda_1-\lambda_2  \nonumber \text{\,and}\\
&\frac{\mathrm{d}^2 \mathbf{L}}{\mathrm{d} a_k^2} = -2X_k\left(1+\displaystyle\frac{\sigma^2}{X_k^2}-\displaystyle\frac{\sigma^4}{X_k^4}+48\displaystyle\frac{\sigma^6}{X_k^6}+360\displaystyle\frac{\sigma^8}{X_k^8} \right), \text{respectively}\nonumber.
\end{align}
By solving  \eqref{lagrange_derivative}$-$\eqref{second_order}, we obtain,  $a_k = \left(1+\displaystyle\frac{\sigma^2}{X_k^2}-\displaystyle\frac{\sigma^4}{X_k^4}+48\displaystyle\frac{\sigma^6}{X_k^6}+360\displaystyle\frac{\sigma^8}{X_k^8} \right)^{-1}$. Similar to the earlier scenario where $S_k>0$, the sign of $\displaystyle\frac{\mathrm{d}^2 \mathbf{L}}{\mathrm{d} a_k^2}$ depends upon the sign of $X_k$. Since $X_k<0$,
$\displaystyle\frac{\mathrm{d}^2 \mathbf{L}}{\mathrm{d} a_k^2} = -2X_k\left(1+\displaystyle\frac{\sigma^2}{X_k^2}-\displaystyle\frac{\sigma^4}{X_k^4}+48\displaystyle\frac{\sigma^6}{X_k^6}+360\displaystyle\frac{\sigma^8}{X_k^8} \right)>0$, and hence, $a_k$ becomes the unique minimizer of $-\widehat {\cal R}$ over  $\left[0,1\right]$. It is interesting to note that
$$a_k = \left(1+\displaystyle\frac{\sigma^2}{X_k^2}-\displaystyle\frac{\sigma^4}{X_k^4}+48\displaystyle\frac{\sigma^6}{X_k^6}+360\displaystyle\frac{\sigma^8}{X_k^8} \right)^{-1},$$ 
turns out to be the unique minimizer of $\widehat {\cal R}$ when $X_k>0$, and the unique maximizer of $\widehat {\cal R}$ when $X_k<0$.

\subsection{Logarithmic mean-square error (log MSE)}
For the log MSE  distortion, the corresponding risk estimate is  (cf. Eq. (10) in Section IV-C),  
\begin{align*}
{\widehat {\cal R}}=&(\log S_k)^2+(\log a_kX_k)^2-2\log a_kX_k \log X_k
+2\left(\displaystyle\frac{\sigma^2}{X_k^2}-1.5\displaystyle\frac{\sigma^4}{X_k^4}+2.17\displaystyle\frac{\sigma^6}{X_k^6}-159.5\displaystyle\frac{\sigma^8}{X_k^8}\right)\\
&-2 \left(0.5\displaystyle\frac{\sigma^2}{X_k^2}-0.75\displaystyle\frac{\sigma^4}{X_k^4}-10\displaystyle\frac{\sigma^6}{X_k^6}-210\displaystyle\frac{\sigma^8}{X_k^8}\right)\log a_kX_k.
\end{align*}
The corresponding Lagrangian, its first- and second-order derivatives are given by
\begin{align*}
\mathbf{L}=&(\log S_k)^2+(\log a_kX_k)^2-2\log a_kX_k \log X_k
+2\left(\displaystyle\frac{\sigma^2}{X_k^2}-1.5\displaystyle\frac{\sigma^4}{X_k^4}+2.17\displaystyle\frac{\sigma^6}{X_k^6}-159.5\displaystyle\frac{\sigma^8}{X_k^8}\right)\\
&-2 \left(0.5\displaystyle\frac{\sigma^2}{X_k^2}-0.75\displaystyle\frac{\sigma^4}{X_k^4}-10\displaystyle\frac{\sigma^6}{X_k^6}-210\displaystyle\frac{\sigma^8}{X_k^8}\right)\log a_kX_k+\lambda_1 \left(a_k-1\right)-\lambda_2 a_k,
\end{align*}
\begin{align*}
\frac{\mathrm{d} \mathbf{L}}{\mathrm{d} a_k}=&2\frac{\log a_k}{a_k}-2 \frac{1}{a_k} \left(0.5\displaystyle\frac{\sigma^2}{X_k^2}-0.75\displaystyle\frac{\sigma^4}{X_k^4}-10\displaystyle\frac{\sigma^6}{X_k^6}-210\displaystyle\frac{\sigma^8}{X_k^8}\right)+\lambda_1-\lambda_2, \text{\,and}\\ 
\frac{\mathrm{d^2} \mathbf{L}}{\mathrm{d} a_k^2}=&\frac{2}{a_k^2}-2 \frac{\log a_k}{a_k^2} + \frac{2}{a_k^2}\left(0.5\displaystyle\frac{\sigma^2}{X_k^2}-0.75\displaystyle\frac{\sigma^4}{X_k^4}-10\displaystyle\frac{\sigma^6}{X_k^6}-210\displaystyle\frac{\sigma^8}{X_k^8}\right), \text{\,respectively}.
\end{align*}
Let $\beta= \left(0.5\displaystyle\frac{\sigma^2}{X_k^2}-0.75\displaystyle\frac{\sigma^4}{X_k^4}-10\displaystyle\frac{\sigma^6}{X_k^6}-210\displaystyle\frac{\sigma^8}{X_k^8}\right)$.  To solve \eqref {lagrange_derivative}$-$\eqref{second_order}, consider all possibilities for $\lambda_1$ and  $\lambda_2$.
\begin{itemize}
\item \textbf{Case 1}: $\lambda_1=0$, and $\lambda_2>0$\\
\eqref{complimentary_slackness} and   \eqref{primal_feasibility}  $\Rightarrow \, a_k=0$, and from  \eqref{lagrange_derivative}  we obtain
\begin{align}
\nonumber
\frac{\mathrm{d} \mathbf{L}}{\mathrm{d} a_k}\Bigg|_{a_k=0}&=2\frac{\log a_k}{a_k}-2 \frac{\beta}{a_k}-\lambda_2\Bigg|_{a_k=0}=0,\\ 
&\Rightarrow \,  -\infty = \frac{\lambda_2}{2},  \label{logmmse_Lagrang_der_case1}
\end{align}
which contradicts \eqref{dual_feasibility}, and hence $a_k=0$ is not a solution. 
\item \textbf{Case 2}: $\lambda_1>0$, and $\lambda_2=0$\\
\eqref{complimentary_slackness} and  \eqref{primal_feasibility} $\Rightarrow \, a_k=1$, and from \eqref{lagrange_derivative}  we obtain
\begin{align}
\nonumber
\frac{\mathrm{d} \mathbf{L}}{\mathrm{d} a_k}\Bigg|_{a_k=1}&=2\frac{\log a_k}{a_k}-2 \frac{\beta}{a_k}+\lambda_1\Bigg|_{a_k=1}=0,\\ 
&\Rightarrow \,  \beta= \frac{\lambda_1}{2}  \overset{\eqref{dual_feasibility}}{>} 0,\label{logmmse_Lagrang_der_case2}
\end{align}
$\Rightarrow \,$ if $\beta>0$, then $a_k=1$.  
\item \textbf{Case 3}: $\lambda_1>0$, and $\lambda_2>0$\\
\eqref{complimentary_slackness} and \eqref{primal_feasibility} imply that no solution exists.
 \item \textbf{Case 4}: $\lambda_1=0$, and $\lambda_2=0$\\
In this case, for  $a_k \in \left[0,1\right]$, \eqref{complimentary_slackness} is satisfied. Using \eqref{lagrange_derivative}, we obtain $a_k$ as follows:
\begin{align}
\nonumber
\frac{\mathrm{d} \mathbf{L}}{\mathrm{d} a_k}&=2\frac{\log a_k}{a_k}-2 \frac{\beta}{a_k}=0,\\ 
&\Rightarrow \, a_k=\exp\left(\beta\right).
\end{align}
 \end{itemize}
Consolidating \textbf{Case 1} to \textbf{Case 4}, we obtain the optimum $a_k=\min\left\{1,\exp\left(\beta\right)\right\}$. Next, we check the second-order condition  \eqref{second_order} to determine whether the proposed solution is the unique minimizer or not. Consider
\begin{align}
\nonumber
\frac{\mathrm{d^2} \mathbf{L}}{\mathrm{d} a_k^2}&=\displaystyle\frac{2}{a_k^2}-2 \displaystyle\frac{\log a_k}{a_k^2} + \displaystyle\frac{2}{a_k^2}\beta,\\ \nonumber
&=\displaystyle\frac{2}{\left(\min\left\{1,\exp\left(\beta\right)\right\}\right)^2}-2 \displaystyle\frac{\log \min\left\{1,\exp\left(\beta\right)\right\}}{\left(\min\left\{1,\exp\left(\beta\right)\right\}\right)^2} + \displaystyle\frac{2}{\left(\min\left\{1,\exp\left(\beta\right)\right\}\right)^2}\beta,\\ 
&=
\begin{cases}
 2 \exp\left(-2\beta\right) >0, \quad \text{if}\, \exp\left(\beta\right)<1,\\ 
\displaystyle2+2\beta>0, \quad  \text{if} \,  \exp\left(\beta\right)>1.
\label{logmmse_second_der}
\end{cases}
\end{align}
From~\eqref{logmmse_second_der}, we observe that $$a_k=\min\left\{1,\exp\left(0.5\displaystyle\frac{\sigma^2}{X_k^2}-0.75\displaystyle\frac{\sigma^4}{X_k^4}-10\displaystyle\frac{\sigma^6}{X_k^6}-210\displaystyle\frac{\sigma^8}{X_k^8}\right)\right\}$$ is the unique minimizer of $\widehat{\mathcal{R}}$ over $\left[0,1\right]$. 
\subsection{Itakura-Saito (IS) distortion}
In this case, the risk estimate is (cf. Eq. (12) in Section IV-D),
\begin{equation}
{\widehat {\cal R}}=a_k\left(1+60\displaystyle\frac{\sigma^6}{X_k^6}+840\displaystyle\frac{\sigma^8}{X_k^8}\right)-\log a_kX_k+\log\,S_k-1.\nonumber
\end{equation}
The corresponding Lagrangian,  its first- and second-order derivatives are  
\begin{align*}
\mathbf{L}=&a_k\left(1+60\displaystyle\frac{\sigma^6}{X_k^6}+840\displaystyle\frac{\sigma^8}{X_k^8}\right)-\log a_kX_k+\log\,S_k-1+\lambda_1 \left(a_k-1\right)-\lambda_2 a_k,\\
\frac{\mathrm{d} \mathbf{L}}{\mathrm{d} a_k}=&\left(1+60\displaystyle\frac{\sigma^6}{X_k^6}+840\displaystyle\frac{\sigma^8}{X_k^8}\right)-\frac{1}{a_k}+\lambda_1-\lambda_2, \text{\, and}\\ 
\frac{\mathrm{d^2} \mathbf{L}}{\mathrm{d} a_k^2}=&\displaystyle\frac{1}{a_k^2}, \text{respectively}.
\end{align*}
\begin{itemize}
 \item \textbf{Case 1}: $\lambda_1=0$, and $\lambda_2>0$\\
\eqref{complimentary_slackness} and   \eqref{primal_feasibility}  $\Rightarrow \, a_k=0$, and from  \eqref{lagrange_derivative}  we obtain
\begin{align}
\nonumber
\frac{\mathrm{d} \mathbf{L}}{\mathrm{d} a_k}\Bigg|_{a_k=0}&=\left(1+60\displaystyle\frac{\sigma^6}{X_k^6}+840\displaystyle\frac{\sigma^8}{X_k^8}\right)-\frac{1}{a_k}-\lambda_2,\Bigg|_{a_k=0}=0,\\ 
&\Rightarrow \,  -\infty = \lambda_2 \label{IS_Lagrang_der_case1}, 
\end{align}
which contradicts \eqref{dual_feasibility}, and hence $a_k=0$ is not a solution. 
\item \textbf{Case 2}: $\lambda_1>0$, and $\lambda_2=0$\\
\eqref{complimentary_slackness} and  \eqref{primal_feasibility} $\Rightarrow \, a_k=1$, and from \eqref{lagrange_derivative}  we obtain
\begin{align}
\nonumber
\frac{\mathrm{d} \mathbf{L}}{\mathrm{d} a_k}\Bigg|_{a_k=1}&=\left(1+60\displaystyle\frac{\sigma^6}{X_k^6}+840\displaystyle\frac{\sigma^8}{X_k^8}\right)-1+\lambda_1,\\ 
&\Rightarrow \, \underbrace{1-\left(1+60\displaystyle\frac{\sigma^6}{X_k^6}+840\displaystyle\frac{\sigma^8}{X_k^8}\right)}_{<0}= \lambda_1, \label{IS_Lagrang_der_case2}
\end{align}
which contradicts \eqref{dual_feasibility}, and hence  $a_k=1$ is not a solution.  
\item \textbf{Case 3}: $\lambda_1>0$, and $\lambda_2>0$\\
In this case, \eqref{complimentary_slackness} and \eqref{primal_feasibility} imply that no solution exists.
 \item \textbf{Case 4}: $\lambda_1=0$, and $\lambda_2=0$\\
In this case, for $a_k \in \left[0,1\right]$, \eqref{complimentary_slackness} is satisfied. One can obtain $a_k$ using  \eqref{lagrange_derivative}, as follows
\begin{align}
\nonumber
\frac{\mathrm{d} \mathbf{L}}{\mathrm{d} a_k}&=\left(1+60\displaystyle\frac{\sigma^6}{X_k^6}+840\displaystyle\frac{\sigma^8}{X_k^8}\right)-\frac{1}{a_k}=0,\\ 
&\Rightarrow \, a_k=\left(1+60\displaystyle\frac{\sigma^6}{X_k^6}+840\displaystyle\frac{\sigma^8}{X_k^8}\right)^{-1}.
\end{align}
 \end{itemize}
One can observe that $\displaystyle\frac{\mathrm{d^2} \mathbf{L}}{\mathrm{d} a_k^2}>0$ implies that $a_k=\left(1+60\displaystyle\frac{\sigma^6}{X_k^6}+840\displaystyle\frac{\sigma^8}{X_k^8}\right)^{-1}$ is the unique minimizer of the function $\widehat{\mathcal{R}}$ over $\left[0,1\right]$.
\subsection{Itakura-Saito (IS) - II distortion}
In this case, the risk estimate ${\widehat {\cal R}}$ (cf. Eq. (13) in Section IV-E), the corresponding Lagrangian, its first- and second-order derivatives turn out to be
\begin{align*}
{\widehat {\cal R}}=&a_k^2\left(1+\displaystyle\frac{\sigma^2}{X_k^2}-3\displaystyle\frac{\sigma^4}{X_k^4}+360\displaystyle\frac{\sigma^6}{X_k^6}+4200\displaystyle\frac{\sigma^8}{X_k^8}\right)-\log a_k^2X_k^2+\log S_k^2-1,\\
\mathbf{L}=a_k^2&\left(1+\displaystyle\frac{\sigma^2}{X_k^2}-3\displaystyle\frac{\sigma^4}{X_k^4}+360\displaystyle\frac{\sigma^6}{X_k^6}+4200\displaystyle\frac{\sigma^8}{X_k^8}\right)-\log a_k^2X_k^2+\log S_k^2-1+\lambda_1 \left(a_k-1\right)-\lambda_2 a_k,\\
\frac{\mathrm{d} \mathbf{L}}{\mathrm{d} a_k}=&2a_k\left(1+\displaystyle\frac{\sigma^2}{X_k^2}-3\displaystyle\frac{\sigma^4}{X_k^4}+360\displaystyle\frac{\sigma^6}{X_k^6}+4200\displaystyle\frac{\sigma^8}{X_k^8}\right)-\displaystyle\frac{2}{a_k}+\lambda_1-\lambda_2,   \text{\,and}\\ 
\frac{\mathrm{d^2} \mathbf{L}}{\mathrm{d} a_k^2}=&2\left(1+\displaystyle\frac{\sigma^2}{X_k^2}-3\displaystyle\frac{\sigma^4}{X_k^4}+360\displaystyle\frac{\sigma^6}{X_k^6}+4200\displaystyle\frac{\sigma^8}{X_k^8}\right)+\displaystyle\frac{2}{a_k^2}, \text{\,respectively}.
\end{align*}
Consider the four cases below:
\begin{itemize}
 \item \textbf{Case 1}: $\lambda_1=0$, and $\lambda_2>0$\\
\eqref{complimentary_slackness} and   \eqref{primal_feasibility}  $\Rightarrow \, a_k=0$, and from  \eqref{lagrange_derivative}  we obtain
\begin{align}
\nonumber
\frac{\mathrm{d} \mathbf{L}}{\mathrm{d} a_k}\Bigg|_{a_k=0}&=2a_k\left(1+\displaystyle\frac{\sigma^2}{X_k^2}-3\displaystyle\frac{\sigma^4}{X_k^4}+360\displaystyle\frac{\sigma^6}{X_k^6}+4200\displaystyle\frac{\sigma^8}{X_k^8}\right)-\displaystyle\frac{2}{a_k}-\lambda_2 ,\Bigg|_{a_k=0}=0,\\ 
&\Rightarrow \,  -\infty = \frac{\lambda_2}{2},\label{ISII_Lagrang_der_case1}
\end{align}
which contradicts \eqref{dual_feasibility}, and hence $a_k=0$ is not a solution. 
\item \textbf{Case 2}: $\lambda_1>0$, and $\lambda_2=0$\\
\eqref{complimentary_slackness} and  \eqref{primal_feasibility} $\Rightarrow \, a_k=1$, and from \eqref{lagrange_derivative}  we obtain
\begin{align}
\nonumber
\frac{\mathrm{d} \mathbf{L}}{\mathrm{d} a_k}\Bigg|_{a_k=1}&=2\left(1+\displaystyle\frac{\sigma^2}{X_k^2}-3\displaystyle\frac{\sigma^4}{X_k^4}+360\displaystyle\frac{\sigma^6}{X_k^6}+4200\displaystyle\frac{\sigma^8}{X_k^8}\right)-2+\lambda_1,\\ 
&\Rightarrow \,  1-\left(1+\displaystyle\frac{\sigma^2}{X_k^2}-3\displaystyle\frac{\sigma^4}{X_k^4}+360\displaystyle\frac{\sigma^6}{X_k^6}+4200\displaystyle\frac{\sigma^8}{X_k^8}\right)= \frac{\lambda_1}{2}  \overset{\eqref{dual_feasibility}}{>} 0,\label{ISII_Lagrang_der_case2}
\end{align}
$\Rightarrow \, $if  $1-\left(1+\displaystyle\frac{\sigma^2}{X_k^2}-3\displaystyle\frac{\sigma^4}{X_k^4}+360\displaystyle\frac{\sigma^6}{X_k^6}+4200\displaystyle\frac{\sigma^8}{X_k^8}\right)>0$,  then  $a_k=1$.  
\item \textbf{Case 3}: $\lambda_1>0$, and $\lambda_2>0$\\
\eqref{complimentary_slackness} and \eqref{primal_feasibility} imply that no solution exists.
 \item \textbf{Case 4}: $\lambda_1=0$, and $\lambda_2=0$\\
In this case, for $a_k \in \left[0,1\right]$, \eqref{complimentary_slackness} is satisfied. Using  \eqref{lagrange_derivative},  one can obtain $a_k$ as follows:
\begin{align}
\nonumber
\frac{\mathrm{d} \mathbf{L}}{\mathrm{d} a_k}&=2a_k\left(1+\displaystyle\frac{\sigma^2}{X_k^2}-3\displaystyle\frac{\sigma^4}{X_k^4}+360\displaystyle\frac{\sigma^6}{X_k^6}+4200\displaystyle\frac{\sigma^8}{X_k^8}\right)-\displaystyle\frac{2}{a_k} =0,\\ 
&\Rightarrow \, a_k=\left(1+\displaystyle\frac{\sigma^2}{X_k^2}-3\displaystyle\frac{\sigma^4}{X_k^4}+360\displaystyle\frac{\sigma^6}{X_k^6}+4200\displaystyle\frac{\sigma^8}{X_k^8}\right)^{-\frac{1}{2}}.
\end{align}
 \end{itemize}
From  \textbf{Case 1} to \textbf{Case 4}, we obtain the optimum shrinkage as
 $ a_k=\min\left\{1, \left(1+\displaystyle\frac{\sigma^2}{X_k^2}-3\displaystyle\frac{\sigma^4}{X_k^4}+360\displaystyle\frac{\sigma^6}{X_k^6}+4200\displaystyle\frac{\sigma^8}{X_k^8}\right)^{-\frac{1}{2}} \right\}$. Since $\displaystyle \frac{\mathrm{d^2} \mathbf{L} \left(a_k,\lambda_1,\lambda_2\right)}{\mathrm{d} a_k^2}>0$,  $a_k$ turns out to be  the unique minimizer of ${\widehat {\cal R}}$ over $\left[0,1\right]$. 
 \vspace{-0.25cm}
\subsection{Hyperbolic cosine distortion measure (COSH)}
The COSH risk estimate is (cf. Eq. (14) in Section IV-F)
\begin{align*}
{\widehat {\cal R}}=\displaystyle\frac{1}{2}\left(\displaystyle\frac{1}{a_k}+\displaystyle\frac{\sigma^2}{a_kX_k^2}+a_k\left(1+60\displaystyle\frac{\sigma^6}{X_k^6}+840\displaystyle\frac{\sigma^8}{X_k^8}\right)\right)-1.
\end{align*}
The corresponding Lagrangian,  its first- and second-order derivatives are
\begin{align*}
\mathbf{L}&=\displaystyle\frac{1}{2}\left(\displaystyle\frac{1}{a_k}+\displaystyle\frac{\sigma^2}{a_kX_k^2}+a_k\left(1+60\displaystyle\frac{\sigma^6}{X_k^6}+840\displaystyle\frac{\sigma^8}{X_k^8}\right)\right)-1+\lambda_1 \left(a_k-1\right)-\lambda_2 a_k,\\
\frac{\mathrm{d} \mathbf{L}}{\mathrm{d} a_k}&=\displaystyle\frac{1}{2}\left(\displaystyle-\frac{1}{a_k^2}-\displaystyle\frac{\sigma^2}{a_k^2X_k^2}+\left(1+60\displaystyle\frac{\sigma^6}{X_k^6}+840\displaystyle\frac{\sigma^8}{X_k^8}\right)\right)+\lambda_1-\lambda_2,  \text{\,and} \\ 
\frac{\mathrm{d^2} \mathbf{L}}{\mathrm{d} a_k^2}&=\displaystyle \left(\displaystyle \frac{1}{a_k^3}+\displaystyle\frac{\sigma^2}{a_k^3X_k^2}\right).
\end{align*}
Consider four cases of  $\lambda_1$ and $\lambda_2$.
\begin{itemize}
 \item \textbf{Case 1}: $\lambda_1=0$, and $\lambda_2>0$\\
\eqref{complimentary_slackness} and   \eqref{primal_feasibility}  $\Rightarrow \, a_k=0$, and from  \eqref{lagrange_derivative}  we obtain
\begin{align}
\nonumber
\frac{\mathrm{d} \mathbf{L}}{\mathrm{d} a_k}\Bigg|_{a_k=0}&=\displaystyle\frac{1}{2}\left(\displaystyle-\frac{1}{a_k^2}-\displaystyle\frac{\sigma^2}{a_k^2X_k^2}+\left(1+60\displaystyle\frac{\sigma^6}{X_k^6}+840\displaystyle\frac{\sigma^8}{X_k^8}\right)\right)-\lambda_2\Bigg|_{a_k=0}=0,\\ 
&\Rightarrow \,  -\infty = \lambda_2,\label{COSH_Lagrang_der_case1}
\end{align}
which contradicts \eqref{dual_feasibility}, and hence $a_k=0$ is not a solution. 
\item \textbf{Case 2}: $\lambda_1>0$, and $\lambda_2=0$\\
\eqref{complimentary_slackness} and  \eqref{primal_feasibility} $\Rightarrow \, a_k=1$, and from \eqref{lagrange_derivative}, we obtain
\begin{align}
\nonumber
\frac{\mathrm{d} \mathbf{L}}{\mathrm{d} a_k}\Bigg|_{a_k=1}&=\displaystyle\frac{1}{2}\left(\displaystyle-\left(1+\displaystyle\frac{\sigma^2}{X_k^2}\right)+\left(1+60\displaystyle\frac{\sigma^6}{X_k^6}+840\displaystyle\frac{\sigma^8}{X_k^8}\right)\right)+\lambda_1=0,\\ 
&\Rightarrow \,  \left(1+\displaystyle\frac{\sigma^2}{X_k^2}\right)-\left(1+60\displaystyle\frac{\sigma^6}{X_k^6}+840\displaystyle\frac{\sigma^8}{X_k^8}\right)= 2 \lambda_1  \overset{\eqref{dual_feasibility}}{>} 0,\label{COSH_Lagrang_der_case2}
\end{align}
$\Rightarrow \,$ if  $\left(1+60\displaystyle\frac{\sigma^6}{X_k^6}+840\displaystyle\frac{\sigma^8}{X_k^8}\right)<\left(1+\displaystyle\frac{\sigma^2}{X_k^2}\right)$, then  $a_k=1$ is  a solution.  
\item \textbf{Case 3}: $\lambda_1>0$, and $\lambda_2>0$\\
\eqref{complimentary_slackness} and \eqref{primal_feasibility} imply that no solution exists.
 \item \textbf{Case 4}: $\lambda_1=0$, and $\lambda_2=0$\\
In this case, for  $a_k \in \left[0,1\right]$, \eqref{complimentary_slackness} is satisfied. Using  \eqref{lagrange_derivative} one can obtain $a_k$ as follows:
\begin{align}
\nonumber
\frac{\mathrm{d} \mathbf{L}}{\mathrm{d} a_k}&=\displaystyle\frac{1}{2}\left(\displaystyle-\frac{1}{a_k^2}\left(1+\displaystyle\frac{\sigma^2}{X_k^2}\right)+\left(1+60\displaystyle\frac{\sigma^6}{X_k^6}+840\displaystyle\frac{\sigma^8}{X_k^8}\right)\right)=0,\\ 
&\Rightarrow \, a_k=\left( \frac{1+\displaystyle\frac{\sigma^2}{X_k^2}}{1+60\displaystyle\frac{\sigma^6}{X_k^6}+840\displaystyle\frac{\sigma^8}{X_k^8}}\right)^\frac{1}{2}.
\end{align}
 \end{itemize}
From \textbf{Case 1} to \textbf{Case 4}, we obtain
$$ a_k=\min\left\{1, \left( \frac{1+\displaystyle\frac{\sigma^2}{X_k^2}}{1+60\displaystyle\frac{\sigma^6}{X_k^6}+840\displaystyle\frac{\sigma^8}{X_k^8}}\right)^\frac{1}{2}\right\}.$$
Since $\displaystyle \frac{\mathrm{d^2} \mathbf{L} \left(a_k,\lambda_1,\lambda_2\right)}{\mathrm{d} a_k^2}>0$,  $a_k$ turns out to be  the unique minimizer of ${\widehat {\cal R}}$ over $\left[0,1\right]$. 
\subsection{Weighted cosine distortion measure (WCOSH)}
Here, we consider ${\widehat {\cal R}}$ (cf. Eq. (23) in Appendix F) defined as follows:
\begin{align*}
{\widehat {\cal R}}= \frac{1}{2}  \frac{a_k}{X_k} \left( 1-\frac{\sigma^2}{X_k^2}+3 \frac{\sigma^4}{X_k^4}+420\frac{\sigma^6}{X_k^6}+8400\frac{\sigma^8}{X_k^8} \right)+ \frac{1}{2a_kX_k} -\frac{1}{S_k}.
\end{align*}
The goal is to find $a_k$ that minimizes  ${\widehat {\cal R}}$ when $S_k>0$, and maximizes it when $S_k<0$. Under the high SNR assumption, we assume that $S_k$ and $X_k$ have the same sign. First, we consider the case where $S_k>0$. To obtain the optimum  $a_k$,  we solve \eqref {lagrange_derivative}$-$\eqref{second_order} where
\begin{align*}
&\mathbf{L}=\frac{1}{2}  \frac{a_k}{X_k} \left( 1-\frac{\sigma^2}{X_k^2}+3 \frac{\sigma^4}{X_k^4}+420\frac{\sigma^6}{X_k^6}+8400\frac{\sigma^8}{X_k^8} \right)+ \frac{1}{2a_kX_k} -\frac{1}{S_k}+\lambda_1 \left(a_k-1\right)-\lambda_2 a_k,\\
&\frac{\mathrm{d} \mathbf{L}}{\mathrm{d} a_k}= \frac{1}{2X_k} \left( 1-\frac{\sigma^2}{X_k^2}+3 \frac{\sigma^4}{X_k^4}+420\frac{\sigma^6}{X_k^6}+8400\frac{\sigma^8}{X_k^8} \right)-\frac{1}{2a_k^2X_k} +\lambda_1-\lambda_2, \text{\,and} \\ 
&\frac{\mathrm{d^2} \mathbf{L}}{\mathrm{d} a_k^2}=\displaystyle\frac{1}{a_k^3X_k}.
\end{align*}
\begin{itemize}
 \item \textbf{Case 1}: $\lambda_1=0$, and $\lambda_2>0$\\
\eqref{complimentary_slackness} and   \eqref{primal_feasibility}  $\Rightarrow \, a_k=0$, and from  \eqref{lagrange_derivative}  we obtain
\begin{align}
\nonumber
\frac{\mathrm{d} \mathbf{L}}{\mathrm{d} a_k}\Bigg|_{a_k=0}&=\frac{1}{2X_k} \left( 1-\frac{\sigma^2}{X_k^2}+3 \frac{\sigma^4}{X_k^4}+420\frac{\sigma^6}{X_k^6}+8400\frac{\sigma^8}{X_k^8} \right)-\frac{1}{2a_k^2X_k}-\lambda_2\Bigg|_{a_k=0}=0,\\ 
&\Rightarrow \,  -\infty = \lambda_2,\label{WCOSH_Lagrang_der_case1}
\end{align}
which  contradicts \eqref{dual_feasibility}, and hence $a_k=0$ is not a solution. 
\item \textbf{Case 2}: $\lambda_1>0$, and $\lambda_2=0$\\
\eqref{complimentary_slackness} and  \eqref{primal_feasibility} $\Rightarrow \, a_k=1$, and from \eqref{lagrange_derivative}  we obtain
\begin{align*}
\nonumber
\frac{\mathrm{d} \mathbf{L}}{\mathrm{d} a_k}\Bigg|_{a_k=1}&=\frac{1}{2X_k} \left( 1-\frac{\sigma^2}{X_k^2}+3 \frac{\sigma^4}{X_k^4}+420\frac{\sigma^6}{X_k^6}+8400\frac{\sigma^8}{X_k^8} \right)-\frac{1}{2X_k}+\lambda_1=0,\\ 
&\Rightarrow \,  \frac{1}{X_k}-\frac{1}{X_k} \left( 1-\frac{\sigma^2}{X_k^2}+3 \frac{\sigma^4}{X_k^4}+420\frac{\sigma^6}{X_k^6}+8400\frac{\sigma^8}{X_k^8} \right)= 2 \lambda_1  \overset{\eqref{dual_feasibility}}{>} 0, \label{WCOSH_Lagrang_der_case2}
\end{align*}
$\Rightarrow \,$ if $1>\left(  \displaystyle 1-\frac{\sigma^2}{X_k^2}+3 \frac{\sigma^4}{X_k^4}+420\frac{\sigma^6}{X_k^6}+8400\frac{\sigma^8}{X_k^8} \right)$, then  $a_k=1$.  
\item \textbf{Case 3}: $\lambda_1>0$, and $\lambda_2>0$\\
\eqref{complimentary_slackness} and \eqref{primal_feasibility} imply that no solution exists.
\item \textbf{Case 4}: $\lambda_1=0$, and $\lambda_2=0$\\
In this case, for  $a_k \in \left[0,1\right]$, \eqref{complimentary_slackness} is satisfied. Using \eqref{lagrange_derivative}, we obtain $a_k$  as follows:
\begin{align*}
\nonumber
\frac{\mathrm{d} \mathbf{L}}{\mathrm{d} a_k}&= \frac{1}{2X_k} \left( 1-\frac{\sigma^2}{X_k^2}+3 \frac{\sigma^4}{X_k^4}+420\frac{\sigma^6}{X_k^6}+8400\frac{\sigma^8}{X_k^8} \right)-\frac{1}{2a_k^2X_k} =0,\\ 
&\Rightarrow \, a_k= \left( 1-\frac{\sigma^2}{X_k^2}+3 \frac{\sigma^4}{X_k^4}+420\frac{\sigma^6}{X_k^6}+8400\frac{\sigma^8}{X_k^8} \right)^{-\frac{1}{2}}.
\end{align*}
 \end{itemize}
From  \textbf{Case 1} to \textbf{Case 4}, we obtain that 
$$a_k=\min \left\{1,\displaystyle  \left( 1-\frac{\sigma^2}{X_k^2}+3 \frac{\sigma^4}{X_k^4}+420\frac{\sigma^6}{X_k^6}+8400\frac{\sigma^8}{X_k^8} \right)^{-\frac{1}{2}}\right\}.$$ 
We observe that,  
$$\displaystyle \frac{\mathrm{d^2} \mathbf{L}}{\mathrm{d} a_k^2}=\displaystyle\frac{1}{(a_k)^3X_k}>0,$$ which indicates that the solution obtained is the unique minimizer.\\
\indent Next, consider the scenario $S_k<0$, where the goal is to obtain  $a_k\in \left[0,1\right]$, which maximizes ${\widehat {\cal R}}$, or equivalently minimizes $-{\widehat {\cal R}}$. We are required to solve \eqref {lagrange_derivative}$-$\eqref{second_order}, where
\begin{align*}
&\mathbf{L}=-\frac{1}{2}  \frac{a_k}{X_k} \left( 1-\frac{\sigma^2}{X_k^2}+3 \frac{\sigma^4}{X_k^4}+420\frac{\sigma^6}{X_k^6}+8400\frac{\sigma^8}{X_k^8} \right)- \frac{1}{2a_kX_k} +\frac{1}{S_k}+\lambda_1 \left(a_k-1\right)-\lambda_2 a_k,\\
&\frac{\mathrm{d} \mathbf{L}}{\mathrm{d} a_k}= -\frac{1}{2X_k} \left( 1-\frac{\sigma^2}{X_k^2}+3 \frac{\sigma^4}{X_k^4}+420\frac{\sigma^6}{X_k^6}+8400\frac{\sigma^8}{X_k^8} \right)+\frac{1}{2a_k^2X_k} +\lambda_1-\lambda_2, \text{\,and} \\ 
&\frac{\mathrm{d^2} \mathbf{L}}{\mathrm{d} a_k^2}=-\displaystyle\frac{1}{a_k^3X_k}.
\end{align*}
Solving  \eqref {lagrange_derivative}$-$\eqref{second_order} yields  
$$a_k=\min \left\{1,\displaystyle  \left( 1-\frac{\sigma^2}{X_k^2}+3 \frac{\sigma^4}{X_k^4}+420\frac{\sigma^6}{X_k^6}+8400\frac{\sigma^8}{X_k^8} \right)^{-\frac{1}{2}}\right\}.$$ 
Since $X_k<0$, 
$$\displaystyle \frac{\mathrm{d^2} \mathbf{L}}{\mathrm{d} a_k^2}=-\displaystyle\frac{1}{a_k^3X_k} > 0,$$ 
implying that $a_k$ is the unique minimizer of $-{\widehat {\cal R}}$ over $a_k \in \left[0,1\right]$. It is worth mentioning that the $a_k$ obtained is the unique minimizer when $S_k>0$ and the unique maximizer when $S_k<0$.\\

\end{document}